\documentclass[11pt,reqno]{amsart}
\usepackage{amsmath,amsthm,amssymb,bm,bbm,dsfont,braket}
\usepackage[mathscr]{eucal}
\usepackage{amsaddr}
\usepackage{upgreek}
\usepackage[left=2cm,right=2cm,top=2cm,bottom=2cm]{geometry}
\usepackage{mathtools}\mathtoolsset{centercolon}
\usepackage{cite}

\usepackage[shortlabels]{enumitem}
\usepackage{soul}
\usepackage{setspace}
\usepackage{graphicx}
\usepackage{verbatim}
\usepackage{float}
\usepackage{placeins}
\usepackage{array}
\usepackage{booktabs}
\usepackage{threeparttable}
\usepackage[update,prepend]{epstopdf}
\usepackage{multirow}
\usepackage{amsfonts,amssymb,dsfont,xfrac,comment}
\usepackage[abs]{overpic}

\usepackage[usenames,dvipsnames]{color}
\usepackage[hidelinks,linktocpage=true]{hyperref}
\hypersetup{
    linktoc=all,     %set to all if you want both sections and subsections linked
	unicode=false,          % non-Latin characters in Acrobat’s bookmarks
	pdftoolbar=true,        % show Acrobat’s toolbar?
	pdfmenubar=true,        % show Acrobat’s menu?
	pdffitwindow=false,     % window fit to page when opened
	pdfstartview={FitH},    % fits the width of the page to the window
	pdfnewwindow=true,      % links in new PDF window
	colorlinks=true,        % false: boxed links; true: colored links
	linkcolor=BrickRed,          % color of internal links (change box color with linkbordercolor)
	citecolor=OliveGreen,  % color of links to bibliography
	filecolor=Magenta,      % color of file links
	urlcolor=BlueViolet,    % color of external links
}
\usepackage{doi}
\usepackage{url}
\usepackage{caption, subcaption}
\usepackage{enumitem}
\usepackage{shadethm}
\usepackage{tikz,pgf}

\makeatletter
\def\l@subsection{\@tocline{2}{0pt}{2.5pc}{5pc}{}}
\renewcommand\tocchapter[3]{%
  \indentlabel{\@ifnotempty{#2}{\ignorespaces#2.\quad}}#3%
}
\newcommand\@dotsep{4.5}
\def\@tocline#1#2#3#4#5#6#7{\relax
  \ifnum #1>\c@tocdepth % then omit
  \else
    \par \addpenalty\@secpenalty\addvspace{#2}%
    \begingroup \hyphenpenalty\@M
    \@ifempty{#4}{%
      \@tempdima\csname r@tocindent\number#1\endcsname\relax
    }{%
      \@tempdima#4\relax
    }%
    \parindent\z@ \leftskip#3\relax \advance\leftskip\@tempdima\relax
    \rightskip\@pnumwidth plus1em \parfillskip-\@pnumwidth
    #5\leavevmode\hskip-\@tempdima{#6}\nobreak
    \leaders\hbox{$\m@th\mkern \@dotsep mu\hbox{.}\mkern \@dotsep mu$}\hfill
    \nobreak
    \hbox to\@pnumwidth{\@tocpagenum{#7}}\par
    \nobreak
    \endgroup
  \fi}
\makeatother
\AtBeginDocument{%
\makeatletter
\expandafter\renewcommand\csname r@tocindent0\endcsname{0pt}
\makeatother
}
\def\l@subsection{\@tocline{2}{0pt}{2.5pc}{5pc}{}}

\makeatletter
\ifx\@NODS\undefined%

\let\mathbb=\mathds
\else%
\fi
\makeatother

\DeclareMathOperator{\Tr}{Tr}

\newcommand{\si}{{\sigma}}

\newcommand{\be}{{\mathbf e}}

\newcommand{\al}{{\alpha}}

\newcommand{\ten}{\otimes}
\newcommand{\pl}{\hspace{.1cm}}

\newcommand{\norm}[2]{\left\lVert{{#1}}\right\lVert_{#2}}

\newcommand*{\cE}{\mathcal{E}}
\newcommand*{\cF}{\mathcal{F}}

\newcommand*{\cI}{\mathcal{I}}

\newcommand*{\id}{\mathrm{id}}

\newcommand*{\eps}{\varepsilon}

\newcommand{\lan}{\langle}
\newcommand{\ran}{\rangle}

\def\be{\begin{equation}}
\def\ee{\end{equation}}
\def\bea{\begin{eqnarray}}
\def\eea{\end{eqnarray}}

\def\eps{\varepsilon}

\renewcommand{\o}[1]{\mathring{#1}}
\theoremstyle{plain}

\newshadetheorem{theorem}{Theorem}[section]%[chapter]
\newshadetheorem{shaded_theorem}{Theorem}[section]
\newshadetheorem{conjecture}{Conjecture} %[section]
\newshadetheorem{corollary}[shaded_theorem]{Corollary}
\newshadetheorem{proposition}[shaded_theorem]{Proposition}
\newshadetheorem{lemma}[shaded_theorem]{Lemma}
\newtheorem{fact}[shaded_theorem]{Fact}%[chapter]
\newshadetheorem{question}{Question}
\newtheorem{definition}[shaded_theorem]{Definition}
\newtheorem*{lemma_dimension_bound}{Lemma~\ref{lemma:dimension_bound}}
\newtheorem*{lemma_convexity}{Lemma~\ref{lemma:convexity}}
\newtheorem*{fact_exponent}{Proposition~\ref{fact:properties_exponent}}

\theoremstyle{remark}
\newtheorem{remark}[shaded_theorem]{Remark}%[section]

\numberwithin{equation}{section}

\makeatletter
\newcommand{\opnorm}{\@ifstar\@opnorms\@opnorm}
\newcommand{\@opnorms}[1]{%
	$\left|\mkern-1.5mu\left|\mkern-1.5mu\left|
	#1
	\right|\mkern-1.5mu\right|\mkern-1.5mu\right|$
}
\newcommand{\@opnorm}[2][]{%
	\mathopen{#1|\mkern-1.5mu#1|\mkern-1.5mu#1|}
	#2
	\mathclose{#1|\mkern-1.5mu#1|\mkern-1.5mu#1|}
}
\makeatother
\renewcommand{\bm}{{\bf m}}
\begin{document}

%%%%% How to disable amsart.cls to capitalize the article title?
\let\origmaketitle\maketitle
\def\maketitle{
	\begingroup
	\def\uppercasenonmath##1{} % this disables uppercasing title
	\let\MakeUppercase\relax % this disables uppercasing authors
	\origmaketitle
	\endgroup
}
%%%%%%%%%%%%%%%%%%%%%%%%%%%%%%%%

\title{\bfseries \Large{
Quantum Broadcast Channel Simulation \\via Multipartite Convex Splitting
}}

\author{ \normalsize {Hao-Chung Cheng$^{1\text{--}5}$, Li Gao$^6$, and Mario Berta$^7$}}
\address{\small  	
$^1$Department of Electrical Engineering and Graduate Institute of Communication Engineering,\\ National Taiwan University, Taipei 106, Taiwan (R.O.C.)\\
% $^1$Department of Electrical Engineering, Department of Mathematics, National Taiwan University\\
$^2$Department of Mathematics, National Taiwan University\\
$^3$Center for Quantum Science and Engineering,  National Taiwan University\\
$^4$Physics Division, National Center for Theoretical Sciences, Taipei 10617, Taiwan (R.O.C.)\\
$^5$Hon Hai (Foxconn) Quantum Computing Center, New Taipei City 236, Taiwan (R.O.C.)\\
$^6$Department of Mathematics, University of Houston, Houston, TX 77204, USA\\
$^7$Institute for Quantum Information, RWTH Aachen University, Aachen, Germany
}

\email{\href{mailto:haochung.ch@gmail.com}{haochung.ch@gmail.com}}
 \email{\href{mailto:gaolimath@gmail.com}{lgao12@uh.edu}}
  \email{\href{mailto:berta@physik.rwth-aachen.de}{berta@physik.rwth-aachen.de}}

\date{\today}
\setlength\parindent{+4ex}

\begin{abstract}
We show that the communication cost of quantum broadcast channel simulation under free entanglement assistance between the sender and the receivers is asymptotically characterized by an efficiently computable single-letter formula in terms of the channel's multipartite mutual information. Our core contribution is a new one-shot achievability result for multipartite quantum state splitting via multipartite convex splitting. As part of this, we face a general instance of the quantum joint typicality problem with arbitrarily overlapping marginals. The crucial technical ingredient to sidestep this difficulty is a conceptually novel multipartite mean-zero decomposition lemma, together with employing recently introduced complex interpolation techniques for sandwiched R\'enyi divergences.

\textbf{}Moreover, we establish an exponential convergence of the simulation error when the communication costs are within the interior of the capacity region.
As the costs approach the boundary of the capacity region moderately quickly, we show that the error still vanishes asymptotically.

%to overcome the difficulty
%In order to lift the corresponding classical results [Cao {\it et al.}, arXiv:2212.11666] to the fully quantum setting
%We prove the single-letter characterization of the sought-after capacity region for simulating general quantum broadcast channels. The simulation cost of quantum broadcast channels is expressed in terms of the multipartite quantum mutual information of channel, which generalizes the well-known quantum reverse Shannon theorem for point-to-point quantum channels.
\end{abstract}

\maketitle

{
  \hypersetup{linkcolor=black}
  \tableofcontents
}

% \tableofcontents

%%%%%%%%%%%%%%%%%%%%%%%%%%%%%%%%

\newpage
\section{Introduction} \label{sec:introduction}

%% Channel simulation
\emph{Quantum channel simulation} endeavors to simulate a noisy quantum channel via a limited resource of noiseless channels and potentially other assistance such as, e.g., free entanglement. For simulating a point-to-point channel $\mathscr{N}_{A\to B}$ in the independent and identically distributed (i.i.d.) setting, Bennett \textit{et al.}~showed that under free entanglement-assistance the amount of asymptotic classical communication rate needed is equal to the quantum mutual information of channel \cite{BSS+02,BDH+14,BCR11}
\begin{align}
    I(\mathscr{N}_{A\to B}) := \sup_{\rho} I(B:R)_{(\mathscr{N}\otimes \id)\left(\psi^\rho\right)},
\end{align}
where the maximum ranges over all input states $\rho_{A}$, $\psi_{AR}^\rho$ is a purification of $\rho_A$ by a reference system $R$, and $I$ denotes the quantum mutual information. Notably, the aforementioned classical communication cost coincides with the entanglement-assisted classical channel capacity \cite{BSS+99, BSS+02, Hol02}. Hence, from both operational and information-theoretic perspectives, point-to-point quantum channel simulation may be viewed as a \emph{reverse} task of the quantum version of Shannon's celebrated noisy coding theorem for simulating a noiseless channel via a noisy channel \cite{Sha48}. Accordingly, such a result is also termed the \emph{Quantum Reverse Shannon Theorem}.

At first glance, one may wonder why to lavish noiseless channels on simulating the noisy one. To answer this question, firstly, the noisy coding theorem in conjunction with the Quantum Reverse Shannon Theorem lead us to characterizing \emph{resource inter-conversion} \cite{haddadpour2016simulation, Sudan20}, e.g.~finding the capacity of simulating a channel $\mathscr{N}$ from an arbitrary channel $\mathscr{M}$ in the presence of pre-shared entanglement in terms of a single quantity. Secondly, for the classical and classical-quantum special case, the channel simulation task may date back to the earlier studies on correlation generations \cite{Wyn73, Wyn75, Wyner75, GK73, Cuf08, Sudan20, Yu22}, and channel resolvability \cite{HV93, HV93b, SV94, Hay06_resolvability, Cuf08, Cuf13, cuff2010coordination, Yassaee15, YC19, CG22}, which deliver applications to wiretap channel coding \cite{Wyn75c, BL13, PTM17, Hay132, Hay15, Hay17}, communication complexity of correlation \cite{harsha2010communication}, and measurement compression \cite{Win02, Win04, WHB+12,Berta14}.
Further information-theoretic applications include strong converse theorems \cite{Win02, BSS+02, BDH+14}, lossy data compression \cite{Win02, SV96, LD09, DHW13}, local purity distillation \cite{HHH+03, HHH+05, Dev05, KD07}, and entanglement distillation \cite{DHW08}. We refer the readers to \cite{BSS+02,BDH+14,WHB+12,Yu22,CRB+22} for more comprehensive reviews on this topic. It is worth mentioning that many of the above-mentioned works presume the input source to the channel to be fixed. Nonetheless, we will consider channel simulation under arbitrary sources; that is, we aim for a channel simulation result that works for the worst-case input scenario.

%% Broadcast channels
Notwithstanding the recent advances on point-to-point quantum channel simulation \cite{BCR11, fang2019quantum, RTB23, LY21b}, the asymptotic capacity region of general \emph{quantum broadcast channel} simulation was left unknown prior to the present work. Broadcast channels \cite{Cov72, el2011network, YHD11, DHL10, SW15, Cheng2021b} are arguably one of the most simplest and practical network models in multi-user coding settings. This setting comprises a base station that wants to transmit down-link information to numerous and separate user equipments that only have their own local decodability.\footnote{Of course the base station can send information to each user equipment individually. This amounts to a point-to-point communication with the \emph{Time-Division Multiple Access} scheme, which is a special case of the \emph{Time Sharing} \cite[Proposition 4.1]{el2011network}. However, when there are enormous users equipments involved in the network setting, such a method is considerably time inefficient. Hence, as referred to a network setting in this paper, we demand the sender and receivers to adopt  simultaneous encoding and decoding strategies}
However, despite considerable efforts, the capacity region for a two-receiver broadcast channel coding is still unknown.\footnote{Even though the term ``single-letter characterization'' for a networking information-theoretic task is debatable \cite{Li22, Kor87}, whether Marton's inner bound is optimal is\,---\,to our best knowledge\,---\,still open \cite{Mar79, Cov98, GvdM81, GP80, LK07, LKP08, CK11, el2011network}.
} On the reverse side of channel simulation, somewhat surprisingly, the capacity region for \emph{classical} broadcast channel simulation under common randomness-assistance has recently been characterized in \cite{CRB+22}. This naturally leads to a question:\\

\noindent ``Can one obtain the capacity region for \emph{quantum} broadcast channel simulation under free entanglement-assistance?''\\

\noindent Unfortunately, multi-user tasks are challenging in the fully quantum setting, since they are closely related to a serious technical barrier, the so-called \emph{quantum joint typicality} \cite{Dut11} or \emph{simultaneous smoothing conjecture} \cite{DF13}, and more generally the \emph{quantum marginal problem} \cite{Kly06}. Consequently, previously only specialized scenarios of broadcast channel simulation could be handled. This includes the fixed i.i.d.~input case via time sharing methods \cite{HD07,Ramakrishnan23}, isometric broadcast channels \cite{AJW18} (see also\cite{Ramakrishnan23}), and the recent result \cite{CRB+22} on classical broadcast channels.

%To our best knowledge, such problems are still largely open. 
%(and some tight one-shot characterizations)

In this paper, we establish the capacity region of general quantum broadcast channel simulation under free entanglement-assistance by circumventing the aforementioned obstacles around the quantum joint typicality conjecture. Taking the two-receiver broadcast channel $\mathscr{N}_{A\to B_1 B_2}$ as an example, we show that the channel simulation is asymptotically achievable if and only if the classical communication costs $r_1$ and $r_2$ from the sender to each of two receivers (see Figure~\ref{fig:broadcast_channels}) satisfy
\begin{align}
\left\{ (r_1, r_2) \in \mathds{R}^2:
	r_1 + r_2 \geq I(\mathscr{N}_{A\to B_1 B_2} ), \,
	r_1 \geq I( \mathscr{N}_{A\to B_1}), \,
	r_2 \geq I( \mathscr{N}_{A\to B_2} )
\right\},
\end{align}
where the sum-rate is constrained by the \emph{bipartite mutual information} of channel $\mathscr{N}_{A\to B_1 B_2}$ \cite{Gill54, Watanabe60} (see the precise definition in Eq.~\eqref{eq:multipartite_mutual_information_channel}). Notably, we do not rely on the time-sharing technique and the capacity region for arbitrary $L$ receivers can be obtained as well (Theorem~\ref{theorem:multipartite_simulation}).\footnote{Our approach actually holds for a more stronger notion of \emph{coherent feedback simulation}; see Remark~\ref{remark:coherent_feedback_simulation}.} Our proof techniques build on a conceptually new version of the \emph{multipartite convex-split lemma}, a corresponding \emph{multipartite Quantum State Splitting} protocol, and the (by now standard) \emph{Post-Selection Technique} \cite{ChristKoenRennerPostSelect}. As our approach gets around a generic instance of the quantum joint typicality conjecture with arbitrarily overlapping marginals, we believe that it may serve as a general recipe for fully quantum network information-theoretic tasks.

\emph{Convex splitting} was introduced by Anshu \textit{et al.} in \cite{anshu2017quantum} (see also \cite{anshu2017unified,CG23}), which originates from the idea of \emph{rejection sampling} in statistics \cite{vonNeumann1951Various,robert1999monte,jain2003direct}, and it has an ample of applications in quantum information theory. In this paper, we generalize it to a conceptually new multipartite version and establish an one-shot error exponent bound.\footnote{A specialized version of the bipartite convex splitting with quantum relative entropy as an error criterion was introduced in \cite[Lemma 2]{AJW18}. We refer the readers to Section~\ref{sec:related_work} for discussions and comparisons.} Taking the bipartite version as an illustration, assume that there are $M$ independent and identical copies of registers $A$'s, $K$ copies of registers $B$'s, and a single register $E$. Now the overall system is prepared in a way that register $E$ is correlated to the $m$-th register $A$ and the $k$-th register $K$ uniformly at random (see Figure~\ref{fig:convex_splitting}). We then prove a tight one-shot bound on the trace distance between such a joint state and the all-tensor-product state (i.e.~all of the systems $A$'s, $B$'s, and $E$ are decoupled) in terms of a generalized quantum R\'enyi information \cite{MDS+13, WWY14, HT14} (Theorem~\ref{theorem:bipartite_convex_splitting}). The additivity of the R\'enyi information then immediately gives us a rate region for $M$ and $K$, for which the trace distance decreases exponentially in the asymptotic limit. This thus can be viewed as a bipartite generalization of the unipartite convex splitting by part of the authors \cite{CG23}. To establish the result and avoid the need of the simultaneous smoothing, we introduce a key ingredient of a decomposition map, the \emph{multipartite mean-zero decomposition lemma} (Lemma~\ref{lemma:tele} \& Lemma \ref{lemma:tele2}).\footnote{Technically speaking, we do not employ the \emph{smooth entropy framework}; instead, we use the interpolation technique as in the unipartite setting \cite{Dup21, CG23}. Hence, what we avoid should be termed as \emph{simultaneous interpolation}.}
We remark that a similar idea is independently proposed by Colomer Saus and Winter for deriving multipartite quantum decoupling theorems, termed the \emph{telescoping trick} in their work \cite{SW23}.

An immediate application of the multipartite convex-split lemma is the multipartite version of \emph{Quantum State Splitting} \cite{HD07,Ramakrishnan23} also termed \emph{mother protocol} in its original (non-multipartite) form \cite{Proc465}. The goal is to transfer the systems $A_1'$, $A_2'$, $\ldots$, $A_L'$ initially with Alice, to Bob, while the entanglement between all of Alice's original systems $A$ and an inaccessible reference system $R$ is preserved. Given any classical communication cost consumed in the protocol, we obtain an one-shot error exponent bound on how well the protocol is performed.

Armed with multipartite Quantum State Splitting, we demonstrate how to combine the Post-Selection Technique \cite{ChristKoenRennerPostSelect} (as used in point-to-point quantum channel simulation \cite{BCR11}) and the quantum sandwiched information (Lemmas~\ref{lemma:dimension_bound} and \ref{lemma:convexity}), to establish a one-shot error exponent bound for multipartite quantum broadcast channel simulation (Theorem~\ref{theorem:multipartite_simulation}) with diamond norm \cite{Kit97,Pau03} as an error criterion. 
We note that such a one-shot bound not only leads to the optimal achievable rate region in the i.i.d.~setting; it also guarantees that the error of channel simulation decreases exponentially fast in the number of blocklength $n$ whenever the rate vector of communication costs is in the interior of the capacity region.
Furthermore, we show the achievability part of a \emph{moderate deviation} result \cite{CTT2017, CH17}. Namely, the error of simulation will vanish asymptotically, even though the rate vector converges to the boundary of the capacity region (Proposition~\ref{proposition:moderate_simulation}).

Last but not least, let us point out some distinctive features of our results. The established multipartite convex splitting (Theorem~\ref{theorem:convex_splitting_multipartite}) and multipartite Quantum State Splitting (Theorem~\ref{theorem:multipartite_QSS}) are \emph{one shot}, in the sense that no mathematical constraint such as the i.i.d.~assumption is needed. As for quantum broadcast channel simulation, we require the i.i.d.~structure for the underlying channel to simulate. Our result (Theorem~\ref{theorem:multipartite_simulation}) is therefore \emph{non-asymptotic} and for \emph{any} {finite blocklength}, i.e.~the assumption of infinitely large blocklength is not required. The reason behind our results is that we do not employ \emph{time-sharing techniques} \cite[Proposition 4.1]{el2011network}, as this is not possible for the one-shot or finite blocklength setting (as pointed out even in classical network information theory \cite[Remark 3]{YAG13}; see also \cite{LCV15}). By its nature, we believe that the proposed one-shot analysis, and in particular the multipartite \emph{mean-zero decomposition lemma} (Lemma~\ref{lemma:tele2}), could be a generic solution to problems in quantum network information theory.

\medskip
The paper is structured in the following. We present an overview of our technical results in Section~\ref{sec:overview}, and Section~\ref{sec:related_work} contains discussions of related work. Notation and definitions for information measures are introduced in Section~\ref{sec:notation}. Section~\ref{sec:convex_splitting} is devoted to establishing the multipartite convex splitting. The multipartite Quantum State Splitting achievability result is derived in Section~\ref{sec:QSS}. We give the proof of quantum broadcast channel simulation in Section~\ref{sec:simulation}. Some technical lemmas are left to the Appendices \ref{sec:lemmas} and \ref{sec:Post-Selection}.

\begin{figure}[h!]
	\centering
	\resizebox{0.9\columnwidth}{!}{ 
		\includegraphics{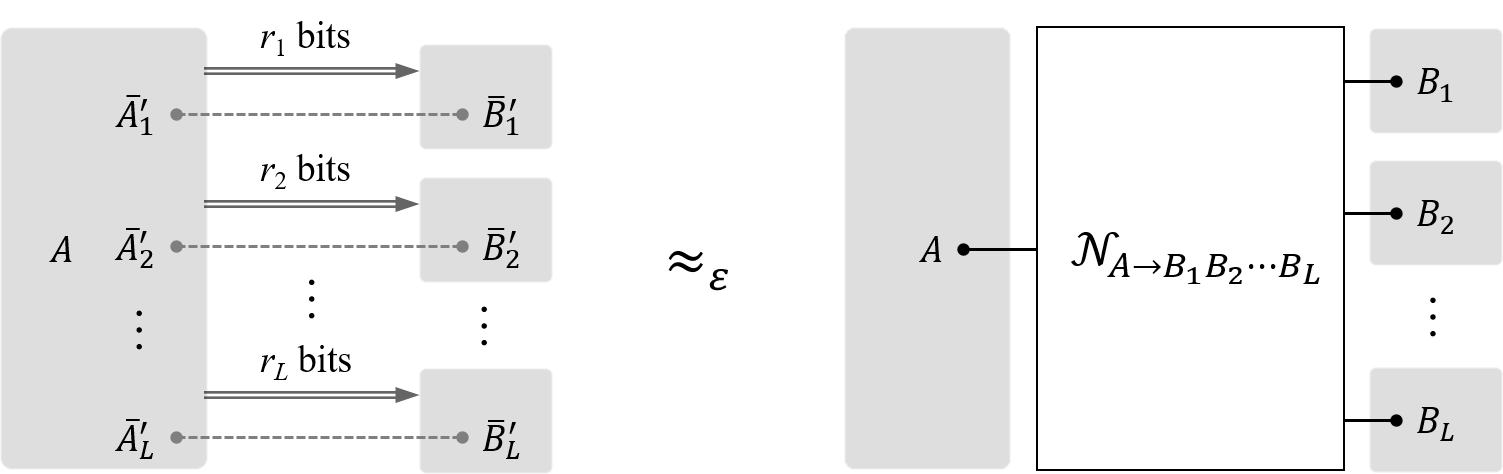}     
	}
	\caption{
    \small
    Depiction of an $L$-receiver quantum broadcast channel simulation.
    The Sender holds system $A$ and share free entanglement with $L$ receivers. By sending $r_1$, $r_2$, $\ldots$, $r_L$ bits of classical information to each receiver and local quantum operation at each receiver, the resulting transformation is close to an $L$-receiver quantum broadcast channel $\mathscr{N}_{A\to B_1B_2\ldots B_L}$ (in diamond norm).
    Note that each gray-shaded region is only allowed to perform local quantum operation.
	}
	\label{fig:broadcast_channels}
\end{figure}

\begin{figure}[h!]
	\centering
	\resizebox{0.9\columnwidth}{!}{ 
		\includegraphics{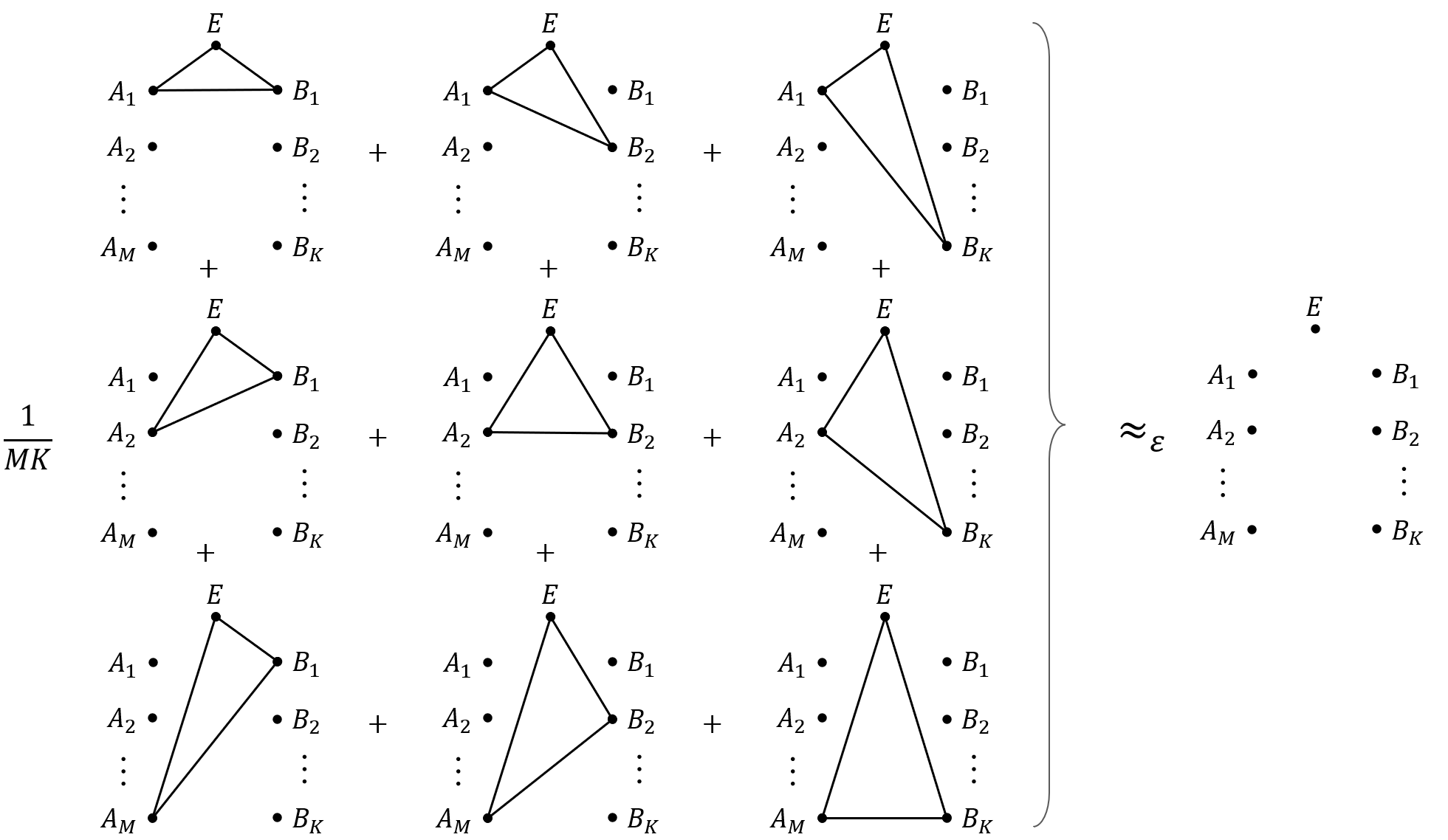}     		   		
	}
	\caption{
    \small
    Depiction of a bipartite convex splitting.
    On the right part, each isolated dot represents an independent quantum system, and hence the overall state is product.
    On the left part, solid lines connected each system is a tripartite (correlated) state, while the other systems are left isolated. 
    When $M$ and $K$ are sufficiently large (see Theorem~\ref{theorem:bipartite_convex_splitting}), then the statistical mixture of the left part is close to the right part of product state (in trace norm).
	}
	\label{fig:convex_splitting}
\end{figure}

%%%%%%%%%%%%%%%%%%%%%%%%%%%%%%%%

\newpage

\subsection{Overview of Results} \label{sec:overview}

Our results are summarized below.
\begin{enumerate}[1)]
    \item \textit{Multipartite convex splitting}:
    For readability, we first present the result of the bipartite convex splitting.
    Let $\rho_{ABE}$, $\tau_A$, $\tau_B$ be states and let $M$ and $K$ be integers.
    The trace distance (i.e.~trace norm divided by two) between the random mixture
    \begin{align}
    	\frac{1}{MK}\sum_{m,k}\rho_{A_{m} B_{k}E}  \bigotimes_{ \bar{m} \neq m, \, \bar{k} \neq k } \tau_{A_{\bar{m}}} \otimes \tau_{B_{\bar{k}}}
    \end{align}
    and the product state $\tau_{A}^{\otimes M} \otimes \tau_{B}^{\otimes K} \otimes \rho_E$ is upper bounded by (Theorem~\ref{theorem:bipartite_convex_splitting})
    \begin{align} \label{eq:error_bipartite_convex_splitting0}
        2\cdot \,2^{ -  E_{\log MK}(\rho_{ABE} \,\Vert\, \tau_A \otimes \tau_B ) }
        +  2^{- E_{\log M}(\rho_{AE} \,\Vert\, \tau_A )}
        +  2^{- E_{\log K}(\rho_{BE} \,\Vert\, \tau_B )},
    \end{align}
    where the error-exponent function
    \begin{align}
    	E_r(\rho_{AE}\,\Vert\, \tau_A) := \sup_{\alpha\in[1,2]} \frac{\alpha-1}{\alpha}\left(r - I_\alpha(\rho_{AE}\,\Vert\, \tau_A)\right)
    \end{align}
		is the Fenchel--Legendre transform of the R\'enyi information \cite{HT14}
		\begin{align}
			I_\alpha(\rho_{AE}\,\Vert\, \tau_A):= \inf_{\sigma_E} D_\alpha(\rho_{AE}\,\Vert\, \rho_A\otimes \sigma_E)
		\end{align}
		and $D_\alpha(\rho\,\Vert\,\sigma) := \frac{1}{\alpha-1} \log \Tr\left[ (\sigma^{\frac{1-\alpha}{2\alpha}} \rho \sigma^{\frac{1-\alpha}{2\alpha}} )^\alpha\right]$ is the quantum sandwiched R\'enyi divergence \cite{MDS+13, WWY14}.
    Moreover, the three error exponents in Eq.~\eqref{eq:error_bipartite_convex_splitting0} are all positive if and only if 
    \begin{align} \label{eq:region_bipartite_convex_splitting0}
		\begin{cases}
		&\log MK > I_1(\rho_{ABE} \,\Vert\, \tau_A \otimes \tau_B ) \\
	&\log M > I_1(\rho_{AE} \,\Vert\, \tau_A )\\
	&\log K > I_1(\rho_{BE} \,\Vert\, \tau_B )\\
	\end{cases}.
	\end{align}
    This then gives us an achievable rate region for a bipartite convex splitting.
    
    The result generalizes to arbitrary $L$-party case. Let $\rho_{A_1 A_2\ldots A_L E}$ and $\tau_{A_{\ell}}$ be states and $M_\ell$ be an integer for each $\ell \in \{1,2,\ldots, L\} =: [L]$.
    The trace distance between the random mixture
    \begin{align}
    	\frac{1}{M_1 \cdots M_L}\sum_{m_\ell \in [M_\ell], \, \ell\in[L]}\rho_{A_{ 1, m_1} \cdots A_{ L, m_L} E}\bigotimes_{  \bar{m}_{\ell} \in [M_{\ell}]/\{m_\ell\}, \, \ell \in [L] } \tau_{ A_{\ell, \bar{m}_{\ell} } }
    \end{align}
    and the product state $\bigotimes_{\ell\in[L]} \tau_{A}^{\otimes M_\ell} \otimes \rho_E$ is upper bounded by (Theorem~\ref{theorem:convex_splitting_multipartite})
    \begin{align} 
   \sum_{\varnothing\neq S\subseteq[L] } 2^{|S|} \cdot 2^{ - 
    E_{\log M_S}\left( \rho_{A_S E} \,\Vert\, \tau_{A_S}\right)},
    \end{align}
    where $A_S$ denotes systems $A_\ell$ for all $\ell \in S$, and $M_S := \prod_{\ell\in S} M_\ell$.
    Again, the achievable rate region is given by 
    \begin{align}
    \left\{  (M_1,M_2, \ldots, M_L) :  \log M_S > I_1(\rho_{A_S} \,\Vert\, \tau_{A_S} ), \, \forall S\subseteq[L] \right\}.
    \end{align}
    
    \item \textit{Multipartite Quantum State Splitting}:
    Let $\rho_{AA_1'A_2'\ldots A_L' R}$ be a pure state holding by Alice and an inaccessible reference system $R$. Suppose that Alice and the $\ell$-th Receiver share many-copies of entangled state $|\tau\rangle_{\bar{A}_\ell' \bar{B}_\ell}$, where $\bar{A}_\ell' \cong \bar{B}_\ell \cong A_\ell'$.
    The goal of an $L$-party Quantum State Splitting is to transfer each system $A_\ell'$ to the $\ell$-th Receiver via $r_\ell$ bits of classical communication. Then, the error in terms of trace distance is upper bounded by (Theorem~\ref{theorem:multipartite_QSS})
    \begin{align}
			\sqrt{ \sum\nolimits_{\varnothing\neq S\subseteq[L] } 2^{|S|} \cdot 2^{ - 
    	E_{r_S}\left( \rho_{A_S' R} \,\Vert\, \tau_{A_S'}\right)}},
    \end{align}
    where $r_S:= \sum_{ \ell \in S } r_\ell$. Moreover, the error exponents are all positive if and only if for all $S \subseteq [L]$ ,
		\begin{align}
    	r_S > I_1\left( \rho_{A_S' R} \,\Vert\, \tau_{A_S'}\right).
		\end{align}

    \item \textit{Quantum Broadcast Channel Simulation}:
    Consider a general $L$-receiver quantum broadcast channel $\mathscr{N}_{A \to B_1B_2\ldots B_L}$, and free entanglement is present between Sender and $L$ Receivers. By sending $n r_\ell$ bits of classical information to the $\ell$-th Receiver, respectively, the simulated channel is close to $\mathscr{N}_{A \to B_1B_2\ldots B_L}^{\otimes n}$ in diamond norm \cite{Kit97,Pau03} with error at most (Theorem~\ref{theorem:multipartite_simulation})
    \begin{align} \label{eq:error_broadcast0}
    	\eps_n \leq k_n \cdot \sqrt{ \sum\nolimits_{\varnothing \neq S \subseteq [L]} 2^{|S|} \cdot 2^{- n E_{r_S}(\mathscr{N}_{A\to  B_S}) }},
		\end{align}
	where $B_S$ denotes the systems $B_\ell$ for all $\ell \in S$, and $k_n$ is a polynomial pre-factor depending on $L$ and the dimension of the input space. The function $E_{r_S}(\mathscr{N}_{A\to  B_S})$ is the Fenchel--Legendre transform of $\sup_{\psi_{AR}} I_\alpha(\mathscr{N}_{A\to B_S}(\psi_{AR})\,\Vert\, \mathscr{N}_{A\to B_S}(\psi_A))$. 
 Via a minimax identity (Proposition~\ref{fact:properties_exponent}), it is equal to the error exponent corresponding to the worst-case purification $\psi_{AR}$. 
 We remark that this result holds for \emph{any} finite-blocklength $n$. The achievability result given in Eq.~\eqref{eq:error_broadcast0} leads us a lower bound on the overall error exponent (Proposition~\ref{proposition:expoent_simulation})
    \begin{align}
      \liminf_{n\to \infty} -\frac1n \log \eps_n \geq \frac12\min_{\varnothing \neq S\subseteq [L]}  E_{r_S}\left(\mathscr{N}_{A\to B_S}\right).
    \end{align}
  Then, we show that the capacity region of the broadcast channel simulation is given by (Theorem~\ref{theorem:QRST_multipartite})
		\begin{align}
			\left\{ (r_1 , r_2, \ldots, r_L):\sum\nolimits_{ \ell \in S } r_\ell \geq I(\mathscr{N}_{A\to B_S}), \forall S \subseteq [L]\right\}.
		\end{align}

	The achievability in Eq.~\eqref{eq:error_broadcast0} also implies an achievability of the \emph{moderate deviation} result \cite{CTT2017, CH17} as follows. For any subsets $\varnothing \neq S \subseteq [L]$, assume that the rates satisfy $\sum_{\ell\in S} r_{\ell,n} - I(\mathscr{N}_{A\to B_S}) = \Theta(n^{-t})$ for some $t\in (0,\sfrac12)$. That is, the $L$-tupe rate asymptotically converges to the boundary of the capacity region at certain speed. Then, we show that the error $\eps_n$ of simulation still vanishes asymptotically (Proposition~\ref{proposition:moderate_simulation})
  \begin{align}
  	\eps_n \leq O\left(2^{- n^{1-2t} } \right) \to 0.
  \end{align}
\end{enumerate}

%%%%%%%%%%%%%%%%%%%%%%%%%%%%%%%%

\subsection{Related Work} \label{sec:related_work}

The time reversal of Quantum State Splitting corresponds to coherent \emph{Quantum State Merging} 
\cite{Proc465}. Further, a variant thereof, termed \emph{non-coherent} Quantum State Merging was proposed by Horodecki \emph{et al.} in \cite{Horodecki2005, horodecki2007quantum}. For the latter, one gives free LOCC and quantifies the entanglement needed to achieve the task of Quantum State Merging, whereas for the former one gives free entanglement and quantifies the communication requirements. The multipartite version of non-coherent Quantum State Merging was already studied in \cite{horodecki2007quantum} as well, which then corresponds to the task of fully quantum \emph{distributed compression}. We note that this multipartite Quantum State Merging result relies on time-sharing techniques \cite[Proposition 4.1]{el2011network}, meaning that once the the (simple) corner points of the capacity region are achieved, time sharing leads to the convex hull of them.\footnote{We remark that the time-sharing technique requires \emph{synchronization} between the sender and receivers \cite[Remark 4.3]{el2011network}. Such a requirement can be practically demanding when the there is a large number of users in the network.} Subsequently, the need of time sharing was removed with the use of chained typicality projectors \cite{DH10, Dut11}, that work for this special task because of the non-overlapping structure of the relevant quantum marginals. Moreover, we believe that these results could be lifted to the one-shot setting by means of the techniques from \cite{DF13} and using the minimax smoothing from \cite{anshu2019minimax} even to tight cost functions in terms of smooth conditional min-entropy. However, working instead with the coherent Quantum State Merging task, it is unclear to us how to transform these entanglement cost functions to communication cost functions, in a way that retains joint smoothing and allows to later apply the Post-Selection Technique for quantum broadcast channel simulation. In short, for multipartite coherent \emph{Quantum State Merging} one needs to solve an instance of the joint smoothing problem with overlapping marginals, whereas for multipartite non-coherent \emph{Quantum State Merging} one gets away with non-overlapping marginals (becaus of the different structure of the cost function).

The technique of (unipartite) \emph{convex splitting} was introduced by Anshu \textit{et al.} in \cite{anshu2017quantum, anshu2017unified}. The idea originated from \emph{rejection sampling} in statistics \cite{vonNeumann1951Various},~\cite[Chapter 2.3]{robert1999monte}, \cite{jain2003direct}, and one of its specialized cases dates back to the classical \emph{soft covering} by Wyner \textit{et al.} \cite{Wyner75, HV93, HV93b, AD89, AW02, Hay06, Cuf13, Hay15, HM16, PTM17, YT19, CG22, SGC22b}. It was recently generalized to an one-shot error-exponent result by parts of the authors \cite{CG23}. The {bipartite} classical convex-split lemma\,---\,for which all the density operators share the same eigenbasis\,---\,was shown in \cite[Fact 7]{anshu2017unified}. Its straightforward generalization to the multipartite version was shown in \cite[Lemma 39]{CRB+22}, which is the key lemma for showing classical broadcast channel simulation.

A bipartite \emph{quantum} convex-split lemma was shown in \cite[Lemma 2]{AJW18}.\footnote{The authors termed \cite[Lemma 2]{AJW18} as the ``tripartite'' convex-split lemma. However, we intend to call it a ``bipartite'' convex splitting for the following reason. The unipartite convex splitting \cite{anshu2017quantum, CG23} aims to decouple systems $A$ and $E$, while the bipartite convex splitting aims to decouple systems $A$ into two parts, i.e.~systems $A$ and $B$ in Eq.~\eqref{eq:error_bipartite_convex_splitting0}. Our terminology is also consistent with the classical convex-split lemma used in \cite[Fact 7]{anshu2017unified}.} However, it is not clear to us whether this would lead to the optimal achievable rate region in the i.i.d.~setting. In fact, \cite[Theorem 4]{AJW18} provides an asymptotic i.i.d.~analysis with an assumption that the whole joint state, i.e.~$\rho_{ABE}$ in Eq.~\eqref{eq:error_bipartite_convex_splitting0}, is pure. To the best of our knowledge, this will only give an \emph{isometric broadcast channel simulation} (see also\cite{Ramakrishnan23}), rendering the simulation of \emph{general} quantum broadcast channels previously open. On the other hand, a ``bipartite'' convex-split lemma was also shown in \cite{AJW19a}, wherein the system $E$ is absent. We remark that a multipartite convex-split lemma with the presence of system $E$ is crucial for the multipartite Quantum State Splitting and quantum broadcast channel simulation since the system $E$ plays the role of the reference system $R$ with which we want to protect the entanglement.

Indeed, establishing a one-shot achievability lemma for bipartite or multipartite settings that will yield the right achievable rate region is the central problem in quantum network information theory \cite{Kly06,DF13, horodecki2007quantum, DH10, Dut11, DGH+20, Sen21}. Our breakthrough here is to introduce a mean-zero decomposition lemma (Lemma~\ref{lemma:tele}) such that we can generally bypass the \emph{simultaneous smoothing/interpolation} obstacles in the fully quantum setting. Taking the bipartite convex splitting as an example. The established one-shot bound in Eq.~\eqref{eq:error_bipartite_convex_splitting0} immediately implies the achievable region given in Eq.~\eqref{eq:region_bipartite_convex_splitting0}. Hence, no sophisticated tools in asymptotic equipartition property such as typical projection onto subsystems, gentle measurement lemma \cite{Win99}, and the second-order asymptotics are needed. On the contrary, only the additivity of R\'enyi information is needed.

%Hence, the analysis for finding the rate region of quantum network information-theoretic tasks provided in this work is arguably much elegant.

%%%%%%%%%%%%%%%%%%%%%%%%%%%%%%%%

\section{Notation and Information Quantities} \label{sec:notation}

Throughout this paper, the underlying Hilbert spaces associated to quantum registers/systems $A$, $B$, $C$, $\ldots$ are denoted by sans-serif fonts $\mathsf{A}$, $\mathsf{B}$, $\mathsf{C}$, $\ldots$, et cetera.
The set of density operators (i.e.~positive semi-definite operators with unit trace) and bounded operators on $\mathsf{A}$ are denoted by $\mathcal{S}(\mathsf{A})$ and $\mathcal{B}(\mathsf{A})$, respectively.
The notation $|\mathsf{A}|$ stands for the dimension of Hilbert space $\mathsf{A}$.
For a density operator $\rho_A$ or a bounded operator $X_A$ with subscript $A$, we mean the operator is on the Hilbert space $\mathsf{A}$ and often as the corresponding marginal of a multippartite operator $X$.
We will use the term density operator and quantum state interchangeably in this paper.

For any $X \in \mathcal{B}(\mathsf{A})$ and , we define $\| X \|_p := \left( \Tr[|X|^p] \right)^{\sfrac{1}{p}} $ for $p>0$.
We use $\mathds{N}$ and $\mathds{R}$ to denote positive integers and real numbers.
For any positive integer $M\in\mathds{N}$, we shorthand the set $[M]:= \{1,2,\ldots, M\}$.
We denote $\id_R: \mathcal{B}(\mathsf{R}) \to \mathcal{B}(\mathsf{R})$ as the canonical identity map on $\mathsf{R}$, and denote $\id_R$ as the identity operator on $\mathsf{R}$.
For any linear map $\mathscr{N}_{A\to B}: \mathcal{B}(\mathsf{A}) \to \mathcal{B}(\mathsf{B})$, we define the diamond norm \cite{Kit97, Pau03} as
\begin{align}
    \left\|\mathscr{N}\right\|_{\diamond} := \sup_{\norm{\rho_{AR}}{1} =1} \left\| \mathscr{N}_{A\to B}\otimes \id_R(\rho_{AR}) \right\|_1,
\end{align}
where $\mathsf{R} \cong \mathsf{A}$.

For any $\alpha\in(0,1)\cup(1,\infty)$, we define the order-$\alpha$ sandwiched quantum R\'enyi divergence $D_\alpha$ \cite{MDS+13,WWY14} for density operator $\rho \in \mathcal{S}(\mathsf{A})$ and positive semi-definite $\sigma \in \mathcal{B}(\mathsf{A})$ as
\begin{align} \label{eq:sandwiched}
D_\alpha(\rho\, \Vert\, \sigma)
&:=\frac{\alpha}{\al-1}\log \left\|\sigma^{\frac{1-\alpha}{2\alpha}}\rho\si^{\frac{1-\alpha}{2\alpha}}\right\|_{\al},
\end{align}
provided that the support of $\rho$ is contained in that of $\sigma$; otherwise, it is defined to be positive infinity.
%$\texttt{supp}(\rho)\subseteq \texttt{supp}(\sigma)$.
We define a \emph{generalized sandwiched R\'enyi information} \cite{HT14} for a bipartite state $\rho_{AB} \in \mathcal{S}(\mathsf{A}\otimes \mathsf{B})$ and positive semi-definite $\tau_A \in \mathcal{B}(\mathsf{A})$ and the usual \emph{sandwiched R\'enyi information} as
	\begin{align}
		I_\alpha(\rho_{AB} \,\Vert\, \tau_A ) :&= \inf_{\sigma_B \in \mathcal{S}(\mathsf{B}) } D_\alpha(\rho_{AB} \,\Vert\, \tau_A\ten \sigma_B); \label{eq:sandwiched_generalized} \\ 
    I_\alpha(A: B)_\rho &:= I_\alpha(\rho_{AB} \,\Vert\, \rho_A ).
	\end{align}
% When $\tau_A = \rho_A$, it reduces to the usual sandwiched R\'enyi information, i.e.
% \begin{align}
%      I_\alpha(A: B)_\rho &:= I_\alpha(\rho_{AB} \,\Vert\, \rho_A ).
% \end{align}
Moreover, we define the order-$\alpha$ sandwiched R\'enyi information for a quantum channel (i.e.~completely positive and trace-preserving map) $\mathscr{N}_{A\to B}$ as\footnote{Note that our definition of channel sandwiched R\'enyi information given in Eq.~\eqref{eq:sandwiched_mutual_information_channel} is different from the one given in \cite[\S 7]{KW20}, wherein the authors considered $I_\alpha(\mathscr{N}_{A\to B}) := \sup_{\psi_{AR}} I_\alpha (R:B)_{(\mathscr{N}_{A\to B}\otimes \id_R)(\psi)}$, with $B$ and $R$ swapped comparing to \eqref{eq:sandwiched_mutual_information_channel}. The two definitions coincide as $\alpha = 1$.}
\begin{align} \label{eq:sandwiched_mutual_information_channel}
    I_\alpha(\mathscr{N}_{A\to B}) := \sup_{\psi_{AR}} I_\alpha (B:R)_{(\mathscr{N}_{A\to B}\otimes \id_R)(\psi)  }.
\end{align}
Here and throughout this paper, $\sup_{\psi_{AR}}$ denotes maximizing over all pure states $\psi_{AR} \in\mathcal{S}(\mathsf{A}\otimes \mathsf{R})$ and $\mathsf{R} \cong \mathsf{A}$.

\begin{remark}
    If $|\mathsf{B}|<\infty$, then by the compactness of $\mathcal{S}(\mathsf{B})$, lower semi-continuity of the sandwiched R\'enyi divergence in its second argument \cite{Tom16}, and the extreme value theorem, the minimum in Eq.~\eqref{eq:sandwiched_generalized} can be attained. 
    An iterative algorithm with convergence guarantees was provided in \cite{YCL21}.
    In this paper, although we do not necessary impose the finite-dimensional assumption on the Hilbert space (especially for the output space), the minimum in Eq.~\eqref{eq:sandwiched_generalized} can also be attained.
\end{remark}

As $\alpha\to 1$, the above quantities converge monotonically to the Umegaki's \emph{quantum relative entropy} \cite{Ume62},\cite[Lemma 3.5]{MO14}, 
the generalized mutual information,
 the usual quantum mutual information, and the quantum mutual information of channel, respectively:
\begin{align}
\lim_{\alpha\to 1} D_\alpha(\rho \,\Vert\, \sigma) &= D(\rho \,\Vert\, \sigma) := \Tr\left[ \rho (\log \rho - \log \sigma ) \right], \\
\lim_{\alpha\to 1} I_\alpha(\rho_{AB} \,\Vert\, \tau_A ) &= 
I(\rho_{AB} \,\Vert\, \tau_A ) := D(\rho_{AB} \,\Vert\, \tau_A \otimes \rho_B), \\
\lim_{\alpha\to 1} I_\alpha(A: B)_\rho &= I(A: B)_\rho
     := D(\rho_{AB} \,\Vert\, \rho_A \otimes \rho_B), \\
I(\mathscr{N}_{A\to B})
&:= \sup_{\psi_{RA}} I (B:R)_{(\mathscr{N}_{A\to B}\otimes \id_R)(\psi)}.
\end{align}

% Given a quantum channel (i.e.~completely positive and trace-preserving map) $\mathscr{N}_{A\to B}$, we define
% \begin{align}
%     C(\mathscr{N}_{A\to B}) := \max_{\psi_{AR} \in \mathcal{S}(\mathsf{A}\otimes \mathsf{R}) } I(R: B)_{ (\mathscr{N}_{A\to B}\otimes \id_R)(\psi)  }.
% \end{align}

The above information quantities naturally extend to the multipartite setting (where system $A$ now has, say $L\in\mathds{N}$, subsystems).
% Subsequently, we let $\rho_{A_1 A_2 \ldots A_L B} \in \mathcal{S}(\mathsf{A}_1\otimes \mathsf{A}_2\otimes \cdots \otimes \mathsf{A}_L \otimes \mathsf{B})$ be any multipartite state.
We define a \emph{multipartite sandwiched R\'enyi information} for a multipartite state $\rho_{A_1 A_2 \ldots A_L E} \in \mathcal{S}(\mathsf{A}_1\otimes \mathsf{A}_2\otimes \cdots \otimes \mathsf{A}_L \otimes \mathsf{E})$ and a quantum broadcast channel $\mathscr{N}_{A\to B_1 B_2 \ldots B_L} :\mathcal{S}(\mathsf{A}) \to \mathcal{S}(\mathsf{B}_1 \otimes \mathsf{B}_2 \otimes \cdots \otimes \mathsf{B}_L)$ as
\begin{align}
I_\alpha\left(A_1 : A_2 : \cdots : A_L : E\right)_\rho &:=
I_\alpha\left( \rho_{A_1 A_2 \ldots A_L E} \, \Vert \otimes_{\ell\in[L]} \rho_{A_{\ell}} \right); \\
I_\alpha( \mathscr{N}_{A\to B_1 B_2 \ldots B_L} ) 
&:= \max_{\psi_{AR}} I_\alpha\left(B_1 : B_2 : \cdots : B_L : R\right)_{ (\mathscr{N}\otimes \id_R)(\psi)}.
\end{align}
In particular, we term 
\begin{align}
    \lim_{\alpha\to1} I_\alpha\left(A_1 : A_2 : \cdots : A_L : E\right)_\rho &\equiv I\left(A_1 : \cdots : A_L: E \right)_\rho
		:= D\left( \rho_{A_1 \ldots A_L E} \,\Vert \otimes_{\ell \in [L]} \rho_{A_\ell} \otimes \rho_E \right)_\rho, \label{eq:multipartite_mutual_information} \\
    I( \mathscr{N}_{A\to B_1 B_2 \ldots B_L} ) 
    &:= \max_{\psi_{AR}} I\left(B_1 : B_2 : \cdots : B_L : R\right)_{ (\mathscr{N}\otimes \id_R)(\psi)} \label{eq:multipartite_mutual_information_channel}
\end{align}
as the \emph{multipartite quantum mutual information} \cite{Gill54, Watanabe60} $I\left(A_1 : \cdots : A_L: E \right)_\rho$ for state $\rho_{A_1 A_2 \ldots A_L E}$ and the \emph{quantum mutual information} $I( \mathscr{N}_{A\to B_1 B_2 \ldots B_L} ) $ for quantum broadcast channel $\mathscr{N}_{A\to B_1 B_2 \ldots B_L}$, respectively.\footnote{We note that sometimes the multipartite quantum mutual information is defined in an alternative way, e.g.~$I(A:B:C)_\rho := I(A : B)_\rho + I(A : C)_\rho - I(A : BC)_\rho$. However, we will adopt the definition in Eq.~\eqref{eq:multipartite_mutual_information} in this paper.}

We further define the \emph{relative entropy variance} \cite{TH13, Li14} $V(\rho\,\Vert\,\sigma)$, \emph{mutual information variance} ${V(A_1 : \cdots : A_L : E)_\rho}$, and \emph{channel dispersion} $V(\mathscr{N}_{A\to B_1 B_2 \ldots B_L})$ for a broadcast channel $\mathscr{N}_{A\to B_1 B_2 \ldots B_L}$ as
\begin{align}
    V(\rho\,\Vert\,\sigma) &:= \Tr[\rho (\log \rho - \log \sigma)^2] - \left[ D(\rho\,\Vert\,\sigma) \right]^2; \\
    V\left(A_1 : \cdots : A_L : E\right)_\rho
	&:= V\left( \rho_{A_1 \ldots A_L E} \,\Vert \otimes_{\ell \in [L]} \rho_{A_\ell} \otimes \rho_E \right); \\
    V(\mathscr{N}_{A\to B_1 B_2 \ldots B_L}) &:= \sup_{\psi_{AR}: I(B_1:\cdots:B_L:R)_{(\mathscr{N}\otimes \id_R)(\psi)} = I(\mathscr{N}) }  V(B_1:\cdots:B_L:R)_{(\mathscr{N}\otimes \id_R)(\psi)}. \label{eq:V}
\end{align}
	
Below we collect several well-known properties of the generalized R\'enyi information. Essentially they all follow similarly as the special case $L=1$
\cite{WWY14, MO14}, \cite[Lemma 7]{HT14}, \cite{CGH18, CG22}, and \cite[Lemma 17]{LY21b}.

\begin{proposition}[Properties of R\'enyi information] \label{fact:properties_sandwiched}
	For any multipartite state $\rho_{A_1 A_2 \ldots A_L E} $ and $\tau_{A_\ell} \in\mathcal{S}(\mathsf{A}_{\ell})$, $\ell\in[L]$, and $\alpha \geq \sfrac12$, the sandwiched R\'enyi information satisfies the following properties.
	\begin{enumerate}[(a)]
		\item\label{item:limiting_sandwiched} 
		(Monotone decreasing \cite{WWY14},\cite[Lemma 4.6]{MO14}) When $\al \to 1$,
		\begin{align}
		I_\alpha\left(\rho_{A_1 A_2 \ldots A_L E}\, \Vert \otimes_{\ell\in[L]} \tau_{A_\ell} \right)
        \searrow I\left(\rho_{A_1 A_2 \ldots A_L E}\, \Vert \otimes_{\ell\in[L]} \tau_{A_\ell} \right).
		% = D\left(\rho_{A_1 A_2 \ldots A_L E}\, \Vert \otimes_{\ell\in[L]} \tau_{A_\ell} \otimes \rho_E\right).
		\end{align}
		
		% \item\label{item:derivative_sandwiched} 
		% (First-order derivative) Provided that the underlying Hilbert spaces are all finite dimensional, then
		% \begin{align}
		% 2\left.\frac{\mathrm{d}}{\mathrm{d} \alpha} I_\alpha\left(\rho_{A_1 A_2 \ldots A_L E}\, \Vert \otimes_{\ell\in[L]} \tau_{A_\ell} \right)\right|_{\alpha=1} 
		% =  V\left( \rho_{A_1 \ldots A_L E} \,\Vert \otimes_{\ell \in [L]} \rho_{A_\ell} \otimes \tau_E \right).
		% \end{align}
		% where the relative entropy variance \cite{TH13, Li14} is 
		% $V(\rho\,\Vert\,\sigma) := \Tr[\rho (\log \rho - \log \sigma)^2] - \left[ D(\rho\,\Vert\,\sigma) \right]^2$.
		
		\item\label{item:additivity_sandwiched}
		(Additivity \cite[Lemma 7]{HT14})
        For any integer $n \in \mathds{N}$,
		\begin{align}
		I_\alpha\left(\rho_{A_1 A_2 \ldots A_L E}^{\otimes n}\, \Vert \otimes_{\ell\in[L]} \tau_{A_\ell}^{\otimes n} \right)
		= n I_\alpha\left(\rho_{A_1 A_2 \ldots A_L E}\, \Vert \otimes_{\ell\in[L]} \tau_{A_\ell} \right).
		\end{align}
		
		% \item\label{item:convexity_sandwiched}
		% (Concavity in $\tau$) The following map is jointly convex and lower semi-continuous on $\mathcal{S}(\mathsf{A}_1) \times \cdots \times \mathcal{S}(\mathsf{A}_L)$: 
		% \begin{align} 
		% (\tau_1, \ldots, \tau_L) \mapsto I_\alpha\left( \rho_{A_1 A_2 \ldots A_L} \, \Vert \otimes_{\ell\in[L]} \tau_{A_{\ell}} \right).
		% \end{align}
		
		\item\label{item:concavity_sandwiched}
		(Concavity in $\alpha$ \cite[Theorem 11]{CGH18})\footnote{We note that the concavity established in \cite{CGH18} is for a special case that system $A$ is classical and $\tau_{A} = \rho_{A}$. However, the same reasoning applies here as well.} The following map is concave 
     and upper semi-continuous 
    on $(1,\infty)$:
		\begin{align}
		\alpha \mapsto \frac{1-\alpha}{\alpha} I_\alpha\left(\rho_{A_1 A_2 \ldots A_L E}\, \Vert \otimes_{\ell\in[L]} \tau_{A_\ell} \right).
		\end{align}

    \item\label{item:concavity_p_sandwiched}
    (Concavity in $\psi$ \cite[Lemma 17]{LY21b})
    The following map is concave 
     and upper semi-continuous 
    on $\mathcal{S}(\mathsf{A})$:
    \begin{align}
        \rho_{A} \mapsto I_{\alpha}\left(B_1: B_2: \cdots B_L : R \right)_{\mathscr{N}_{A\to B_1 B_2\ldots B_L}(\psi_{AR}^\rho)},
    \end{align}
    where $\psi_{AR}^\rho$ is a purification of $\rho_A$ and $\mathsf{R} \cong \mathsf{A}$.
	\end{enumerate}
\end{proposition}
% \begin{proof}
% 	{\color{red}
% 		Item~\ref{item:derivative_sandwiched} requires a proof but it should follow from Li and Hao-Chung's soft covering paper.
		
% 		Item~\ref{item:additivity_sandwiched} requires a proof but it should follow from Li and Hao-Chung's convex splitting paper.
% 	}

% 	{\color{red} Still needs lower/upper semi-continuity.}
% \end{proof}

The following Lemmas~\ref{lemma:dimension_bound} and \ref{lemma:convexity} will be used in Section~\ref{sec:simulation} for broadcast channel simulation. We delay their proofs to Appendix~\ref{sec:lemmas}.

\begin{lemma}[Dimension bound] \label{lemma:dimension_bound}
	For any states $\rho_{ABC} \in \mathcal{S}(\mathsf{A} \otimes \mathsf{B} \otimes \mathsf{C})$, $\tau_{A} \in \mathcal{S}(\mathsf{A})$, and $\alpha >0$, we have
	\begin{align}
	I_\alpha\left(\rho_{ABC} \,\Vert\, \tau_A \right)
	\leq I_\alpha\left(\rho_{AB} \,\Vert\, \tau_A \right) + \frac{2\alpha}{\alpha-1} \log |\mathsf{C}|.
	\end{align}
\end{lemma}

\begin{lemma}[Convexity] \label{lemma:convexity}
	Let $L$ be any integer and $\mathcal{I}$ be any finite set.
	Let $\rho_{A_1A_2\ldots A_L E} := \sum_{i\in\mathcal{I}} p_i \rho_{A_1A_2\ldots A_L E}^{i} $ and $\tau_{A_\ell} =  \sum_{i\in\mathcal{I}} p_i \tau_{A_{\ell}}^{i} $, $\ell\in [L]$ be statistical mixtures of states for any $p_i>0$, $\sum_{i\in\mathcal{I}} p_i  = 1$.
	Then, the following holds for every $\alpha \geq \sfrac12$,
	\begin{align}
	\begin{split}
	I_\alpha( \rho_{A_1\ldots A_L E} \,\Vert \otimes_{\ell \in [L]}  \tau_{A_\ell} ) &\leq \sum\nolimits_{i\in\mathcal{I} } p_i I_\alpha\left(\rho_{A_1\ldots A_L E}^{i} \,\Vert \otimes_{\ell\in[L]} \tau_{A_\ell}^{i}\right) + L \cdot H(\{p_i\}_{i\in\mathcal{I}})
	\\ &\le \sum\nolimits_{i \in \mathcal{I} }p_i I_\alpha\left(\rho_{A_1\ldots A_L E}^{i} \,\Vert \otimes_{\ell\in[L]} \tau_{A_\ell}^{i}\right) + L \log |\mathcal{I}|.
	\end{split}
	\end{align}
	% % and
	% % \begin{align}
	% % \begin{split}
	% % I_\alpha\left( A_1 : \cdots : A_L \right)_\rho &\leq \sum\nolimits_{i\in\mathcal{I} } p_i I_\alpha\left( A_1 : \cdots : A_L \right)_{\rho^i} + L \cdot H(\{p_i\}_{i\in\mathcal{I}})
	% % \\ &\le \sum\nolimits_{i \in \mathcal{I} }p_i I_\alpha\left( A_1 : \cdots : A_L \right)_{\rho^i} + L \log |\mathcal{I}|.
	% % \end{split}
	% \end{align}	
	Here, $H(\{p_i\}_{i\in\mathcal{I}}) := - \sum_{i\in\mathcal{I}} p_i \log p_i$ denotes the Shannon entropy.
%	Moreover, 
\end{lemma}

% For any $r>0$, we define the error-exponent functions for a multipartite state $\rho_{A_1 A_2 \ldots A_L B}$, positive semi-definite $\tau_{A_1} \in \mathcal{S}(\mathsf{A}_1)$, $\ldots$, $\tau_{A_L} \in \mathcal{S}(\mathsf{A}_L)$, and a quantum broadcast channel $\mathscr{N}_{A\to B_1 B_2 \ldots B_L} $ as :
% \begin{align}
% E_r (\rho_{A_1 \ldots A_L B} \, \Vert \, \otimes_{\ell\in[L]} \tau_{A_\ell} ) &:= \sup_{\alpha \in [1,2]} \frac{\alpha-1}{\alpha} \left( r - I_\alpha(\rho_{A_1 \ldots A_L B} \, \Vert \, \otimes_{\ell\in[L]} \tau_{A_\ell} ) \right); \label{eq:error_exponent_tau} \\
% E_r(A_1 : \cdots A_L : B)_\rho &:= E_r (\rho_{A_1 \ldots A_L B} \, \Vert \, \otimes_{\ell\in[L]} \rho_{A_\ell} ) \, \Vert \, \rho_A); \label{eq:error_exponent_mutual}\\
% % E_r (\mathscr{N}_{A\to B_1 B_2 \ldots B_L}) &:= \sup_{\alpha \in [1,2]} \frac{\alpha-1}{\alpha} \left( r - I_\alpha( \mathscr{N}_{A\to B_1 B_2 \ldots B_L} ) \right).
% \end{align}
% \begin{align} \label{eq:error_exponent}
% E_{{r}}(A_1: A_2 : \cdots : A_L : B )_\rho 
% &:= 	\sup_{\alpha \in [1,2]} \frac{\alpha-1}{\alpha}\left(
% r  - I_\alpha\left(A_1 : A_2 : \cdots : A_L: B\right)_\rho
% \right); \\
% \end{align}

We introduce the error-exponent functions as the \emph{Fenchel--Legendre transform} of the above R\'enyi information quantities, i.e.~for any $r>0$,
\begin{align}
E_r (\rho_{A_1 A_2 \ldots A_L B} \, \Vert \, \otimes_{\ell\in[L]} \tau_{A_\ell}) &:= \sup_{\alpha \in [1,2]} \frac{\alpha-1}{\alpha} \left( r - I_\alpha(\rho_{A_1 A_2 \ldots A_L E} \, \Vert \, \otimes_{\ell\in[L]} \tau_{A_\ell}) \right); \label{eq:error_exponent_tau}\\
E_r(A_1 : A_2 : \cdots : A_L : E)_\rho &:= E_r (\rho_{A_1 A_2 \ldots A_L E} \, \Vert \otimes_{\ell\in[L]} \rho_{A_\ell}) \label{eq:error_exponent_marginal};\\
E_r (\mathscr{N}_{A\to B_1 B_2 \ldots B_L}) &:= \sup_{\alpha \in [1,2]} \frac{\alpha-1}{\alpha} \left( r - I_\alpha( \mathscr{N}_{A\to B_1 B_2 \ldots B_L} ) \right). \label{eq:error_exponent_channel}
\end{align}

We collect the known properties of the error-exponent function \cite{Hao-Chung, CH17, CHT19, CHDH2-2018, CG22, SGC22b}, which are consequences of the the properties of the R\'enyi information given in
Proposition~\ref{fact:properties_sandwiched}.
and the minimax theorem \cite[\S36]{Roc70}.
\begin{proposition}[Properties of error-exponent function] \label{fact:properties_exponent}
For any multipartite state $\rho_{A_1 A_2 \ldots A_L B} $ and any quantum broadcast channel $\mathscr{N}_{A\to B_1 B_2 \ldots B_L}$, the error-exponent function satisfies the following properties.
\begin{enumerate}[(i)]
	\item\label{item:positivity_exponent} (Positivity) For any $r>0$,
    \begin{align}
        E_r (\rho_{A_1 A_2 \ldots A_L B} \, \Vert \otimes_{\ell\in[L]} \tau_{A_\ell}) > 0 &\Longleftrightarrow r > I(\rho_{A_1 A_2 \ldots A_L B} \, \Vert \otimes_{\ell\in[L]} \tau_{A_\ell}); \\
        E_r (\mathscr{N}_{A\to B_1 B_2 \ldots B_L}) > 0 &\Longleftrightarrow r > I( \mathscr{N}_{A\to B_1 B_2 \ldots B_L} ).
    \end{align}

	\item\label{item:additivity_exponent} (Additivity) 
    For any integer $n \in \mathds{N}$ and $r>0$,
	\begin{align}
	E_{nr} \left(\rho_{A_1 A_2 \ldots A_L B}^{\otimes n} \, \Vert \otimes_{\ell\in[L]} \tau_{A_\ell}^{\otimes n} \right)
	= n E_{{r}}\left(\rho_{A_1 A_2 \ldots A_L B} \, \Vert \otimes_{\ell\in[L]} \tau_{A_\ell} \right). 
	\end{align}

 \item\label{item:minimax_exponnet} (A minimax identity and saddle-point)
    Provided that the underlying Hilbert spaces are all finite dimensional, for any $r>0$,
    there exist a saddle-point $(\psi,\alpha) \in \mathcal{S}(\mathsf{A}\otimes \mathsf{R}) \times[1,2] $\cite[\S 36]{Roc70} such that
    \begin{align}
    E_r(\mathscr{N}_{A\to B_1 B_2 \ldots B_L}) &= 
     \inf_{\psi_{AR}} E_r(B_1: B_2: \cdots B_L: R)_{(\mathscr{N}\otimes \id_R)(\psi)} \\
     &= \frac{\alpha-1}{\alpha} \left( r - I_\alpha(B_1:B_2\cdots B_L: R)_{(\mathscr{N}\otimes \id_R)(\psi)} \right).
    \end{align}
	% \begin{align}
	% \inf_{ \tau_{A_1}\in\mathcal{S}(\mathsf{A}_1), \ldots, \tau_{A_L}\in\mathcal{S}(\mathsf{A}_L) }
	% E_r(\rho_{A_1 \ldots A_L} \, \Vert\, \otimes_{\ell\in[L]} \tau_{A_\ell}  )
	% = E_{{r}}(A_1 : \cdots : A_L)_\rho .
	% \end{align}
	
	\item\label{item:limiting_exponnet} (Limiting behavior)
    Provided that the underlying Hilbert spaces are all finite dimensional, then    
	for any sequence $r_n := I\left(\mathscr{N}_{A\to B_1 B_2\ldots B_L}\right) + a_n$ satisfying  
    $a_n \downarrow 0$,
	we have
    \begin{align}
	\liminf_{n\to \infty}\frac{E_{r_n}(\mathscr{N})}{a_n^2} 
	&\geq \frac{1}{2V(\mathscr{N}_{A\to B_1 B_2\ldots B_L})}.
    \end{align}

\end{enumerate}

\end{proposition}
\noindent We delay the proof to Appendix~{\ref{sec:lemmas}}.

\begin{remark}
    Proposition~\ref{fact:properties_exponent}-\ref{item:minimax_exponnet} indicates that the error exponent of channel $E_r(\mathscr{N}_{A\to B_1 B_2 \ldots B_L})$ can be viewed as the error exponent $E_r(B_1: B_2: \cdots B_L: R)_{(\mathscr{N}\otimes \id_R)(\psi)}$
    for the output state with respect to the worst-case input $\psi_{AR}$.
\end{remark}

% \begin{proof}
% {\color{red} Item~\ref{item:minimax_exponnet} requires a proof.}

% Note that the map
% \begin{align}
% (\tau_{A_1}, \ldots, \tau_{A_L}, \alpha) \mapsto
% \frac{\alpha-1}{\alpha} \left( r - I_\alpha(\rho_{A_1  \ldots A_L} \,\Vert\, \otimes_{\ell \in [L]} \tau_{A_\ell} ) \right)
% \end{align}
% is  convex in $\tau_{A_1}\in\mathcal{S}(\mathsf{A}_1), \ldots, \tau_{A_L}\in\mathcal{S}(\mathsf{A}_L)$ by Proposition~\ref{fact:properties_sandwiched}-\ref{item:convexity_sandwiched} and concave in $\alpha \in (1,\infty)$ by Proposition~\ref{fact:properties_sandwiched}-\ref{item:concavity_sandwiched}. Hence, the assertion follows from Sion's minimax theorem.

% {\color{red} Still needs lower/upper semi-continuity.}

%\begin{align}
%\inf_{ \tau_{A_1}\in\mathcal{S}(\mathsf{A}_1), \ldots, \tau_{A_L}\in\mathcal{S}(\mathsf{A}_L) }
%\sup_{\alpha \in [1,2]} \frac{\alpha-1}{\alpha} \left( r - I_\alpha(\rho_{AE} \,\Vert\, \tau_A ) \right)
%\end{align}

% \end{proof}

%%%%%%%%%%%%%%%%%%%%%%%%%%%%%%%%

\section{Convex Splitting} \label{sec:convex_splitting}

In Ref.~\cite{CG23}, part of the authors established a one-shot error-exponent bound for \emph{unipartite} convex splitting, i.e.~for any density operators $\rho_{AE} \in \mathcal{S}(\mathsf{A\otimes E})$ and $\tau_{A} \in \mathcal{S}(\mathsf{A})$, and $M\in \mathbb{N}$,
\begin{align} \label{eq:error_unipartite_convex_splitting}
		% \begin{split}
		% \Delta_{M,K}(\rho_{AE} \,\Vert\, \tau_{A}) 
  \frac12\left\| \frac1M \sum\nolimits_{m\in[M]} \rho_{A_m E} \bigotimes\nolimits_{\bar{m}\in[M]\setminus \{m\}} \tau_{A_{\bar{m}}}
  - \tau_{A}^{\otimes M} \otimes \rho_E \right\|_1
    &\leq  2^{- E_{\log M}(\rho_{AE} \,\Vert\, \tau_A )},
		% \end{split}
	\end{align}
where $\rho_{A_m E} = \rho_{AE}$ and $\tau_{A_m} = \tau_A$ for all $m\in[M]$, and the error-exponent function $E_{\log M}$ is defined in Eq.~\eqref{eq:error_exponent_tau}. While this result is a neat application of complex interpolation, a straightforward generalization to the multipartite case does not give the R\'enyified quantities simultaneously. The key ingredient to bypass this difficulty is a \emph{mean-zero decomposition lemma} that will be introduced later in Lemma~\ref{lemma:tele} and Lemma~\ref{lemma:tele2}. We remark that a similar idea is independently proposed by Colomer Saus and Winter for deriving multipartite quantum decoupling theorems, termed the \emph{telescoping trick} in their work \cite{SW23}.

%%%%%%%%%%%%%%%%%%%%%%%%%%%%%%%%

\subsection{Bi-partite Convex Splitting} \label{sec:bipartite}

Let $\rho_{ABE}$ be a tripartite density operator in $\mathcal{S}(\mathsf{A\otimes B\otimes E})$, and let $\tau_{AB}$  be a bipartite density operator $\mathcal{S}(\mathsf{A\otimes B})$. Given integers $M$ and $K$, we define the density operator % \in \mathcal{S}(\mathsf{A}_1\cdots \mathsf{A}_M \mathsf{B}_1\cdots \mathsf{B}_K \mathsf{E})$ as
\begin{align} \label{eq:total_state}
	\begin{split}
	\omega_{A_1\ldots A_M B_1\ldots B_KE} := \frac{1}{MK}\sum_{(m,k)\in[M]\times[K]}
	\rho_{A_m B_kE} &\otimes \tau_{A_1} \otimes \tau_{A_2} \cdots \otimes \tau_{A_{m-1}} \otimes \tau_{A_{m+1}} \otimes \cdots \otimes \tau_{A_M}
	\\
	&\otimes \tau_{B_1} \otimes \tau_{B_2} \cdots \otimes \tau_{B_{k-1}} \otimes \tau_{B_{k+1}} \otimes \cdots \otimes \tau_{B_K},
%	\in \mathcal{S}(\mathsf{A_1\cdots A_M B_1\cdots B_K})
	\end{split}
\end{align}
where for each $m\in[M] $ and $k\in[K]$, we have the systems $A_m\cong A$, $B_k\cong B$ and the states
$\rho_{A_m B_k} = \rho_{AB}$, $\tau_{A_m} = \tau_A$, and $\tau_{B_k} = \tau_B$.
%Without loss of generality, we will further assume that the density operator $\tau_A$ is invertible; otherwise, we can always restrict the associated Hilbert space $\mathsf{A}$ to the support of $\tau_A$.
We use trace distance as the error criterion for bipartite convex splitting,
\begin{align}
	\begin{split} \label{eq:covering}
		\Delta_{M,K}(\rho_{ABE} \,\Vert\, \tau_{AB}) &:=
		\frac12\left\| \omega_{A_1\ldots A_M B_1\ldots B_KE} - \tau_{A}^{\otimes M} \otimes \tau_{B}^{\otimes K}\otimes \rho_E \right\|_1.
	\end{split}
\end{align}

Given $p\geq 1$, recall that the \emph{Kosaki's weighted $L_p$-norm} with respective to a density operator $\sigma$ \cite{Kos84, JX03}  is defined as,
\begin{align}
	\begin{split} \label{eq:norm}
	\left\|X\right\|_{p,\sigma} &:= \left( \Tr\left[ \left| \sigma^{\frac{1}{2p}} X \sigma^{\frac{1}{2p}} \right|^p \right] \right)^{\sfrac1p}
	\end{split}
\end{align}
and the associated \emph{noncommutative weighted $L_p$-space} is denoted as $L_p(\sigma)=
	 \left\{ X: \|X\|_{p, \sigma} < \infty  \right\}$.
%and the following \emph{noncommutative weighted $L_p$-space} with respective to $\rho_A\otimes \sigma_B$ as the set of operators with finite noncommutative weighted $L_p$-norm, i.e.~
%\begin{align}
%	L_{p, \rho_A\otimes \sigma_B} := \left\{ X: \|X\|_{p, \rho_A\otimes \sigma_B} < \infty  \right\}.
%\end{align}
For two positive operators $X$ and $Y$ %with $\mathtt{supp}(X)\subseteq\mathtt{supp}(Y)$,
we introduce the notation of \emph{non-commutative quotient for $X$ over $Y$} as\footnote{If $Y$ is not invertible, we then take the Moore--Penrose pseudo-inverse of $Y$ instead.}
\begin{align} \label{eq:quotient}
	\frac{X}{Y} := Y^{-\frac{1}{2}} X Y^{-\frac{1}{2}}.
\end{align}

Our approach (as in the unipartite convex splitting \cite{CG23}) is to formulate the error $\Delta_M\left( \rho_{ABE} \, \Vert\, \tau_{AB}\right)$ as the weighted $L_p$ norm and estimate it via complex interpolation. We refer to  \cite[Appendix B]{CG23} for a minimal introduction of complex interpolation needed for this paper and the readers are referred to  \cite{BL76} for more information on this topic . 
We start with rewriting the error quantity in Eq.~\eqref{eq:covering} using non-commutative $L_p$-norm.
Given a density operator $\tau_{AB}=\tau_A\ten \tau_B$, we define the $\tau$-preserving conditional expectation 
%and positive semi-definite operators $X_{AB} \geq 0$ and:
\begin{align}
	\label{eq:Theta}
&\mathrm{E}_{A} : \mathcal{B}(\mathsf{A})\to \mathcal{B}(\mathsf{A}),
\;X_{A} \mapsto \Tr_{A}\left[ \tau_{A} X_{A}  \right]\cdot \mathds{1}_{A}; \nonumber \\
	&\mathrm{E}_{B} :  \mathcal{B}(\mathsf{B})\to \mathcal{B}(\mathsf{B}), 
 \;X_{B} \mapsto \Tr_{B}\left[ \tau_{B} X_{B}  \right]\cdot \mathds{1}_{B}; \nonumber\\
&\mathrm{E}_{AB}:=\mathrm{E}_A\ten \mathrm{E}_B : \mathcal{B}(\mathsf{A\otimes B})\to \mathcal{B}(\mathsf{A\otimes B}).
\end{align}
The following map is an key object in the bipartite analysis.
\begin{align} 
\begin{split} \label{eq:tele}
&T: \mathcal{B}(\mathsf{A\otimes B})\to \mathcal{B}(\mathsf{A\otimes B}) \ , \ T:=\id_{AB}- \mathrm{E}_A-\mathrm{E}_B+\mathrm{E}_{AB}.
\end{split}
\end{align}
Here and in the following, $\mathrm{E}_A$ can be interpreted as $\mathrm{E}_A\ten \id_B$ and similar for $\mathrm{E}_B$. 
We will shortly see how the map $T$ comes into play. Let us first state the following lemma that estimates its norm between $L_p$-spaces.

\begin{lemma} \label{lemma:tele}
Let $\tau=\tau_{A}\ten \tau_{B}\ten \tau_{E} \in \mathcal{S}(\mathsf{A}\otimes \mathsf{B}\otimes \mathsf{E})$ be a tripartite product density operator and a map $T$ introduced in Eq.~\eqref{eq:tele}. Then
\[ \norm{T: L_1(\tau_{ABE})\to L_1(\tau_{ABE})}{}\le 4\pl, \pl \norm{T: L_2(\tau_{ABE})\to L_2(\tau_{ABE})}{}\le 2.\]
\end{lemma}
\begin{proof}We first argue for $L_1$. Note that $\mathrm{E}_A \mathrm{E}_B= \mathrm{E}_B \mathrm{E}_A= \mathrm{E}_{AB}$. By triangle inequality, it suffices to show both $\mathrm{E}_A$ and $\mathrm{E}_B$ are contraction. Indeed, consider the duality
$L_1(\tau)^*=L_\infty(\tau)$, where the pairing is given the $\tau$-inner product $$
\lan X,Y \ran_\tau=\Tr(X\tau^{1/2}Y\tau^{1/2})\pl .$$
Note that $\mathrm{E}_A$ is completely positive unital map, hence a contraction on $L_\infty(\tau)\cong \mathcal{B}(\mathsf{A\otimes B\otimes E})$ (equipped with usual operator norm). Also, as $\mathrm{E}_A$ is self-adjoint for $\tau$-inner product. Then
\[ \norm{ \mathrm{E}_A:L_1(\tau)\to L_1(\tau)}{}=\norm{ \mathrm{E}_A:=L_\infty(\tau)\to L_\infty(\tau)}{}=1. \]
The same argument applies for $\mathrm{E}_B$. Thus we prove that $\norm{T: L_1\to L_1}{}\le 4$ by triangle inequality.

For $L_2$, we note that $\mathrm{E}_A$ is a $\tau$-preserving conditional expectation. That is, 
$\Tr[\tau \mathrm{E}_A(X)]=\Tr[\tau X]$ for all $X\in\mathcal{B}(\mathsf{A}\ten \mathsf{B}\ten \mathsf{E})$,
and $\mathrm{E}_A^2=\mathrm{E}_A$ is a unital complete positive idempotent. In particular, $\mathrm{E}_A$ is a projection on $L_2(\tau)$. As $\mathrm{E}_B$ is a projection commute with $\mathrm{E}_B$, we have 
\[  X-\mathrm{E}_A(X)-\mathrm{E}_B(X)+\mathrm{E}_{AB}(X)=(X-\mathrm{E}_A(X))-\mathrm{E}_B(X-\mathrm{E}_{A}(X))\ .\]
Hence
\[ \norm{X}{L_2(\tau)}\ge \norm{X-\mathrm{E}_A(X)}{L_2(\tau)}\ge \norm{\mathrm{E}_B(X)-\mathrm{E}_{AB}(X)}{L_2(\tau)}.\]
The second assertion follows from triangle inequality for $L_2(\tau)$.
\end{proof}

In the following, we use the short notation $A^M:=A_1\cdots A_M\cong A^{\ten M}$ and $B^K:=B_1\cdots B_K\cong B^{\ten K}$.
Given $(m,k)\in [M]\times [K]$, we define the following maps
\begin{align*}
	\pi_{m,k}: &\;\mathcal{B}(\mathsf{A} \otimes \mathsf{B})\to  \mathcal{B}(\mathsf{A}^{M} \otimes \mathsf{B}^{K})\\
&\;\;X_{AB} \mapsto \mathds{1}_{A_1 B_1}\otimes \cdots \otimes \mathds{1}_{A_{m-1} B_{k-1}} \otimes  X_{A_m B_k} \otimes \mathds{1}_{A_{m+1} B_{k+1}} \otimes \cdots \otimes \mathds{1}_{A_M B_K};\\
\mathrm{E}_{A_m} : &\;X_{A_m} \mapsto \Tr_{A_m}\left[ \rho_{A_m} X_{A_m}  \right]\cdot \mathds{1}_{A_m};\\
	\mathrm{E}_{B_k} : &\;X_{B_k} \mapsto \Tr_{B_k}\left[ \rho_{B_k} X_{B_k}  \right]\cdot \mathds{1}_{B_k};\\
	\Theta:=
	&\;\frac{1}{MK}\sum_{(m,k)\in[M]\times [K]} \pi_{m,k}
\end{align*}
Here, the condition expectation $\mathrm{E}_{A_m}$ is only acting on the ${A}_m$ system while other systems remain unchanged; similarly for $\mathrm{E}_{B_k}$. 
It is clear that
\[ \pi_{m,k}\mathrm{E}_A=\mathrm{E}_{A_m}\pi_{m,k}\ , \pi_{m,k}\mathrm{E}_B=\mathrm{E}_{B_k}\pi_{m,k}\ ,  \pi_{m,k}\mathrm{E}_{AB}=\mathrm{E}_{A_mB_k}\pi_{m,k} \pl.\]
Take $\tau_{ABE}=\tau_A\ten \tau_B\ten \sigma_E$ for an arbitrary density $\sigma_E\in \mathcal{S}(\mathsf{E})$. We have
\begin{align}
		\Delta_{M,K}\left( \rho_{ABE}\,\Vert\,\tau_{AB} \right) &=
		\frac12\left\| \omega_{A_1\ldots A_M B_1\ldots B_KE} - \tau_{A}^{\otimes M} \otimes \tau_{B}^{\otimes K}\otimes \rho_E \right\|_1
\\ 
&= \frac12\left\| \frac{\omega_{A_1\ldots A_M B_1\ldots B_KE}}{\tau_{A^MB^KE}} - \mathds{1}_{A^MB^K}\ten \frac{\rho_E}{\sigma_E} \right\|_{L_1(\tau_{A^MB^KE})}
\\ 
&= \frac12\left\| \frac{1}{MK}\sum_{m,k}\pi_{m,k}\left(\frac{\rho_{ABE}}{\tau_A\ten \tau_B\ten \sigma_E}\right) -\mathds{1}_{A^MB^K}\ten\frac{\rho_E}{\sigma_E} \right\|_{L_1(\tau_{A^MB^KE})}
		\\ 
  &= \frac12\left\| \frac{1}{MK}\sum_{m,k}\pi_{m,k}\left(\frac{\rho_{ABE}}{\tau_A\ten \tau_B\ten \sigma_E}-\mathds{1}_{AB}\ten\frac{\rho_E}{\sigma_E}\right)  \right\|_{L_1(\tau_{A^MB^KE})}
\\ 
&= \frac12\left\| \frac{1}{MK}\sum_{m,k}\pi_{m,k}\circ (\id_{AB}-\mathrm{E}_{AB}) \left(\frac{\rho_{ABE}}{\tau_A\ten \tau_B\ten \sigma_E}\right)  \right\|_{L_1(\tau_{A^MB^KE})}
\\ 
&= \frac12\left\| \Theta\circ (\id_{AB}-\mathrm{E}_{AB}) \left(\frac{\rho_{ABE}}{\tau_A\ten \tau_B\ten \sigma_E}\right)  \right\|_{L_1(\tau_{A^MB^KE})}. \label{eq:rewrite}
\end{align}

A key observation is that, for $X\equiv X_{ABE}=\frac{\rho_{ABE}}{\tau_A\ten \tau_B\ten \sigma_E}$, we can decompose in Eq.~\eqref{eq:rewrite} into the following three terms:
\begin{align} \label{eq:decomposition}
(\id_{AB}-\mathrm{E}_{AB})(X)=X-\mathrm{E}_{AB}(X)=(X-\mathrm{E}_AX-\mathrm{E}_BX+\mathrm{E}_{AB}X)+(\mathrm{E}_AX-\mathrm{E}_{AB}X)+(\mathrm{E}_BX-\mathrm{E}_{AB}X)\pl.
\end{align}
Let us evaluate one of them:
\begin{align*}
&\left\| \Theta(\mathrm{E}_AX-\mathrm{E}_{AB}X)  \right\|_{L_1(\tau_{A^MB^KE})}
\\ 
&=\norm{\Theta\left( \mathds{1}_{A}\ten \frac{\rho_{BE}}{\tau_B\ten \sigma_E}-\mathds{1}_{AB}\ten \frac{\rho_{E}}{ \sigma_E}\right) }{L_1(\tau_{A^MB^KE})}
\\
&= \norm{ \frac{1}{MK}\sum_{m,k} 1_{A^M}\ten \frac{\rho_{B_kE}}{\tau_B\ten \sigma_E}\otimes 1_{B_1\cdots B_{k-1}B_{k+1}\cdots B_K}  - \mathds{1}_{AB} \ten \frac{\rho_E}{\sigma_E}}{L_1(\tau_{A^MB^KE})}
\\
&=\norm{ \frac{1}{K}\sum_{k}  \tau_{A^M}\otimes \rho_{B_kE}\otimes \tau_{B_1} \otimes \tau_{B_2} \cdots \otimes \tau_{B_{k-1}} \otimes \tau_{B_{k+1}} \otimes \cdots \otimes \tau_{B_K} - \tau_{A^MB^K} \ten \rho_E}{1}
\\
&=\norm{ \tau_{A^M}}{1}\cdot \norm{ \frac{1}{K}\sum_{k}  \rho_{B_kE}\otimes \tau_{B_1} \otimes \tau_{B_2} \cdots \otimes \tau_{B_{k-1}} \otimes \tau_{B_{k+1}} \otimes \cdots \otimes \tau_{B_K} - \tau_{B^K} \ten \rho_E}{1}
\\
&= 2\Delta_{K}\left( \rho_{BE} \,\Vert\, \tau_B\right).
\end{align*}
which is exactly the unipartite  convex-splitting error for density operators $\rho_{BE}$ and $\tau_B$. 
Similarly, 
\begin{align*}
\left\| \Theta(\mathrm{E}_BX-\mathrm{E}_{AB}X)  \right\|_{L_1(\tau_{A^MB^KE})}
&=\Delta_{K}\left( \rho_{AE} \, \Vert \, \tau_A\right).
\end{align*}
Given the unipartite convex-split result in Eq.~\eqref{eq:error_unipartite_convex_splitting}, it suffices to deal with the term
\[ \Theta(X-\mathrm{E}_AX-\mathrm{E}_BX+\mathrm{E}_{AB}X)=\Theta\circ T(X) \pl. \]

Using Lemma~\ref{lemma:tele} for estimating the norm of map $T$, we obtain a key technical lemma to bound the norm $\Theta\circ T$ as follows.
\begin{lemma}[Map norm of bipartite biconvex splitting] \label{lemma:key}
	Given any $M,K\in \mathds{N}$ and density operator $\tau_{ABE}=\tau_A\ten \tau_B\ten \sigma_E$, let $\Theta$ and $T$ be the map defined as in Eq.~\eqref{eq:Theta}. Then, for every density operator $\sigma_E$, we have
	\begin{align}
		\left\| \Theta\circ T : L_{p}(\tau_{ABE}) \to L_{p}(\tau_{A^MB^KE}) \right\| \leq  2^{\frac{3}{p}-1} (MK)^{ \frac{1-p}{p} }, \quad \forall \, p\in[1,2].
	\end{align}
\end{lemma}
\begin{proof}%[Proof of Lemma~\ref{lemma:key}]
We first note that for any $p\in [1,\infty]$ and every $(m,k)\in [M]\times [K]$,
	\begin{align}
		 \pi_{m,k}: L_{p}(\tau_{ABE}) \to L_{p}(\tau_{A^MB^KE})
	\end{align}
is an isometry. Indeed,
\begin{align*} \norm{\pi_{m,k}(X_{ABE})}{L_{p}(\tau_{A^MB^KE})} &=\norm{\mathds{1}_{A^{M/\{m\} }B^{K/\{k\}}}\ten X_{A_kB_mE}}{L_{p}(\tau_{A^MB^KE})}
\\ &=\norm{ \tau_{A^{M/\{m\} }B^{K/\{k\}}}^{\frac{1}{p}}\ten \tau_{A_mB_k}^{\frac{1}{2p}}X_{A_kB_mE}\tau_{A_mB_k}^{\frac{1}{2p}}}{p}
\\ &=\norm{ \tau_{ABE}^{\frac{1}{2p}}X_{ABE}\tau_{ABE}^{\frac{1}{2p}}}{p}=\norm{X_{ABE}}{L_{p}(\tau_{ABE})}
\end{align*}
For $p=1$, by triangle inequality we have
	\begin{align}
		&\left\| \Theta: L_{1}(\tau_{ABE}) \to L_{1}(\tau_{A^MB^KE}) \right\|\leq \frac{1}{MK}\sum_{m,k} \| \pi_{m,k}:  L_{1}(\tau_{ABE}) \to L_{1}(\tau_{A^MB^KE}) \|\le 1, \label{eq:L1}
	\end{align}
and hence by Lemma \ref{lemma:tele}
\[ \left\| \Theta\circ T: L_{1}(\tau_{ABE}) \to L_{1}(\tau_{A^MB^KE}) \right\|\le 4\pl.\]
For $p=2$, we note that for every $(m,k)\neq (m',k')\in [M]\times [K]$, the range of $\pi_{m,k}\circ T$ and $\pi_{m',k'}\circ T$ are mutually orthogonal in $L_2( \tau_{A^MB^KE})$. 
	Indeed, without loss of generosity, we assume $k\neq k'$, and for any $X_{ABE},Y_{ABE}$ denote $\mathring{X}:=T(X)=X-\mathrm{E}_AX-\mathrm{E}_BX+\mathrm{E}_{AB}X$. Then 
	\begin{align}
		&\langle \pi_{m,k}(\mathring{X}), \pi_{m',k'}(\mathring{Y}) \rangle_{\tau_{A^MB^KE}} \\
&=  \lan \mathds{1}_{A^{[M]/\{m\}}B^{[K]/{k}}} \otimes \mathring{X}_{A_m B_kE}, \mathds{1}_{A^{[M]/\{m'\}}B^{[K]/{k'}}} \otimes \mathring{X}_{A_{m'} B_{k'}E}\ran_{\tau_{A^MB^KE}}\\
&=  \Tr\Big(\tau_{A^MB^KE}^{\frac{1}{2}}\mathring{X}_{A_m B_kE} \tau_{A^MB^KE}^{\frac{1}{2}}\mathring{Y}_{A_{m'} B_{k'}E}\Big)\\
&=  \Tr\Big(\tau_{A_mA_{m'}B_kB_{k'}E}^{\frac{1}{2}}\mathring{X}_{A_m B_kE} \tau_{A_mA_{m'}B_kB_{k'}E}^{\frac{1}{2}}\mathring{Y}_{A_{m'} B_{k'}E}\Big)\\
&=  \Tr\Big(\tau_{A_mA_{m'}E}^{\frac{1}{2}}\mathrm{E}_{B_{k}}(\mathring{X}_{A_mB_kE}) \tau_{A_mA_{m'}E}^{\frac{1}{2}}\mathrm{E}_{B_{k'}}(\mathring{X}_{A_{m'}B_{k'} E})\Big)\\
&= \Tr\Big(\tau_{A_mA_{m'}E}^{\frac{1}{2}}(\pi_{m',k'}\circ \mathrm{E}_B(\mathring{X})) \tau_{A_mA_{m'}E}^{\frac{1}{2}}(\pi_{m',k'}\circ \mathrm{E}_B(\mathring{Y}))\Big)
\\
&= 0
\end{align}
because $\mathrm{E}_B\mathring{X}=\mathrm{E}_B\mathring{Y}=0 $. In the second last inequality, we used the commutation relation $\mathrm{E}_{B_{k}}\pi_{m,k}=\pi_{m,k}\mathrm{E}_B$.

By the orthogonality, for any $X_{ABE}$
	\begin{align}
		\left\| \frac{1}{MK} \sum_{m,k} \pi_{m,k} \left( \mathring{X} \right) \right\|_{L_2(\tau_{A^MB^KE})}^2
&\overset{\textnormal{(a)}}{=} \frac{1}{MK^2} \sum_{m,k}  \left\| \pi_{m,k} \left( \mathring{X} \right) \right \|_{L_2(\tau_{A^MB^KE})}^2
\\
		&\overset{\textnormal{(b)}}{=} \frac{1}{MK^2} \sum_{m,k}  \left\| \mathring{X} \right\|_{L_2(\tau_{ABE})}^2\\
		&= \frac{1}{MK}  \left\| \mathring{X} \right\|_{L_2(\tau_{ABE})}^2 %\label{eq:L2}
\\
  &\overset{\textnormal{(c)}}{\le } \frac{4}{MK}  \left\| X \right\|_{L_2(\tau_{ABE})}^2 \label{eq:L2}
	\end{align}
	where (a) follows from orthogonality;
	(b) is because $\pi_{m,k}$ are isometry;
	and the last inequality (c) follows by the Lemma \ref{lemma:tele}.
	The case of general $p\in[1,2]$ follows from complex interpolation for $\frac{1}{p}=\frac{1-\theta}{1}+\frac{\theta}{2}$, 
 $\theta = \frac{2(p-1)}{p} \in [0,1]$,
	\begin{align}
		\left\| \Theta\circ T: L_{p} \to L_{p}  \right\|
		&\leq \left\| \Theta: L_{1} \to L_{1}\right\|^{1-\theta} \left\| \Theta: L_{2} \to L_{2}\right\|^\theta =  2^{\frac{3-p}{p} }(MK)^{ \frac{1-p}{p} },
	\end{align}
	which completes the proof.
\end{proof}
\begin{remark}
The above estimate holds for asymmetric case $L_{p,\gamma}(\tau)$ with norm $\norm{X}{L_p(p,\gamma,\tau)}=\left\| \tau^{\frac{1-\gamma}{p}} X\tau^{\frac{\gamma}{p}} \right\|_{p}$ for any $\gamma\in [0,1]$, although this point will not be used in our discussion.
\end{remark}

%and the following \emph{noncommutative weighted $L_p$-space} with respective to $\rho_A\otimes \sigma_B$ as the set of operators with finite noncommutative weighted $L_p$-norm, i.e.~
%\begin{align}
%	L_{p, \rho_A\otimes \sigma_B} := \left\{ X: \|X\|_{p, \rho_A\otimes \sigma_B} < \infty  \right\}.
%\end{align}

%Given $M\in\mathds{N}$ and density operators $\rho_{AB}$ and $\tau_A > 0$,
%our key construction is the following formulation that for density operator $\sigma_B$:
%\begin{align} \label{eq:formulation}
%	\left\| \omega_{A_1\ldots A_M B} - \tau_A^{\otimes n}\otimes \rho_{B} \right\|_1
%	&= \left\| \Theta\left( \frac{\rho_{AB}}{\tau_A\otimes \sigma_B} \right) \right\|_{1,\gamma, \tau_A^{\otimes M}\otimes \sigma_B}.
%\end{align}
%Note that the above identity is for any $\gamma \in [0,1]$ with corresponding noncommutative quotient $\frac{\rho_{AB}}{\tau_A\ten \sigma_B}$.
Combining with unipartite convex splitting $\Delta_M(\rho_{AE}\Vert\tau_A),\Delta_K(\rho_{BE}\Vert\tau_B)$ given in Eq.~\eqref{eq:error_unipartite_convex_splitting}, we have
\begin{shaded_theorem}[$2$-party convex splitting] \label{theorem:bipartite_convex_splitting}
	For any density operators $\rho_{ABE} \in \mathcal{S}(\mathsf{A\otimes B\otimes E})$ and $\tau_{AB} \in \mathcal{S}(\mathsf{A\otimes B})$, and $M,K\in \mathbb{N}$,  we have
    \begin{align} \label{eq:error_bipartite_convex_splitting}
		\begin{split}
		\Delta_{M,K}(\rho_{ABE} \,\Vert\, \tau_{AB}) &\leq 2\cdot \,2^{ -  E_{\log MK}(\rho_{ABE} \,\Vert\, \tau_A \otimes \tau_B ) }
        +  2^{- E_{\log M}(\rho_{AE} \,\Vert\, \tau_A )}
        +  2^{- E_{\log K}(\rho_{BE} \,\Vert\, \tau_B )},
		\end{split}
	\end{align}
    where the error exponents are defined in Eq.~\eqref{eq:error_exponent_tau}.
	% \begin{align} \label{eq:error_bipartite_convex_splitting}
	% 	\begin{split}
	% 	\Delta_{M,K}(\rho_{ABE} \,\Vert\, \tau_{AB}) &\leq 4 \,\mathrm{e}^{ - \sup_{\alpha \in [1,2]} \frac{\alpha-1}{\alpha} \left( \log MK -  I_\alpha(\rho_{ABE} \,\Vert\, \tau_A \otimes \tau_B ) \right)} \\
	% 	&+2 \,\mathrm{e}^{ - \sup_{\alpha \in [1,2]} \frac{\alpha-1}{\alpha} \left( \log M -  I_\alpha(\rho_{AE} \,\Vert\, \tau_A ) \right)}
	% 	+2 \,\mathrm{e}^{ - \sup_{\alpha \in [1,2]} \frac{\alpha-1}{\alpha} \left( \log K -  I_\alpha(\rho_{BE} \,\Vert\, \tau_B ) \right)}.
	% 	\end{split}
	% \end{align}
	% Here, the sandwiched R\'enyi divergence is defined in \eqref{eq:sandwiched},
 % \cite{MDS+13, WWY14} is defined as:
	% \begin{align}
	% D_\alpha(\rho\,\Vert\,\sigma) &:= \frac{1}{\alpha-1} \log \Tr\left[ \left( \sigma^{\frac{1-\alpha}{2\alpha}} \rho \sigma^{\frac{1-\alpha}{2\alpha}} \right)^\alpha \right],
	% \end{align}
	% and the generalized sandwiched R\'enyi information is defined in \eqref{eq:sandwiched_generalized}.
% Note that $\sigma_E$ in the three infimums above can be different.
	Moreover, the exponents are all  negative if and only if,
	\begin{align} \label{eq:region_original}
		\begin{cases}
		&\log MK > I(\rho_{ABE} \,\Vert\, \tau_A \otimes \tau_B ), \\
	&\log M > I(\rho_{AE} \,\Vert\, \tau_A ),\\
	&\log K > I(\rho_{BE} \,\Vert\, \tau_B ).\\
	\end{cases}
	\end{align}

\end{shaded_theorem}

\begin{remark}
Theorem~\ref{theorem:bipartite_convex_splitting} holds for infinite-dimensional Hilbert spaces $\mathsf{A}$, $\mathsf{B}$, and $\mathsf{E}$ as well.
\end{remark}

%%%%%%%%%%%%%%%%%%%%%%%%%%%%%%%%

\subsection{Multipartite Convex Splitting} \label{sec:multivariate}

In this section, we derive multipartite convex splitting. Following the ``pedestrian'' argument for the bipartite case, we will focus more on illustrating the mathematical structure of convex splitting.
Let $A_1,\cdots, A_L$ be $L$ systems. For each $\ell\in [L]$, we denote $A_\ell^{M_{\ell}}:=A_{\ell,1}\cdots A_{\ell,M_\ell}\cong A_\ell^{\ten M_{\ell}}$    as  $M_\ell$ copies of system $A_\ell$. Given a multipartite state $\rho_{A_1 \ldots A_L E}$, a product state $\tau_{A_1 \ldots A_L}=\tau_{A_1}\ten \cdots \ten \tau_{A_L}$, and integers $M_1$, $M_2$, $\ldots$, $M_L$, we define 
\begin{align} \label{eq:total_state_multipartite}
	\omega_{A^{M_1} \cdots A^{M_L} E} := \frac{1}{M_1 \cdots M_L}\sum_{m_\ell \in [M_\ell], \, \ell\in[L]}
	\rho_{A_{ 1, m_1} \cdots A_{ L, m_L} E} \left(\bigotimes_{  \bar{m}_{\ell} \in [M_{\ell}]/\{m_\ell\}, \,
		\ell \in [L] } \tau_{ A_{\ell, \bar{m}_{\ell} } },\right)
\end{align}
where for each $\ell\in [L]$,  $m_\ell,\bar{m}_\ell \in [M_L]$,
we let
$\rho_{A_{ 1, m_1} \cdots A_{ L, m_L} E} = \rho_{A_1 \ldots A_L E}$ and
$\tau_{A_{\ell, \bar{m}_{\ell} }} = \tau_{A_\ell}$. The error for the $L$-partite convex splitting is
\begin{align}
	\begin{split} \label{eq:covering_mulipartite}
		\Delta_{M_1, \cdots, M_L}\left( \rho_{A_1 \ldots A_L E} \,\Vert\, \tau_{A_1 \ldots A_L}  \right) &:=
		\frac12\left\| \,\omega_{A^{M_1} \cdots A^{M_L} E} - \bigotimes_{\ell\in[L]}\tau_{A_\ell}^{\otimes M_\ell} \otimes \rho_{E}\, \right\|_1.
	\end{split}
\end{align}

The key argument in the multipartite case is the decomposition map as we introduced in Eq.~\eqref{eq:tele} for the bipartite case.
% a telescoping trick independently discovered by by Colomer Saus and Winter \cite{CSW}. 
For a set $S=\{j_1,\cdots, j_k\}\subseteq [L]$, we introduce the following notation
\begin{align*} &\rho_{A_{S} B}=\rho_{A_{j_{1}}\ldots A_{j_k} E}\pl, \;  \tau_{A_{S}}=\tau_{A_{j_{1}}}\otimes \cdots \otimes \tau_{A_{j_k} }\pl , \\ \pl  
&M_{S}=\prod\nolimits_{\ell\in {S}}M_{j}\pl.
\end{align*}
Given the product density operator $\tau_{A_{[L]}}=\tau_{A_1}\ten \cdots \ten \tau_{A_{L}}$, we define the $\tau$-preserving conditional expectation
\begin{align*}
&\mathrm{E}_{A_\ell}: \mathcal{B}(A_\ell) \to \mathcal{B}(A_\ell)\pl, \pl \mathrm{E}_{A_\ell}(X)=\Tr(X\tau_{A_{\ell}}) \mathds{1}_{A_{\ell}} \\
&\mathrm{E}_{{S}}:=\prod\nolimits_{\ell\in {S}} \mathrm{E}_{A_{\ell}} : \mathcal{B}(A_{S}) \to \mathcal{B}(A_{S})\pl.
\end{align*}
For a subset $S\subseteq [L]$ and its complement $S^c=[L]\setminus {S}$, we define the \emph{mean-zero map}
\begin{align}
&T_{{S}}:   \mathcal{B}(A_{[L]}) \to \mathcal{B}(A_{[L]})\pl, \pl \\
&T_{{S}}:=\sum\nolimits_{{S}^c\subseteq  R \subseteq [L]} (-1)^{|R\setminus {S}^c|} \mathrm{E}_{ R }
\\
&\quad\; =\mathrm{E}_{{S}^c}\left( \sum\nolimits_{  R' \subseteq {S}} (-1)^{|R'|} \mathrm{E}_{R'}\right) \pl.
\end{align}
Here and in the following, the map $\mathrm{E}_{R}:=\mathrm{E}_{R}\ten \id_{R^c}$ (respectively for  $T_{R}$) can be interpreted as $\mathrm{E}_R= \prod_{\ell\in R} \mathrm{E}_{A_{\ell}}$ on $A_R$ and identity map on other systems $A_{R^c}$. 

The key property of the mean-zero map is that $T_{S}$ decomposes an element map $X$ into its “${S}$-mean-zero” part $T_{S}(X)$, which i) supported on $A_S$ (identity on other systems) and ii) all partial condition expectations (e.g.~a weighted partial trace) from $S$ are zero. We remark that a similar idea is independently proposed by Colomer Saus and Winter for deriving multipartite quantum decoupling theorems, termed the \emph{telescoping trick} in their work \cite{SW23}.

\begin{lemma}[Mean-zero decomposition] \label{lemma:tele2} Let ${S}\subseteq [L]$ be a subset and  $\ell\in [L]$. Then\begin{enumerate}
\item[i)] $\mathrm{E}_\ell \circ T_{S}=T_S$ if $\ell\notin S$ and $\mathrm{E}_\ell \circ T_{S}=0$ if $\ell\in S$.
\item[ii)] We have
\begin{align} \mathrm{E}_{{S}^c} X= \sum\nolimits_{R\subseteq {S}
} T_{R} ( X).
\end{align}
In particular,
\begin{align} X-\mathrm{E}_{[L]} X= \sum\nolimits_{\varnothing\neq {S}\subseteq[L] }T_{{S}}(X).
\end{align}
\end{enumerate}
\end{lemma}
\begin{proof}By the definition of $T_S$, if $\ell\in S^c$, $\mathrm{E}_\ell \mathrm{E}_{{S}^c}=\mathrm{E}_{{S}^c}$,
\begin{align*}
\mathrm{E}_\ell\circ T_{{S}}&= \mathrm{E}_\ell \mathrm{E}_{{S}^c}\left( \sum_{  R \subset {S}} (-1)^{|R|} \mathrm{E}_{R}\right) =\mathrm{E}_\ell \mathrm{E}_{{S}^c}\left( \sum_{  R \subset {S}} (-1)^{|R|} \mathrm{E}_{R}\right) =T_{{S}}
\end{align*}
If $\ell\in S$,
\begin{align*}
\mathrm{E}_\ell\circ T_{{S}}&= \mathrm{E}_\ell \mathrm{E}_{{S}^c}\left( \sum_{  R \subset {S}} (-1)^{|R|} \mathrm{E}_{R}\right) 
\\ 
&=\mathrm{E}_{{S}^c} \mathrm{E}_\ell\left( \sum_{  \ell \in R \subset {S}} (-1)^{|R|} \mathrm{E}_{R}+ \sum_{  \ell \notin  R' \subset {S}} (-1)^{|R'|} \mathrm{E}_{R'}\right)
\\ 
&=\mathrm{E}_{{S}^c}\left( \sum_{  \ell \in R \subset {S}} (-1)^{|R|} \mathrm{E}_{R}+ \sum_{  \ell \notin  R' \subset {S}} (-1)^{|R'|} \mathrm{E}_{R'\cup \{\ell \}}\right)
\\ 
&=\mathrm{E}_{{S}^c}\left( \sum_{  \ell \in R \subset {S}} (-1)^{|R|}+(-1)^{|R\setminus\{\ell\}|} \mathrm{E}_{R}\right)=0.
\end{align*}
The ii) is a direct application of Fubini's theorem, i.e.
\begin{align*}
    \sum_{R\subset {S}} T_{R}&=  \sum_{R\subset {S}} \sum_{ R^c\subset  O \subset [L]} (-1)^{|O\setminus R^c|} \mathrm{E}_{O}\\
    &=\sum_{{S}^c\subset R^c} \sum_{ R^c\subset  O \subset [L]} (-1)^{|O\setminus R^c|} \mathrm{E}_{O}\\
   &= \sum_{{S}^c\subset  O \subset [L]}\sum_{k=0}^{|O\setminus {S}^c|} \binom{|O\setminus {S}^c|}{k} (-1)^{k} \mathrm{E}_{O} \\
   &= \sum_{{S}^c\subset  O \subset [L]} (1-1)^{|O\setminus{S}^c|} \mathrm{E}_{O} \\
   &= \mathrm{E}_{{S}^c}
\end{align*}
because the only non-zero coefficient is when $O={S}^c$.  The second assertion follows from $T_{\varnothing}=E_{[L]}$ and
\begin{align*}
  X-\mathrm{E}_{[L]} X=\mathrm{E}_{\varnothing} X-\mathrm{E}_{[L]} X=  \sum_{{S}\subset [L]} T_{{S}} ( X) -T_{\varnothing} (X)=  \sum_{\varnothing\neq {S}\subset[L] }T_{{S}}(X).
\end{align*}
	\vspace{-\baselineskip}
\end{proof}

By triangle inequality, the norm of the mean-zero map can be estimate as follows.
\begin{lemma}\label{lemma:multitele}
For any ${S}\subset [L]$ and $p\in [1,\infty]$
\[\norm{T_{S}:L_p(\tau_{A_{[L]}}\ten \sigma_E) \to L_p(\tau_{A_{[L]}}\ten \sigma_E)}{}\le 2^{|{S}|} \pl. \]
\end{lemma}
\begin{proof} It suffices to note that for each $R\subset [L]$, the $\tau$-preserving conditional expectation  $\mathrm{E}_R$ is a contraction on $L_p(\tau)$ for all $p\in [1,\infty]$.
\end{proof}
\begin{remark}
The $L_2$ estimate for $T_{S}$ can be at least improved to 
\[ \norm{T_{S}:L_2(\tau_{A_{[L]}}\ten \sigma_E) \to L_2(\tau_{A_{[L]}}\ten \sigma_E)}{}\le 2^{|{S}|-1}.\]
\end{remark}

We now discussing the multipartite convex splitting map. Given a multi-index ${\bf m}=(m_1,\cdots, m_\ell)\in [M_1]\times \cdots [M_L]$, we introduce the short notation
\begin{align}
&A_{\bm}=A_{1,m_1}\cdots A_{L,m_L}\pl, \pl A_{\bm^c}=\bigotimes_{ \bar{m}_\ell\in [M_{\ell}]\setminus\{m_{\ell}\}, \ell \in [L] }A_{\ell,\bar{m}_\ell}\pl, \\
&A_\bm\ten A_{\bm^c}= \bigotimes_{m_\ell\in [M_\ell],\ell \in [L]} A_{\ell,m_\ell}\cong A_1^{\ten M_1}\cdots A_L^{\ten M_L}.\end{align}
Define the following map
\begin{align*}
	\pi_{\bm}: &\;\mathcal{B}(\mathsf{A_1\cdots A_L E})\to  \mathcal{B}(\mathsf{A_1^{M_1}\cdots A_L^{M_L}E})\  \\ 
 &\; X_{A_1\cdots A_L E} \mapsto   X_{A_{\bm } E} \otimes  \mathds{1}_{A_{\bm}^c}; \\
\mathrm{E}_{A_{\ell, m_\ell }} : &\;\mathcal{B}(A_{\ell, m_\ell })\to \mathcal{B}(A_{\ell, m_\ell })\pl\\
&\;X \mapsto \Tr\left[ \tau_{A_\ell} X  \right]\cdot \mathds{1}_{A_{\ell,m_\ell}};\\
	\Theta_{[L]}:=
	&\;\frac{1}{M_{[L]}}\sum_{\bm\in [M_1]\times \cdots \times [M_L]} \pi_{\bm}.
\end{align*}

\begin{lemma}\label{lemma:contraction}
For any $\bm$, $\pi_{\bm}$ is an isometry from $L_p(\tau_{A_1\cdots A_L}\ten \sigma_E)$ to $L_p\left(\tau_{A_1^{M_1}\cdots A_L^{M_L}}\ten \sigma_E\right)$ for all $p\in [1,\infty]$. Thus, by triangle inequality,
\begin{align*}
\norm{\Theta_{[L]}: L_p(\tau_{A_1\cdots A_L}\ten \sigma_E)\to L_p(\tau_{A_1^{M_1}\cdots A_L^{M_L}}\ten \sigma_E)}{}\le 1.
\end{align*}
\end{lemma}
\begin{proof}The proof is similar to the bipartite case in Lemma \ref{lemma:key}
\end{proof}

Similar to the bipartite case, the condition expectation $\mathrm{E}_{A_{\ell, m_\ell }}$ is only acting on the ${A}_{\ell, m_\ell}$ system  and 
\[ \pi_{\bm}\mathrm{E}_{A_\ell}=\mathrm{E}_{A_{\ell,m_\ell}}\pi_{\bm}\ ,
\;\pi_{\bm}\mathrm{E}_{A_{S}}=\mathrm{E}_{A_{{S},\bm_{S}}}\pi_{\bm},\]
where $A_{S,\bm_S}=\prod_{\ell \in S} A_{\ell, m_\ell}$. Given the density $\tau_{A_1\cdots A_L}=\tau_{A_1}\ten \cdots \ten\tau_{A_L}$, we write \[\tau_{\bf A}:=\tau_{A_1^{M_1}\cdots A_L^{M_L}}= \left(\bigotimes_{m_{\ell}\in [M_L] \ell\in [L]} \tau_{A_\ell}\right) =\bigotimes_{\ell\in[L]}\tau_{A_\ell}^{\otimes M_\ell} \pl.\]
For an arbitrary full rank density $\sigma_E\in S(\mathsf{E})$, the multipartite convex splitting error can be expressed by 
\begin{align*}
&\Delta_{M_1, \cdots, M_L}\left( \rho_{A_1 \ldots A_L E} \,\Vert\, \tau_{A_1 \ldots A_L}  \right)\\
&=\frac{1}{2}\norm{\frac{1}{M_{[L]}}\sum_{\bm }
	\rho_{A_{\bm} E}\ten \tau_{A_{\bm}^c}-\bigotimes_{\ell\in[L]}\tau_{A_\ell}^{\otimes M_\ell} \otimes \rho_{E} }{1}\\
  &=\frac{1}{2}\norm{\frac{1}{M_{[L]}}\sum_{\bm}
	\frac{\rho_{A_{\bm} E}\ten \tau_{A_{\bm}^c}}{\tau_{A_1^{M_1}\cdots A_L^{M_L}}\ten \sigma_E}-\bigotimes_{\ell\in[L]}{\mathds{1}}_{A_\ell}^{\otimes M_\ell} \otimes \frac{\rho_{E}}{\sigma_E} }{L_1(\tau_{\bf A}\ten \sigma_E)}
\\&=\frac{1}{2}\norm{\frac{1}{M_{[L]}}\sum_{\bm }
	\frac{\rho_{A_{\bm} E}}{\tau_{A_{\bm}}\ten \sigma_E} \ten {\bf 1}_{A_{\bm}^c}-\bigotimes_{\ell\in[L]}{\mathds{1}}_{A_\ell}^{\otimes M_\ell} \otimes \frac{\rho_{E}}{\sigma_E} }{L_1(\tau_{\bf A}\ten \sigma_E)}
\\&=\frac{1}{2}\norm{\frac{1}{M_{[L]}}\sum_{\bm}\pi_{\bm} 
	\left(\frac{\rho_{A_{ 1} \cdots A_{ L} E}}{\tau_{A_1 \ldots A_L}\ten \sigma_E} -{\mathds{1}}_{A_1\cdots A_L}\ten\frac{\rho_{E}}{\sigma_E}  \right)}{L_1(\tau_{\bf A}\ten \sigma_E)}
 \\&=\frac{1}{2}\norm{\Theta_{[L]}
	\left(\id_{[L]}-\mathrm{E}_{[L]}\right)(X)}{L_1(\tau_{\bf A}\ten \sigma_E)}.
\end{align*}
where $X=\frac{\rho_{A_{ 1} \cdots A_{ L} E}}{\tau_{A_1 \ldots A_L}\ten \sigma_E}$. Using the mean-zero decomposition lemma given in Lemma~\ref{lemma:tele2},
\begin{align*}
\Delta_{M_1, \cdots, M_L}\left( \rho_{A_1 \ldots A_L E} \,\Vert\, \tau_{A_1 \ldots A_L}  \right)
&=\frac{1}{2} \norm{\Theta_{[L]}\left(\sum_{\varnothing\neq S\subset [L]} T_S(X)\right)}{L_1(\tau_{\bf A}\ten \sigma_E)}
\\ &\le  \frac{1}{2} \sum_{\varnothing\neq S\subset [L]}\norm{\Theta_{[L]}\circ T_S(X)}{L_1(\tau_{\bf A}\ten \sigma_E)} \\
&=:\sum_{\varnothing\neq S\subseteq [L]}\Delta_S\pl,
\end{align*}
where for each subset $S\subseteq [L]$, we define the error term
\[\Delta_S(\rho_{A_1\cdots A_LE} \,\Vert\, \tau_{A_1\cdots A_{L}})=\frac{1}{2}\norm{\Theta_{[L]}\circ T_S(X)}{L_1(\tau_{\bf A}\ten \sigma_E)} \pl.\]
Note that $T_S(X)={\bf 1}_{A_{S^c} }\ten \mathring{X}_S$ is supported on $A_{S}$ where $\mathring{X}_S=T_S(X_S)$, so the convex splitting for $A_{S^c}$ is trivial. More precisely, one have  
\begin{align*}
&\Delta_S(\rho_{A_1\cdots A_LE}\,\Vert\,\tau_{A_1\cdots A_{L}})\\ 
&=\frac{1}{2}\norm{\Theta_{[L]}( \mathds{1}_{A_{S^c} }\ten\mathring{X}_S)}{L_1(\tau_{\bf A}\ten \sigma_E)}\\
&=\frac{1}{2}\left\|\frac{1}{M_{[L]}}\sum_{\bm } \left(\mathds{1}_{{S^c}}\ten \mathring{X}_{S}\right)_{A_{\bm}}\otimes \mathds{1}_{A_{\bm^c}}- \mathds{1}_{{\bf A}} \ten \frac{\rho_E}{\sigma_E}\right\|_{L_1(\tau_{\bf A}\ten \sigma_E)}
\\
&=\frac{1}{2}\left\|\frac{1}{M_{S}M_{S^c}}\mathds{1}_{{\bf A}_{S^c}} \ten \sum_{\bm }   \mathring{X}_{A_{\bm_S}}\otimes \mathds{1}_{A_{\bm_S^c}}  - \mathds{1}_{{\bf A}_S} \ten \frac{\rho_E}{\sigma_E}\right\|_{L_1(\tau_{{\bf A}_S}\ten \tau_{{\bf A}_{S^c}}\ten \sigma_E)}
\\
&=\frac{1}{2}\left\|\frac{1}{M_{S}} \sum_{\bm }   \mathring{X}_{A_{\bm_S}}\otimes \mathds{1}_{A_{\bm_S^c}}  - \mathds{1}_{{\bf A}_S} \ten \frac{\rho_E}{\sigma_E}\right\|_{L_1(\tau_{{\bf A}_S}\ten \tau_{{\bf A}_{S^c}}\ten \sigma_E)}
\\ &=\frac{1}{2}\norm{\Theta_{S}( \mathring{X}_S)}{L_1(\tau_{\bf A_{S}}\ten \sigma_E)}
\\ &=\frac{1}{2}\norm{\Theta_{S}\circ T_S( \frac{\rho_{A_{ S}  E}}{\tau_{A_S}\ten \sigma_E})}{L_1(\tau_{{\bf A}_{S}}\ten \sigma_E)}
\\ &=\Delta_{S}(\rho_{A_SE}\,\Vert\,\tau_{A_S}),
\end{align*}
which is an $|S|$-partite convex splitting term $\Theta_{[S]}\circ T_{[S]}$. Here, $\bf A_{S}=\bigotimes_{m_\ell \in [M_\ell], \ell \in S} A_\ell$, and we use the notation ${\bm_S}= (m_\ell)_{\ell \in S}$ as the restriction of multi-index $\bm \in \times_{\ell\in [L]} [M_\ell]$ to $\bm \in \times_{\ell\in S} [M_\ell]$.

\begin{lemma}
For any subset $S\subseteq [L]$ and any density operator $\sigma_E$,
\[\Delta_S(\rho_{A_1\cdots A_LE}\,\Vert\,\tau_{A_1\cdots A_{L}})\le 2^{|S|-1}M_{[L]}^{\frac{1-p}{p}} \norm{\frac{\rho_{A_SE}}{\tau_{A_S}\ten \sigma_E } }{ L_p (\tau_{A_S}\ten \sigma_E)}, p\in [1,2].\]
\end{lemma}
\begin{proof}
By the above observation, it suffices to argue for the case $S=[L]$. Recall that by H\"older inequality,  for a density operator $\tau$, the identity map  $\id: L_p(\tau)\to L_1(\tau)$ is a contraction. 
Then it suffices to show that for $p\in [1,2]$
\[\norm{\Theta_{[L]}\circ T_{[L]}:L_p(\tau_{A_1\cdots {A_L}}\ten \sigma_E)\to  L_p\left(\tau_{A_1^{M_1}\cdots A_L^{M_L}}\ten \sigma_E\right)}{}\le 2^{|S|}M_{[L]}^{\frac{1-p}{p}}\pl.\]
For $p=1$, we use Lemma \ref{lemma:multitele} and \ref{lemma:contraction}
\begin{align*}
\norm{\Theta_{[L]}\circ T_{[L]}:L_1\to L_1 }{}\le \norm{\Theta_{[L]}:L_1\to L_1 }{}\norm{ T_{[L]}:L_1\to L_1 }{}=2^{L}\pl.
\end{align*}
For $p=2$, given an operator $X\in \mathcal{B}(\mathsf{A}_1\ten\ldots\ten\mathsf{A}_L\ten E)$, we adopt the short notation $\o{X}=T_{[L]}(X)$. Note that by Lemma \ref{lemma:tele2}, 
$\mathrm{E}_{A_\ell} (\o{X})=0$
for all $\ell\in [L]$. This implies that the set $\{\pi_{\bm}(\o{X} )\}_{\bm \in [M_1]\times\cdots \times [M_L]}$ is orthogonal in $L_2(\tau_{\bf A}\ten \sigma_E)$. Indeed, for $\bm\neq \bm'$, without loss of generality, assume $m_{1}\neq m'_1$. Write $\tau_{{\bf A }E}=\tau_{\bf A} \ten \sigma_E$. We have 
\begin{align*}
		&\langle \pi_{\bm}(\mathring{X}), \pi_{\bm'}(\mathring{Y}) \rangle_{\tau_{{\bf A }E}} \\
&=  \lan \mathds{1}_{A_{\bm^c}E} \otimes \mathring{X}_{A_\bm}, \mathds{1}_{A_{{\bm'}^c}} \otimes \mathring{Y}_{A_{\bm'}E}\ran_{\tau_{{\bf A }E}}\\
&=  \Tr\left[\tau_{{\bf A }E}^{\frac{1}{2}}\mathring{X}_{A_\bm} \tau_{{\bf A }E}^{\frac{1}{2}}\mathring{Y}_{A_{\bm'}E}\right]\\
&=  \Tr\left[\tau_{A_{\bm \cup \bm'}E}^{\frac{1}{2}}\mathring{X}_{A_\bm} \tau_{A_{\bm \cup \bm'}E}^{\frac{1}{2}}\mathring{Y}_{A_{\bm'}E}\right]\\
&=  \Tr\left[\tau_{A_{\bm \cup \bm'\setminus \{m_1,m_1'\} } E}^{\frac{1}{2}}E_{A_{1,m_1}}(\mathring{X}_{A_\bm}) \tau_{A_{\bm \cup \bm'\setminus \{m_1,m_1'\}}E}^{\frac{1}{2}}E_{A_{1,m_1'}}(\mathring{Y}_{A_{\bm'}E})\right]\\
&=  \lan E_{A_{1,m_1}}(\pi_{\bm}(\o{X})), E_{A_{1,m_1'}}(\pi_{\bm'}(\o{Y})) \rangle_{ \tau_{{\bf A }E} }\\
&=  \lan \pi_{\bm}\circ E_{A_1}( \o{X}), \pi_{\bm'}\circ E_{A_1}(\o{Y}) \rangle_{ \tau_{{\bf A }E} }\\
&= 0.
\end{align*}
By the orthogonality, for any $X_{A_1\cdots A_L E}$
	\begin{align}
		\left\| \Theta_{[L]}(\mathring{X}) \right\|_{L_2(\tau_{{\bf A }E})}^2
&\overset{\textnormal{(a)}}{=} \frac{1}{M_{[L]}^2} \sum_{\bm}  \left\| \pi_{\bm} \left( \mathring{X} \right) \right \|_{L_2(\tau_{{\bf A }E})}^2
\nonumber\\
		&\overset{\textnormal{(b)}}{=} \frac{1}{M_{[L]}^2} \sum_{\bm}  \left\| \mathring{X} \right\|_{L_2(\tau_{A_1\cdots A_L E})}^2 \nonumber\\
		&= \frac{1}{M_{[L]}} \left\| \mathring{X} \right\|_{L_2(\tau_{A_1\cdots A_L E})}^2 
 \nonumber \\
  &\overset{\textnormal{(c)}}{\le } \frac{2^{2L}}{M_{[L]}}  \left\| X \right\|_{L_2(\tau_{A_1\cdots A_L E})}^2 \label{eq:L2}
	\end{align}
	where (a) follows from orthogonality;
	(b) is because $\pi_{\bm}$ are isometry;
	and the last inequality (c) follows by the Lemma \ref{lemma:multitele}.
For general $p\in[1,2]$, we apply complex interpolation for $\frac{1}{p}=\frac{1-\theta}{1}+\frac{\theta}{2}$, 
 $\theta = \frac{2(p-1)}{p} \in [0,1]$,
	\begin{align*}
		\left\| \Theta_{[L]}\circ T_{[L]}: L_{p} \to L_{p}  \right\|
		&\leq \left\| \Theta_{[L]}\circ T_{[L]}: L_{1} \to L_{1}\right\|^{1-\theta} \left\| \Theta_{[L]}\circ T_{[L]}: L_{2} \to L_{2}\right\|^\theta = 2^L(M_{[L]})^{ \frac{1-p}{p} },
	\end{align*}
	which finishes the proof for $S=[L]$. 
\end{proof}

\begin{remark}{\rm As we mentioned, the key property of the element $\mathring{X}:=T_{[L]}(X)$ is that $\mathring{X}$ has all ``partial mean'' zero $E_\ell T_S(X)=0, \forall \ell \in [L]$. Indeed,  for each $\ell\in [L]$, the partial mean-zero $E_\ell(\mathring{X})=0$ implies that for any fixed $(m_1,\cdots, m_{\ell-1}, m_{\ell+1},m_{L} )$ the set $\{\pi_\bm( \mathring{X})\}_{m_{\ell}=1}^{M_\ell}$ is orthogonal. Then to make the whole set $\{\pi_\bm( \mathring{X})\}_{\bm\in \prod[M_\ell]}$ orthogonal is exactly our motivation for the mean-zero decomposition Lemma~\ref{lemma:tele2}.

}    
\end{remark}
Using triangle inequality, we have the following one-shot multipartite convex splitting lemma. Note that in applying the above estimate the $\sigma_E$ can be optimized differently for each term $\Delta_S$.
\begin{shaded_theorem}[$L$-party convex splitting] \label{theorem:convex_splitting_multipartite}
    Let $\rho_{A_1A_2\ldots A_L E}$ and $\tau_{A_1 \ldots A_L}=\tau_{A_1}\ten \tau_{A_2} \cdots \ten \tau_{A_L}$ be multipartite states.
    For integers $M_1$, $M_2$, $\ldots$, $M_L \in \mathds{N}$,
	\begin{align} \label{eq:error_convex_splittgin_multipartite}
    \Delta_{M_1, \cdots, M_L}\left( \rho_{A_1 \ldots A_L E} \,\Vert\, \tau_{A_1 \ldots A_L}  \right) \le \frac{1}{2}\sum_{\varnothing\neq S\subseteq [L] } 2^{|S|} \cdot 2^{ -E_{\log M_S}\left( \rho_{A_S E} \,\Vert\, \tau_{A_S}\right)}, 
\end{align}
where the error-exponent function is defined in Eq.~\eqref{eq:error_exponent_tau},
$M_S:= \Pi_{ \ell \in S } M_\ell$, and $A_S$ denotes  systems $A_\ell$ for all $\ell \in S$. Moreover, the error exponents are all positive if and only if for all subsets $S\subseteq[L]$,
	\begin{align}
    \sum\nolimits_{\ell\in S} \log M_\ell > I_1(\rho_{A_SE} \,\Vert\, \tau_{A_S} ).
  \end{align}
\end{shaded_theorem}

\begin{remark}
Theorem~\ref{theorem:convex_splitting_multipartite} holds even if the underlying Hilbert spaces are all infinite-dimensional.
\end{remark}

%%%%%%%%%%%%%%%%%%%%%%%%%%%%%%%%

\section{Multipartite Quantum State Splitting} \label{sec:QSS}

In this section, we derive the one-shot multipartite Quantum State Splitting. 
We first provide a formal definition (see Figure~\ref{fig:QSS}).

\begin{figure}[h!]
	\centering
	\resizebox{1\columnwidth}{!}{ 
		\includegraphics{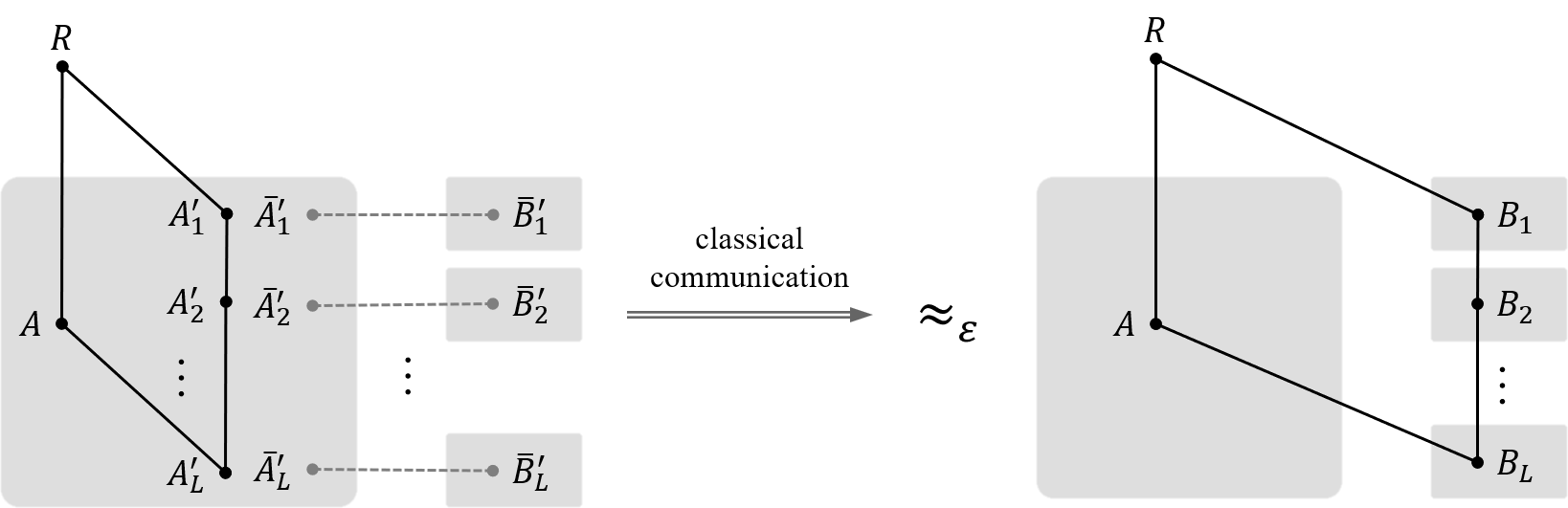}     		   		
	}
    \caption{
    \small
    Depiction of $L$-party Quantum State Splitting protocol.
    Sender holds registers $A$, $A_1'$, $\bar{A}_1'$, $A_2'$, $\bar{A}_2'$, $\cdots$, $A_L'$, $\bar{A}_L'$, and pre-shares entanglement with each of the $L$ Receivers holding register $B_\ell'$.
    $\ell \in [L]$; register $R$ represents an inaccessible Reference.
    Each grey-shaded region is allowed to perform local quantum operation at its holding registers, and (limited) noiseless classical communication from the sender to the $L$ receivers is permitted.
    After executing the protocol, we wish the systems $A_1'$, $\cdots$, $A_L'$ end up at each of the $\ell$-th receiver's side (at the right part of the figure).
    Note that each gray-shaded region is only allowed to perform local quantum operation.
	}
	\label{fig:QSS}
\end{figure}

\begin{definition}[$L$-party Quantum State Splitting] \label{definition:QSS}

	Let $\rho_{RAA_1'A_2'\ldots A_L'}$ be a pure state as input of the protocol.
	\begin{enumerate}[1.]
		\item Quantum registers $A_1'$, $A_2'$, $\ldots$, $A_L'$ and $A$ at Alice,
		$B_1 \cong A_1'$, $B_2 \cong A_2'$, $\ldots$, $B_L \cong A_L'$ at $L$ Receivers, 
		and $R$ at an inaccessible Reference.
		
		\item A resource of   entanglement, say $|\tau\rangle_{\bar{A}_1' \bar{A}_2' \ldots \bar{A}_L \bar{B}_1' \bar{B}_2' \ldots \bar{B}_L'}$, shared between Sender (holding registers $\bar{A}_1' \bar{A}_2' \ldots \bar{A}_L'$) and $L$ Receivers (each holding register $\bar{B}_\ell'$, $\ell \in [L]$), and noiseless one-way classical communication from the sender to $L$ receivers are available.
		
		\item The sender applies a local operation on her system and the shared entanglement to obtain $L$-tuple
		  $\vec{r} := (r_1, r_2, \ldots r_L)$ bits of classical messages.
		%	quantum system $M_\text{q}$.
		
		\item The sender sends the above message to each receiver, respectively,
		via noiseless one-way classical communication. % {\color{red}\cite{BW92}}.
		%	quantum communication.
		
		\item 
		Upon receiving the messages, each receiver applies a local operation on his shared entanglement to obtain an overall state $\widehat{\rho}_{RAB_1 B_2 \ldots B_L}$.
				\end{enumerate}
    
A $( \vec{r}, \eps)$ $L$-party Quantum State Splitting protocol for $\rho_{RAA_1'\ldots A_L'}$ with entanglement \\$|\tau\rangle_{\bar{A}_1' \bar{A}_2' \ldots \bar{A}_L \bar{B}_1' \bar{B}_2' \ldots \bar{B}_L'}$ satisfies
		\begin{align}
			\frac12\left\| \widehat{\rho}_{RAB_1 \ldots B_L} - \rho_{RAB_1\ldots B_L} \right\|_1 \leq \eps,
		\end{align}
		where $\rho_{RAB_1 \ldots B_L} := (\id_{A_1'\ldots A_L' \to B_1 \ldots B_L} \otimes \id_{AR}) \rho_{RAA_1'A_2'\ldots A_L'}$.
	
\end{definition}

For readability, we first show a special case of the bipartite Quantum State Splitting.

\begin{shaded_theorem}[$2$-party Quantum State Splitting] \label{theorem:bipartite_QSS}
	For any pure state $\rho_{AA_1'A_2'R} = |\rho\rangle\langle \rho|_{AA_1'A_2'R} $ , there exists a $(r_1, r_2, \eps)$ $2$-party Quantum State Splitting protocol for $\rho_{RAA_1'A_2'}$ with entanglement $|\tau\rangle_{{A}_1' {B}_1}$ and $|\tau\rangle_{{A}_2' {B}_2}$ such that
	\begin{align} \label{eq:error_bipartite_QSS}
	\eps &\leq
	\sqrt{
	2 \cdot 2^{ - E_{r_1+r_2} \left( \rho_{A_1'A_2'R} \, \Vert \, \tau_{A_1'} \otimes \tau_{A_2'} \right)  } 
	+  2^{ - E_{r_1} \left( \rho_{A_1'R} \, \Vert \, \tau_{A_1'} \right)  } 
	+  2^{ - E_{r_2} \left( \rho_{A_2'R} \, \Vert \, \tau_{A_2'} \right)  } 
	},
	% \\
 % &\leq
	% % &{\color{red}{ \overset{?}{=} }}
	% 	\sqrt{
	% 		4 \cdot 2^{ - E_{r_1+r_2} (A_1': A_2' : R)_\rho }
	% 		+ 2 \cdot 2^{ - E_{r_1} (A_1': R)_\rho }
	% 		+ 2 \cdot 2^{ - E_{r_2} (A_2': R)_\rho }
	% 	}.
	\end{align}
  where the error-exponent functions are defined in Eq.~\eqref{eq:error_exponent_tau}.
%	\begin{align} 
%	\eps &\leq \sqrt{ 4 \cdot \delta_{1,2} 
%		+ 2 \cdot \delta_1
%		+ 2 \cdot  \delta_2 }, \label{eq:error_bipartite_QSS} \\
%	\delta_{1,2} &:= 2^{- E_{r_1+r_2}(A_1': A_2' : R)_\rho }, \\
%	\delta_1 &=	2^{- E_{r_1}(A_1':  R)_\rho }, \\
%	\delta_2 &= 2^{- E_{r_2}(A_2': R)_\rho }.
%	\end{align}
Moreover, the error exponents are all positive if and only if,
	\begin{align} 
		\begin{cases}
		&r_1+r_2 > I(\rho_{A_1'A_2'R} \,\Vert\, \tau_{A_1'} \otimes \tau_{A_2'} ), \\
	&r_1 > I(\rho_{A_1'R} \,\Vert\, \tau_{A_1'} ),\\
	&r_2 > I(\rho_{A_2'R} \,\Vert\, \tau_{A_2'} ).\\
	\end{cases}
    \end{align}
%	Moreover, the overall error exponent 
%	\begin{align}
%	\frac12 \min \left\{  E_{r_1+r_2}(A_1': A_2' : R)_\rho , 
%	E_{r_1}(A_1':  R)_\rho,
%	E_{r_2}(A_2':  R)_\rho
%	\right\}
%	\end{align}
%%	\begin{align}
%%	\sup_{\alpha \in [0,1]} \min \left\{
%%	\left( \log M_1 M_2 - I_\alpha \left( \rho_{RA_1'A_2'} \, \Vert \, \tau_{A_1'} \otimes \tau_{A_2'} \right) \right),
%%	\frac{\alpha-1}{2\alpha} \left( \log M_1 -  I_\alpha \left( \rho_{RA_1'} \, \Vert \, \tau_{A_1'} \right) \right),
%%	\frac{\alpha-1}{2\alpha} \left( \log M_1 -  I_\alpha \left( \rho_{RA_1'} \, \Vert \, \tau_{A_1'} \right) \right)
%%	\right\}
%%	\end{align}
%	is positive if and only if 
%	\begin{align} \label{eq:region_original_biQSS}
%	\begin{cases}
%	& r_1 + r_2 > I( A_1' : A_2' : R )_\rho,  \\
%	& r_1 > I(A_1' : R )_\rho, \\
%	& r_2 > I(A_2' : R )_\rho. \\
%	\end{cases}
%	\end{align}
	
\end{shaded_theorem}

%\begin{remark} 
%    We pose a question that if we allowed to use any entanglement, can we achieve the error in the following form:
%	\begin{align}
%	\eps &\leq \inf_{(\tau_{A_1'}, \tau_{A_2'})\in \mathcal{S}(\mathsf{A}_1')\times \mathcal{S}(\mathsf{A}_2') }
%	\sqrt{
%	4 \cdot 2^{ - E_{r_1+r_2} \left( \rho_{A_1'A_2'R} \, \Vert \, \tau_{A_1'} \otimes \tau_{A_2'} \right)  } 
%	+ 2 \cdot 2^{ - E_{r_1} \left( \rho_{A_1'R} \, \Vert \, \tau_{A_1'} \right)  } 
%	+ 2 \cdot 2^{ - E_{r_2} \left( \rho_{A_2'R} \, \Vert \, \tau_{A_2'} \right)  } 
%	}
%	\\
%	&{\color{red}{ \overset{?}{=} }}
%		\sqrt{
%	\inf_{\tau_{A_1'}, \tau_{A_2'}} 4 \cdot 2^{ - E_{r_1+r_2} \left( \rho_{A_1'A_2'R} \, \Vert \, \tau_{A_1'} \otimes \tau_{A_2'} \right)  } 
%	+ \inf_{\tau_{A_1'}} 2 \cdot 2^{ - E_{r_1} \left( \rho_{A_1'R} \, \Vert \, \tau_{A_1'} \right)  } 
%	+ \inf_{ \tau_{A_2'}} 2 \cdot 2^{ - E_{r_2} \left( \rho_{A_2'R} \, \Vert \, \tau_{A_2'} \right)  } 
%	}?
%	\end{align}
%\end{remark}

\begin{proof}
Unipartite Quantum State Splitting has been shown via unipartite convex splitting in Ref.~\cite{ADK+17} (see also \cite{AJW18}),
	Below we will demonstrate applying their approach with the newly established bipartite convex splitting (Theorem~\ref{theorem:bipartite_convex_splitting}) to achieving a bipartite Quantum State Splitting protocol using $r_1$ and $r_2$ bits of classical communication with the desired error $\eps$.
 
Let $M = 2^{r_1}$ and $K = 2^{r_2}$.
We fix two states $\tau_{A_1'}$  and $\tau_{A_2'}$ that will be specify later.
To begin the protocol, we let the sender (Alice) and the first receiver (Bob $1$) share $M$-copies of  entanglement $\otimes_{m\in[M]} |\tau\rangle_{A_{1,m}' B_{1,m}}$, and Alice and the second receiver (Bob $2$) share $K$-copies of entanglement $\otimes_{k\in[K]} |\tau\rangle_{A_{2,k}' B_{2,k}}$, 
where for each $(m,k) \in [M]\times [K]$, Bob $1$ holds register $B_{1,m} \cong B_1$  that purifies Alice's register $A_{1,m}' \cong A_1'$, and Bob $2$ holds register $B_{2,k} \cong B_2$  that purifies register Alice's $A_{2,k}' \cong A_2'$.
Hence, we start with the following pure state:
\begin{align} \label{eq:initial_bipartite}
|\overline{\omega}\rangle := |\rho\rangle_{A A_{1}' A_{2}' R } 
\otimes_{m\in[M]} |\tau\rangle_{A_{1,m}' B_{1,m}} 
\otimes_{k\in[K]} |\tau\rangle_{A_{2,k}' B_{2,k}}. 
\end{align}

Suppose, hopefully, by the protocol, we end up with the following pure state:
\begin{align} \label{eq:target_bipartite}
\begin{split} 
|{\omega}\rangle
&:= \frac{1}{\sqrt{MK}} \sum_{(m,k)\in [M] \times [K] } |m\rangle_M |k\rangle_K
|\rho\rangle_{A B_{1,m} B_{2,k} R } |0\rangle_{A_{1,m}' A_{2,k}'}
\otimes_{\bar{m} \in [M]\setminus \{m\}} |\tau\rangle_{A_{1,\bar{m}}' B_{1,\bar{m}}}
\otimes_{\bar{k} \in [K]\setminus \{k\}} |\tau\rangle_{A_{2,\bar{k}}' B_{2,\bar{k}}},
\end{split}
\end{align}
where Alice holds registers $M$, $K$, $A$, $A_{1,[M]}' := A_{1,1}'\ldots A_{1,M}'$ and $A_{2,[K]}' := A_{2,1}' \ldots A_{2,K}'$,
and Bob 1 and Bob $2$ hold registers $B_{1,[M]} := B_{1,1} \ldots B_{1,M}$ and $B_{2,[K]} := B_{2,1} \cdots B_{2,K}$, respectively.
Here, we shorthand for $A_{1,[M]}' := A_{1,1}'\ldots A_{1,M}'$ and $A_{2,[K]}' := A_{2,1}' \ldots A_{2,K}'$ (and similarly for $B_{1,[M]}$ and $B_{2,[K]}$) to ease the burden of notation.
(We slightly abuse notation to use $M$ and $K$ denoting registers representing classical systems $[M]$ and $[K]$.)
Alice measures on registers $M$ and $K$, and send on the outcome $m$ to Bob $1$ via $r_1$ bits of classical communication, and the outcome $k$ to Bob $2$ via $r_2$ bits of classical communication.
Then, the two receivers can pick up registers $B_{1,m}$ and $B_{2,k}$, respectively, to end up with $|\rho\rangle_{A B_{1,m} B_{2,k} R } \cong |\rho\rangle_{A B_1 B_2 R }$, which is exactly the target state we aimed for the bipartite Quantum State Splitting protocol.

In what follows, we will show that there exists a local operation protocol at Alice such that we can approximate $|{\omega}\rangle$ in Eq.~\eqref{eq:target_bipartite} within an error $\eps$ no larger than the right-hand side of Eq.~\eqref{eq:error_bipartite_QSS}.
Note that a reduced state of $|\overline{\omega}\rangle$ is
\begin{align}
\overline{\omega}_{A_{1,[M]}'A_{2,[K]}'R} &= \otimes_{m\in[M]} \tau_{A_{1,m}'} \otimes_{k\in[K]} \tau_{A_{2,k}'} \otimes \rho_R.
\end{align}
Then, the bipartite convex split lemma (Theorem~\ref{theorem:bipartite_convex_splitting}) guarantees that $\overline{\omega}_{A_{1,[M]}'A_{2,[K]}'R}$ can be approximated  
by a state
\begin{align} %\label{eq:omega_bipartite_QSS}
\begin{split}
\omega_{A_{1,[M]}'A_{2,[K]}' R}
&:= \frac{1}{MK} \sum_{(m,k)\in [M] \times [K] }
\rho_{A_{1,m}' A_{2,k}' R } 
\otimes_{\bar{m} \in [M]\setminus \{m\}} \tau_{A_{1,\bar{m}}'}
\otimes_{\bar{k} \in [K]\setminus \{k\}} \tau_{A_{2,\bar{k}}'},
\end{split}
\end{align}
within an error $\eps'$ (in terms of trace distance) no larger than the right-hand side of Eq.~\eqref{eq:error_bipartite_convex_splitting} (by substituting register $A$ by $A_1'$, $B$ by $A_2'$, and $E$ by $R$).
Observe that $\omega_{A_{1,[M]}'A_{2,[K]}' R}$ is a reduced state of the desired pure state $|{\omega}\rangle$ given in Eq.~\eqref{eq:target_bipartite}. 
Hence, by Uhlmann's theorem (Fact~\ref{fact:Uhlmann}), there exists an isometry $V$ acting on register $AA_1'A_2'A_{1,[M]}'A_{2,[K]}'$ to register $MKAA_{1,[M]}'A_{2,[K]}'$ such that $V |\overline{\omega}\rangle $ is $\sqrt{2\eps'}$-close (in trace distance) to $|\omega\rangle$.
Moreover, since the isometry $V$ only acts on Alice's registers, this constitutes the bipartite Quantum State Splitting protocol with error $\eps$ no larger than $\sqrt{2\eps'}$.
% Lastly, we optimize the states $\tau_{A_1'}$  and $\tau_{A_2'}$ to conclude the proof.
\end{proof}

The above bipartite Quantum State Splitting is straightforwardly genearlized to any $L$-party case as follows.
\begin{shaded_theorem}[$L$-party Quantum State Splitting] \label{theorem:multipartite_QSS}
For any pure state $\rho_{AA_1'A_2'\ldots A_L' R} = |\rho\rangle\langle \rho|_{AA_1'A_2'\ldots A_L' R} $ , there exists a $(\vec{r}, \eps)$ $L$-party Quantum State Splitting protocol for $\rho_{RAA_1'A_2'\ldots A_L' R}$ with entanglement $|\tau\rangle_{{A}_\ell' {B}_\ell}$, $\ell \in [L]$ such that
	\begin{align} \label{eq:error_multipartite_QSS}
	\eps &\leq
	\sqrt{ \sum\nolimits_{\varnothing\neq S\subseteq[L] } 2^{|S|} \cdot 2^{ - 
    E_{r_S}\left( \rho_{A_S' R} \,\Vert\, \tau_{A_S'}\right)}
    },
    \end{align}
    where the error-exponent function is defined in Eq.~\eqref{eq:error_exponent_tau},
    $r_S:= \sum_{ \ell \in S } r_\ell$, $A_S'$ denotes  systems $A_\ell'$ for all $\ell \in S$, and $\tau_{A_S'} := \otimes_{\ell\in S} \tau_{A_\ell'}$. Moreover, the error exponents are all positive if and only if for all $\varnothing \neq S \subseteq [L]$ ,
\begin{align}
    % \sum\nolimits_{\ell\in[S]} r_\ell \
    r_S> I\left( \rho_{A_S' R} \,\Vert\, \tau_{A_S'}\right).
\end{align}
\end{shaded_theorem}

%%%%%%%%%%%%%%%%%%%%%%%%%%%%%%%%

\section{Entanglement-Assisted Quantum Broadcast Channel Simulation} \label{sec:simulation}

Let $\mathscr{N}_{A\to B_1 B_2\cdots B_L}$ be an $L$-receiver quantum broadcast channel, which is a completely positive and trace-preserving map from system $A$ to $L$ systems $B_1 B_2 \ldots B_L$.
% Moreover, we assume that the underlying Hilbert spaces $\mathsf{A}$, $\mathsf{B}_1$, $\mathsf{B}_2, $$\ldots$, $\mathsf{B}_L$ are all finite-dimensional.

\begin{definition}[$L$-party Quantum Broadcast Channel Simulation]\label{def:channel-simulation}
Let $\mathscr{N}_{A\to B_1 B_2\cdots B_L}$ be a quantum broadcast channel as input of the simulation protocol.
\begin{enumerate}[1.]
	\item 
	Sender holds registers $A$. Each Receiver holds register
	$B_1$, $B_2$, $\ldots$, $B_L$, respectively.
%	and $R$ is an inaccessible Reference.
	
	\item Free resource of perfect  entanglement is shared between Sender (holding registers $\bar{A}_1' \bar{A}_2' \ldots \bar{A}_L'$) and each of the $L$ Receivers (holding register $\bar{B}_\ell'$, $\ell \in [L]$).
	
	\item Sender applies a local operation on her systems and sends an $L$-tuple $\vec{r} := (r_1, r_2, \ldots r_L)$
	bits of classical information to each of the $L$ Receivers, respectively.
	%	quantum system $M_\text{q}$.

	\item 
	Upon receiving the message, each Receiver applies a local operation on his own system.
\end{enumerate}		
	A $( \vec{r}, \eps)$ $L$-party Quantum Broadcast Channel Simulation protocol for $\mathscr{N}_{A\to B_1 B_2\cdots B_L}$ satisfies
	\begin{align}
	\frac12\left\| \widehat{\mathscr{N}}_{A\to B_1 B_2\cdots B_L} -\mathscr{N}_{A\to B_1 B_2\cdots B_L} \right\|_{\diamond} \leq \eps,
	\end{align}
	where $\widehat{\mathscr{N}}_{A\to B_1 B_2\cdots B_L}$ is the effectively resulting linear transformation from Sender's register $A$ to $L$ Receivers' registers $B_1$, $\ldots$, $B_L$. 
	The $L$-tuple $\vec{r}$ denotes the classical communication costs.

\end{definition}

\begin{definition}[Capacity region for i.i.d.~broadcast channel simulation]
	The capacity region of simulating $\mathscr{N}_{A\to B_1 B_2\cdots B_L}$, denoted as $\mathcal{C}(\mathscr{N}_{A\to B_1 B_2\cdots B_L})$ is the closure of
	\begin{align}
	% \mathcal{C}( \mathscr{N} ) := \texttt{cl} 
    \bigcup_{\eps>0} \left\{ \vec{r} \in \mathds{R}^L: \exists \, (n \vec{r}, \eps)\text{-$L$-party Quantum Broadcast Channel Simulation protocol for $\mathscr{N}^{\otimes n}$}  \right\}.
	\end{align}
%	We denote by $\mathcal{C}( \mathscr{N} ) := \limsup_{\eps \to 0} \mathcal{C}_\eps( \mathscr{N} )$ as its capacity region.
\end{definition}
	
\begin{shaded_theorem}[Non-asymptotic achievability for $2$-party Quantum Broadcast Channel Simulation] \label{theorem:bipartite_simulation}
	Let $\mathscr{N}_{A\to B_1 B_2}$ be an $L$-receiver quantum broadcast channel.
	For any integer $n$, there exists an $(n r_1, n r_2, \eps)$ Quantum Broadcast Channel Simulation protocol for $\mathscr{N}_{A\to B_1 B_2}^{\otimes n}$ satisfying
	\begin{align} 
	\eps &\leq   k_n \cdot \sqrt{ 2 \cdot \delta_{1,2} 
		+ \delta_1
		+ \delta_2 }, \label{eq:error_bipartite_simulation} \\
	\delta_{1,2} &:=  2^{- n E_{r_1+r_2}(\mathscr{N}_{A\to B_1 B_2}) }, \\
	\delta_1 &:=	 2^{- n E_{r_1}(\mathscr{N}_{A\to B_1} ) }, \\
	\delta_2 &:=  2^{- n E_{r_2}(\mathscr{N}_{A\to  B_2}) },
	\end{align}
	% where the maximization is over all pure states $\theta_{AR}$ (where $R$ is a register purifying $A$),
	% $\rho_{B_1 B_2 R} := (\mathscr{N}_{A\to B_1B_2}\otimes \id_R) (\theta_{AR})$, and
	where the error-exponent functions are defined in Eq.~\eqref{eq:error_exponent_channel}, and
    the polynomial prefactor is
	$k_n := (n+1)^{\frac{5}{2} (|\mathsf{A}|^2-1) } $.
 % The error-exponent function is defined in \eqref{eq:error_exponent_mutual}.
	Moreover, the three exponents $E_{r_1+r_2}(\mathscr{N}_{A\to B_1 B_2})$ , $E_{r_1}(\mathscr{N}_{A\to B_1})$, $E_{r_2}(\mathscr{N}_{A\to B_2})$ are all positive
	if and only if 
	\begin{align} \label{eq:region_original_bi_simulation}
	\begin{cases}
	& r_1 + r_2 > I( \mathscr{N}_{A\to B_1 B_2} ),  \\
	& r_1 > I( \mathscr{N}_{A\to B_1} ), \\
	& r_2 > I( \mathscr{N}_{A\to B_2} ). \\
	\end{cases}
	\end{align}
	% for all purification $\theta_{AR}$.
\end{shaded_theorem}

\begin{shaded_theorem}[Capacity region for $2$-party Quantum Broadcast Channel Simulation] \label{theorem:QRST_bipartite}
	Let $\mathscr{N}_{A\to B_1 B_2}$ be a quantum broadcast channel.
	The capacity region of simulating $\mathscr{N}_{A\to B_1 B_2}$ is given by
	\begin{align} \label{eq:capacity_region_bipartite}
	\mathcal{C}(\mathscr{N}_{A\to B_1 B_2}) &=
	% \bigcap_{\theta_{AR}}
	\left\{ \vec{r} \in \mathds{R}^2:
	r_1 + r_2 \geq I(\mathscr{N}_{A\to B_1 B_2} ), \,
	r_1 \geq I( \mathscr{N}_{A\to B_1}), \,
	r_2 \geq I( \mathscr{N}_{A\to B_2} )
	\right\}.
	\end{align}
	% where the intersection is for all purification $\theta_{AR}$ and 	$\rho_{B_1 B_2 R} := \mathscr{N}_{A\to B_1B_2}\otimes \id_R (\theta_{AR})$.
\end{shaded_theorem}

\begin{remark} \label{remark:coherent_feedback_simulation}
    Though our target is to simulate the broadcast channel $\mathscr{N}_{A\to B_1 B_2}$, the established achievability given in Theorems~\ref{theorem:bipartite_simulation} and \ref{theorem:QRST_bipartite} are actually stronger for the \emph{coherent feedback simulation} of $\mathscr{N}_{A\to B_1 B_2}$; namely, we are able to simulate the isometry $\mathscr{V}_{A\to B_1 B_2 E}$, where the complementary part $E$ is retained at Sender \cite{Win06, BDH+14}.
\end{remark}

\begin{remark}
    The minimax identity given in Proposition~\ref{fact:properties_exponent}-\ref{item:minimax_exponnet} shows that the error-exponent function of channel in Theorem~\ref{theorem:bipartite_simulation} can be viewed as the error-exponent function for the channel output state induced by the worst-case input state.
    Namely, the error terms (i.e.~$\delta_{1,2}$, $\delta_1$, and $\delta_2$ in Theorem~\ref{theorem:bipartite_simulation}) for channel simulation are dominated by the worst-case input states, respectively.
\end{remark}

\begin{proof}[Proof of Theorem~\ref{theorem:QRST_bipartite}]
	The achievability directly follows from the exponential decreases of error given in Theorem~\ref{theorem:bipartite_simulation}, and noting that $\lim_{n\in\infty} \frac1n \log k_n = 0$.

	For the converse, we first note that the the requirements on the separate rates $r_1$ and $r_2$ follow from the converse of the respective single-sender single-receiver reverse Shannon theorem \cite{BSS+02,BDH+14,BCR11}. That is, for an asymptotically vanishing error we need
	\begin{align}
	\text{$r_1 \geq I( \mathscr{N}_{A\to B_1})$ and $r_2 \geq I( \mathscr{N}_{A\to B_2} )$.}
	\end{align}
	For the rate sum constraint $r_1+r_2$ the argument is similar as in the classical case \cite[Theorem 36]{CRB+22} and a simple version for the first order asymptotics as needed here is as follows. The idea is to analyze the correlations between the purifying reference system $R$ and the respective receivers in terms of the multipartite mutual information. Any quantum broadcast simulation protocol applies an encoder on the sender's side, uses classical communication at a rate $r_i$ from the sender to the $i$-th receiver's side, and then applies local decoders on the receiver's end. The multipartite mutual information is monotone under local operations at the receiver's end (similarly as the mutual information) and increases at most by $nr_i$ by sending bits at a rate $r_i$ to the $i$-the receiver, as shown by the following dimension bound (cf. the textbook methods\cite{wilde2011classical})
	\begin{align}
		I(C_1X_1:C_2X_2:R)_\sigma\leq I(C_1:C_2:R)_\sigma+\log|\mathsf{X}_1|+\log|\mathsf{X}_2|
	\end{align}
	for any classical-quantum state $\sigma_{X_1X_2C_1C_2}$, classical on systems $X_1X_2$. However, at the end of any $\varepsilon$-good protocol (Definition \ref{def:channel-simulation}) for a purified i.i.d.~input state $(\psi^\rho_{AR})^{\otimes n}$, the output state on the relevant systems has to be $\varepsilon$-close in variational distance to $\left[(\mathscr{N}_{A\to B_1B_2}\otimes\mathcal{I}_R)(\psi^\rho_{AR})\right]^{\otimes n}$. As the continuity of the multipartite mutual information is inherited from the continuity of the von Neumann entropy (see, e.g., \cite{Winter16} for tight estimates), we need at the end of the protocol that the multipartite mutual information on the relevant systems obeys
	\begin{align}
	n\cdot I(B_1:B_2:R)_{(\mathscr{N}\otimes\mathcal{I})(\psi^{\rho})}\leq nr_1+nr_2+ O(n\varepsilon\log|\mathsf{B}_1||\mathsf{B}_2|)\,,
	\end{align}
	where we also employed that the multipartite mutual information is additive on tensor product states. Since this has to hold for any i.i.d.~input states $\rho_A^{\otimes n}$, the claim follows from taking the limit $\varepsilon\to0$ for asymptotically perfect protocols.
%except for the entanglement-assistance $\bar{A}_i'\bar{B}_i'$ that is uncorrelated with the reference system, and as such we have that $I(\bar{B}_1':\bar{B}_2':R)=0$ for the corresponding state. 
\end{proof}

\begin{proposition}[Error exponent] \label{proposition:expoent_simulation}
	Let $\mathscr{N}_{A\to B_1 B_2}$ be a quantum broadcast channel.
	There exists a sequence of $(n \vec{r}, \eps_n)$ Quantum Broadcast Channel Simulation protocol for $\mathscr{N}_{A\to B_1 B_2}^{\otimes n}$ satisfying achievable error exponent bound:
	\begin{align}
	\liminf_{n\to\infty} - \frac1n \log \eps_n \geq
	\frac12 \min \left\{  E_{r_1+r_2}(\mathscr{N}_{A\to B_1 B_2}) , 
	E_{r_1}(\mathscr{N}_{A\to B_1}),
	E_{r_2}(\mathscr{N}_{A\to B_2})
	\right\}.
	\end{align}
	Moreover, the overall error exponent in the lower bound is positive if and only if $\vec{r}$ is in the interior of the capacity region $\mathcal{C}(\mathscr{N})$ given in Eq.~\eqref{eq:capacity_region_bipartite}.
\end{proposition}

\begin{remark}
The optimal error exponent (under \emph{channel purified distance}) was established for point-to-point quantum channel simulation \cite{LY21b}.
Even for the special case of point-to-point quantum channel simulation, it is not clear whether the error exponent given in Proposition~\ref{proposition:expoent_simulation} or the result in \cite{LY21b} will give a better error exponent under diamond norm due to the fundamental different expressions of the error-exponent functions.

Our intention in this paper, however, is not to derive the optimal error exponent, but to devise a machinery for analyzing one-shot quantum broadcast channel simulation and the corresponding capacity region.
\end{remark}

% \begin{comment}
% \medskip
The overall error exponent for Quantum Broadcast Channel Simulation is positive if and only if the classical communication cost $\vec{r}$ is in the interior of the capacity region $\mathcal{C}(\mathscr{N})$ as shown in Theorem~\ref{theorem:bipartite_simulation} and Proposition~\ref{proposition:expoent_simulation}.
One may wonder if the broadcast channel simulation is still achievable as $\vec{r}$ asymptotically converges to the boundary $\mathcal{C}(\mathscr{N})$? Note that in this scenario the overall error exponent will vanish asymptotically.
In the following Proposition~\ref{proposition:moderate_simulation}, we show that the broadcast channel simulation is still achievable once $\vec{r}$ converges to the boundary $\mathcal{C}(\mathscr{N})$ of speed $\Theta(n^{-t})$ for some $t\in(0,\sfrac12)$.
We call such $\Theta(n^{-t})$ a \emph{strictly moderate sequence} \cite{HW16,WH17,CG22,RTB23} and such a study as a \emph{moderate deviation analysis} (compared to the \emph{large deviation analysis} given in Theorem~\ref{theorem:bipartite_simulation} and Proposition~\ref{proposition:expoent_simulation}, in which $\vec{r}$ is bounded away from the boundary of $\mathcal{C}(\mathscr{N})$).

% We call $(a_n)_{n\in\mathds{N}}$ to be a \emph{strong moderate deviation sequence} if $a_n = \Theta(n^{-t})$ for some $t\in(0,\sfrac12)$.
 
\begin{proposition}[Achievability for moderate deviation analysis] \label{proposition:moderate_simulation}
	Let $\mathscr{N}_{A\to B_1 B_2}$ be a quantum broadcast channel on finite-dimensional Hilbert spaces.
 % , and let
	% $(a_n)_{n\in\mathds{N}}$, $(b_n)_{n\in\mathds{N}}$, and $(c_n)_{n\in\mathds{N}}$ be any strong moderate deviation sequences.
	For any sequence $\vec{r}_n := (r_{1,n}, r_{2,n})\in \mathcal{C}(\mathscr{N}_{A\to B_1 B_2})$ satisfying
    \begin{align}
        a_n := \min\left\{
        \frac{r_{1,n} - I(\mathscr{N}_{A\to B_1})}{\sqrt{V(\mathscr{N}_{A\to B_1})}},
        \frac{r_{2,n} - I(\mathscr{N}_{A\to B_2})}{\sqrt{V(\mathscr{N}_{A\to B_1})}}
        \frac{r_{1,n} + r_{2,n} - I(\mathscr{N}_{A\to B_1 B_2})}{\sqrt{V(\mathscr{N}_{A\to B_1 B_2})}}
        \right\}
        = \Theta(n^{-t})
    \end{align}
    for some $t\in (0,\sfrac12)$,
	% \begin{align}
	% r_{1,n} = I(B_1:R)_\rho + a_n,
	% \text{or }
	% r_{2,n} = I(B_2:R)_\rho + b_n,
	% \text{or }
	% r_{1,n} + r_{2,n} = I(B_1:B_2:R)_\rho + c_n,
	% \end{align}
	then
	there exists a sequence of $(n \vec{r}_n, \eps_n)$ Quantum Broadcast Channel Simulation protocols for $\mathscr{N}_{A\to B_1 B_2}^{\otimes n}$ such that
    \begin{align}
        \liminf_{n\to\infty} -\frac{1}{n a_n^2} \log \eps_n \geq \frac14.
    \end{align}
	% \begin{align}
	% \eps_n \precsim \max_{ \theta_{AR} } \max\left\{
	% 2^{ - \frac{ n a_n^2}{4 V(B_1:R)_\rho} },
	% 2^{ - \frac{ n b_n^2}{4 V(B_2:R)_\rho} },
	% 2^{ - \frac{ n c_n^2}{4 V(B_1:B2:R)_\rho} }.
	% \right\}
	% \end{align}
\end{proposition}

\begin{remark}
For the special case of point-to-point quantum channel simulation, Proposition~\ref{theorem:multipartite_simulation} 
translates to the following scenario: given $\eps_n = 2^{- n a_n^2}$, then there exists a sequence of $(nr_n, \eps_n)$-quantum channel simulation protocols for $\mathscr{N}_{A\to B}^{\otimes n}$ such that $r_n \leq I(\mathscr{N}) + 2\sqrt{V(\mathscr{N}) }\cdot a_n$.
Such an achievable rate for classical communication cost coincides with the result given in \cite{RTB23} (under \emph{channel purified distance}).
We expect that the best achievable moderate deviation expansion of the classical communication cost under diamond norm is $I(\mathscr{N}) + \sqrt{2V(\mathscr{N})}\cdot a_n$, which is left as a future work.

% implies that for rate of communication cost $r_n = I(\mathscr{N}) + a_n$, the error is upper bounded by $2^{- \sfrac{na_n^2}{4}}$ asymptotically (by ignoring the higher-order terms).
% We expect that the best decaying asymptotics (under diamond norm) should be $2^{- \sfrac{na_n^2}{2}}$ and left it as a future work.

% In \cite{RTB23}, an optimal characterization for point-to-point channel simulation under the \emph{channel purified distance} was proved, in which the same error decaying rate $2^{- \sfrac{na_n^2}{4}}$ was also obtained.
\end{remark}

The proof of Proposition~\ref{proposition:moderate_simulation} follows from the achievability given in Theorem~\ref{theorem:bipartite_simulation}, the properties of error-exponent function (Proposition~\ref{fact:properties_exponent}), and 
the standard moderate deviation analysis given in \cite{Hao-Chung, CH17, CG22}.

\begin{proof}

Let us first consider the unipartite case, i.e.~$\mathscr{N}_{A\to B}$ provided that $V(\mathscr{N}_{A\to B}) >0$.
Let 
\begin{align}
    r_n &:= I(\mathscr{N}) + a_n,
    % \psi_n &:= \argmax I(B:R)_{\mathscr{N}(\psi)}
\end{align}
where $(a_n)_{n\in\mathds{N}}$ is any positive sequence satisfying $a_n = \Theta(n^{-t})$ for some $t\in(0,\sfrac12)$.
We recall Proposition~\ref{fact:properties_exponent}-\ref{item:limiting_exponnet}:
\begin{align}
	\liminf_{n\to \infty}\frac{E_{r_n}(\mathscr{N})}{a_n^2} 
	&\geq \frac{1}{2V(\mathscr{N})}. \label{eq:moderate_goal_uni0}
\end{align}
Then by Theorem~\ref{theorem:bipartite_simulation} (with $L=1$), we estimate the error probability for simulating $\mathscr{N}_{A\to B}^{\otimes n}$ using rate $n r_n$ as
\begin{align}
	\liminf_{n\to \infty} -\frac{1}{n a_n^2 }\log \eps(\mathscr{N}_{A\to B}^{\otimes n}) &\geq 
	\liminf_{n\to \infty} -\frac{1}{n a_n^2 }\log \left\{ 
	k_n \sqrt{ 2^{-E_{n r_n}(\mathscr{N}^{\otimes n})  }  } \right\} \\
	&= \liminf_{n\to \infty}\left\{ \frac{E_{n r_n}(\mathscr{N}^{\otimes n})  }{2 n a_n^2}
	 - \log \frac{k_n}{n a_n^2}  \right\} \\
	&=  \liminf_{n\to \infty}\left\{ \frac{E_{r_n}(\mathscr{N})  }{2 a_n^2}
	- \log \frac{k_n}{n^{1 - 2t}} \right\} \\
	&\geq \frac{1}{4 V(\mathscr{N}_{A\to B})},
\end{align}
where the last line holds for any polynomial pre-factor $k_n$ for any $t\in(0,\sfrac12)$.
This leads to our assertion of the moderate deviation achievability for simulating point-to-point channel $\mathscr{N}_{A\to B}$ via re-scaling $r_n = I(\mathscr{N}) + a_n$ by $r_n = I(\mathscr{N}) + \sqrt{V(\mathscr{N})} \cdot a_n$.

Similar reasoning straightforwardly applies to the multipartite case, i.e.~for any non-empty subset $S\subseteq[L]$ (with fixed $L \in \mathds{N}$), we have
\begin{align}
	\liminf_{n\to \infty}\frac{E_{r_n}(\mathscr{N}_{A\to B_S})}{a_n^2} 
	&\geq \frac{1}{2V(\mathscr{N}_{A\to B_S})},
\end{align}
and thus
\begin{align}
	\liminf_{n\to \infty} -\frac{1}{n a_n^2 }\log \eps(\mathscr{N}_{A\to B_S}^{\otimes n}) 
	&\geq \frac{1}{4 V(\mathscr{N}_{A\to B_S})}.
\end{align}
Combined with Theorem~\ref{theorem:bipartite_simulation}.
The overall asymptotic error decay will be dominated by the worst channel output subset $S$. This concludes the proof.
\end{proof}

% \end{comment}

We now present the proof of bipartite quantum broadcast channel simulation.
\begin{proof}[Proof of Theorem~\ref{theorem:bipartite_simulation}]
Before diving into the proof, let us first elaborate on the proof structure. Our main technique relies on the multipartite Quantum State Splitting established in Section~\ref{sec:QSS} and the \emph{Post-Selection Technique} \cite{ChristKoenRennerPostSelect}. Note that the idea of unipartite Quantum State Splitting and the Post-Selection Technique have been used to prove point-to-point quantum channel simulation \cite{BCR11}. Unlike \cite{BCR11} working with the smooth entropy formalism, we will demonstrate how the Post-Selection Technique can work with R\'enyi information measures by employing certain properties shown in Lemmas~\ref{lemma:dimension_bound} and \ref{lemma:convexity}.

We start with the \emph{de Finetti} type input state: 
$\zeta_{AR}^n = \int \psi_{AR}^{\otimes n} \, \mathrm{d}(\psi_{AR})$,
where $\psi_{AR} \in \mathcal{S}(\mathsf{A}\otimes \mathsf{R})$ is a pure state, and the integration is with respect to the Haar measure on the unitary group acting on $\mathsf{A}\otimes \mathsf{R}$.
Moreover, we denote by $\zeta_{ARR'}^n$ a purification of $\zeta_{AR}^n$.
Let $\mathscr{U}_{A\to B_1 B_2 E}$ be a Stinespring dilation of the broadcast channel $\mathscr{N}_{A\to B_1 B_2}$.
Then, Sender first simulate a local isometry $\mathscr{N}_{A\to B_1 B_2}$ at her side to obtain the state
\begin{align}
\zeta_{B_1 B_2 E R R'}^n := \left(\mathscr{U}_{A\to B_1 B_2 E}^{\otimes n} \otimes \id_{RR'} \right) \zeta_{ARR'}^n.
\end{align}

Next, we apply the $2$-party Quantum State Splitting protocol (Theorem~\ref{theorem:bipartite_QSS}) with registers $A\leftarrow E$, $A_1'\leftarrow B_1$, $A_2'\leftarrow B_2$, and $R\leftarrow RR'$, to send the $B_1$-part to Receiver $1$ via $n r_1$ bits of classical communication and send the  $B_2$-part to Receiver $2$ via $n r_2$ bits of classical communication.
The pre-shared entanglement used in the protocol is
many copies of $|\zeta^n\rangle_{B_1 \bar{B}_1}$ and $|\zeta^n\rangle_{B_2 \bar{B}_2}$, where the first Receiver holds register $\bar{B}_1^n$ and the second Receiver holds register $\bar{B}_2^n$.
The above procedure is embodied by a simulated isometry $\widehat{\mathscr{U}}_{A\to B_1 B_2 E}^n$ acting on the state $\zeta_{ARR'}^n$.
Then, the resulting error, denoted by $\eps'(\zeta^n)$, is 
\begin{align}
\frac12 \left\| \left( \widehat{\mathscr{U}}_{A\to B_1 B_2 E}^n - \mathscr{U}_{A\to B_1 B_2 E}^{\otimes n} \right) (\zeta_{ARR'}^n)\right\|_1
&\leq \eps'(\zeta^n), \\
\eps'(\zeta^n) &:= \sqrt{ 2 \cdot \delta_{1,2}'(\zeta^n)
	+  \delta_1'(\zeta^n)
	+  \delta_2'(\zeta^n) }, \label{eq:after_bi_QSS} \\
\delta_{1,2}'(\zeta^n) &:= 2^{- E_{nr_1+nr_2}(B_1: B_2 : R R')_{(\mathscr{N}^{\otimes n} \otimes \id)(\zeta^n)} }, \\
\delta_1'(\zeta^n) &:=	2^{- E_{nr_1}(B_1:  RR')_{(\mathscr{N}^{\otimes n} \otimes \id)(\zeta^n)} }, \\
\delta_2'(\zeta^n) &:= 2^{- E_{nr_2}(B_2: RR')_{(\mathscr{N}^{\otimes n} \otimes \id)(\zeta^n)} },
\end{align}
where the error-exponent function is defined in Eq.~\eqref{eq:error_exponent_marginal}.
Moreover, by tracing out the system $E^n$ and the monotonicity of trace distance \cite{NC09}, we obtain the bound:
\begin{align}
\frac12 \left\| \left( \widehat{\mathscr{N}}_{A\to B_1 B_2}^{n} - \mathscr{N}_{A\to B_1 B_2}^{\otimes n} \right) (\zeta_{ARR'}^n)\right\|_1 \leq \eps'(\zeta^n).
\end{align}
Here, the channel $\widehat{\mathscr{N}}_{A\to B_1 B_2}^{n}$ is the effectively simulated proximity to our target $\mathscr{N}_{A\to B_1 B_2}^{\otimes n}$.

By using the Post-Selection Technique (Proposition~\ref{posti}), this guarantees that
\begin{align}
\begin{split} \label{eq:prefactor1}
\frac12 \left\| \widehat{\mathscr{N}}_{A\to B_1 B_2}^{n} - \mathscr{N}_{A\to B_1 B_2}^{\otimes n} \right\|_\diamond \leq (n+1)^{|\mathsf{A}|^2 - 1} \eps'(\zeta^n).
\end{split}
\end{align}
It remains to remove the auxiliary system $R'$ and to upper bound the error $\eps'(\zeta^n)$ with the one dominated by the worst-case state $\psi_{AR}^{\otimes n}$ in the mixture of the {de Finetti} type state $\zeta_{AR}^n$.

Note that we can assume $|\mathsf{R}'|\leq  (n+1)^{|\mathsf{A}|^2-1}$ by Fact~\ref{posti}.
Further, Lemma~\ref{lemma:dimension_bound} shows that, for each $\alpha>0$,
\begin{align}
I_\alpha(B_1 : B_2 : RR')_{(\mathscr{N}^{\otimes n} \otimes \id)(\zeta^n)} \leq I_\alpha(B_1 : B_2 : R)_{(\mathscr{N}^{\otimes n} \otimes \id)(\zeta^n)} + \frac{2\alpha}{\alpha-1} \log \left[ (n+1)^{|\mathsf{A}|^2 - 1} \right],
\end{align}
which by the definition of the error-exponent function in Eq.~\eqref{eq:error_exponent_marginal}, in turn, implies that
\begin{align} \label{eq:prefactor2}
E_{nr_1 + nr_2}(B_1: B_2 : RR')_{(\mathscr{N}^{\otimes n} \otimes \id)(\zeta^n)} \geq E_{nr_1 + nr_2}(B_1: B_2 : R)_{(\mathscr{N}^{\otimes n}  \otimes  \id)(\zeta^n)} - 2  \log \left[ (n+1)^{|\mathsf{A}|^2 - 1} \right];
\end{align}
likewise for $E_{nr_1 }(B_1 : RR')_{(\mathscr{N}^{\otimes n} \otimes  \id)(\zeta^n)} $ and $E_{nr_2}( B_2 : RR')_{(\mathscr{N}^{\otimes n} \otimes  \id)(\zeta^n)} $.

Using Carath\'eodory's type theorem (Fact~\ref{mario}), we can write
\begin{align}
\zeta_{AR}^n = \sum\nolimits_{i\in\mathcal{I}} p_i (\psi_{AR}^i)^{\otimes n}
\end{align}
for some pure state $\psi_{AR}^i$,  finite set $\mathcal{I}$ with $|\mathcal{I}| = (n+1)^{2|\mathsf{A}||\mathsf{R}|-2}$, and  probability distribution $\{p_i\}_{i\in\mathcal{I}}$.
Then, Lemma~\ref{lemma:convexity} with $L=2$ and the additivity of R\'enyi information (Proposition~\ref{fact:properties_sandwiched}-\ref{item:additivity_sandwiched}) show that
\begin{align}
I_\alpha(B_1 : B_2 : R)_{(\mathscr{N}^{\otimes n} \otimes \id)(\sum_i p_i (\psi^i)^{\otimes n})} &\leq \sum\nolimits_{i\in\mathcal{I}} p_i I_\alpha(B_1 : B_2 : R)_{(\mathscr{N}^{\otimes n} \otimes \id)((\psi^i)^{\otimes n})} + 2 \log \left[ (n+1)^{2|\mathsf{A}||\mathsf{R}| - 2} \right] \\
&= \sum\nolimits_{i\in\mathcal{I}} p_i n I_\alpha(B_1 : B_2 : R)_{(\mathscr{N} \otimes \id)(\psi^i)} + 2 \log \left[ (n+1)^{|\mathsf{A}||\mathsf{R}| - 2} \right] \\
&\leq \max_{i\in\mathcal{I}} n I_\alpha(B_1 : B_2 : R)_{(\mathscr{N}\otimes \id)(\psi^i)} + 2 \log \left[ (n+1)^{2|\mathsf{A}||\mathsf{R}| - 2} \right] \\
&\leq n I_\alpha(\mathscr{N}_{A\to B_1 B_2})+ 2 \log \left[ (n+1)^{2|\mathsf{A}||\mathsf{R}| - 2} \right].
\end{align}
This then implies that
\begin{align} \label{eq:prefactor3}
E_{nr_1 + nr_2}(B_1: B_2 : R)_{(\mathscr{N}^{\otimes n} \otimes \id)(\zeta^n)} \geq  n E_{r_1 + r_2}(\mathscr{N}_{A\to B_1 B_2}) - \log \left[ (n+1)^{2|\mathsf{A}||\mathsf{R}| - 2} \right].
\end{align}
Similar bounds hold for $E_{nr_1 }(B_1 : R)_{(\mathscr{N}^{\otimes n} \otimes \id)(\zeta^n)} $ and $E_{nr_2}( B_2 : R)_{(\mathscr{N}^{\otimes n} \otimes \id)(\zeta^n)} $ as well.

Lastly, combining the prefactors incurred in Eq.~\eqref{eq:prefactor1}, Eq.~\eqref{eq:prefactor2}, and Eq.~\eqref{eq:prefactor3}, we conclude the proof from Eq.~\eqref{eq:after_bi_QSS} with the overall polynomial prefactor:
\begin{align}
 (n+1)^{\frac{|\mathsf{A}|^2-1}{2}}  (n+1)^{|\mathsf{A}|^2-1}(n+1)^{{|\mathsf{A}||\mathsf{R}|-1}} 
\leq (n+1)^{\frac{5}{2} ( |\mathsf{A}|^2-1 ) } =: k_n.
\end{align} 
\end{proof}

The scenario of simulation arbitrary $L$-party quantum broadcast channel follows the same proof as in the bipartite case (Theorem~\ref{theorem:bipartite_simulation}) and the multipartite Quantum State Splitting established in Theorem~\ref{theorem:multipartite_QSS}.
\begin{shaded_theorem}[Non-asymptotic achievability for $L$-party Quantum Broadcast Channel Simulation] \label{theorem:multipartite_simulation}
	Let $\mathscr{N}_{A\to B_1 B_2\ldots B_L}$ be an $L$-receiver quantum broadcast channel.
	For any integer $n$, there exists an $(n \vec{r}, \eps)$ Quantum Broadcast Channel Simulation protocol for $\mathscr{N}_{A\to B_1 B_2\ldots B_L}^{\otimes n}$ satisfying
	\begin{align} 
	\eps &\leq   k_n \cdot \sqrt{ \sum\nolimits_{\varnothing \neq S \subseteq [L]}
    2^{|S|} \cdot 2^{- n E_{r_S}(\mathscr{N}_{A\to  B_S}) }
    },
	\end{align}
	where the error-exponent functions are defined in Eq.~\eqref{eq:error_exponent_channel}, 
    the polynomial prefactor is
	$k_n := (n+1)^{\frac{3+L}{2} (|\mathsf{A}|^2-1) } $,
    $r_S:= \sum_{ \ell \in S } r_\ell$, and $B_S$ denotes  systems $B_\ell$ for all $\ell \in S$.
\end{shaded_theorem}

\begin{shaded_theorem}[Capacity region for $L$-party Quantum Broadcast Channels Simulation] \label{theorem:QRST_multipartite}
	Let $\mathscr{N}_{A\to B_1 B_2\ldots B_L}$ be an $L$-receiver quantum broadcast channel.
	The capacity region of simulating $\mathscr{N}_{A\to B_1 B_2\ldots B_L}$ is given by
	\begin{align} \label{eq:capacity_region_multipartite}
	\mathcal{C}(\mathscr{N}_{A\to B_1 \ldots B_L}) =
	\left\{ \vec{r} \in \mathds{R}^L:
	\sum\nolimits_{ \ell \in S } r_\ell \geq I(\mathscr{N}_{A\to B_S}), \forall S \subseteq [L]
	\right\},
	\end{align}
	where $B_S$ denotes systems $B_\ell$ for all $\ell \in S$.
\end{shaded_theorem}
Moreover, the lower bound to the error exponent and the achievability of moderate deviations as stated in Propositions~\ref{proposition:expoent_simulation} and \ref{proposition:moderate_simulation} also immediately hold for $L$-receiver broadcast channel simulation.

%%%%%%%%%%%%%%%%%%%%%%%%%%%%%%%%

\section*{Acknowledgement}

We thank Pau Colomer Saus and Andreas Winter for discussions and agreeing to coordinate the timeline of uploading our and their concurrent work \cite{SW23} to the arXiv simultaneously. HC is supported by the Young Scholar Fellowship (Einstein Program) of the Ministry of Science and Technology, Taiwan (R.O.C.) under Grants No.~MOST 109-2636-E-002-001, No.~MOST 110-2636-E-002-009, No.~MOST 111-2636-E-002-001, No.~MOST 111-2119-M-007-006, and No.~MOST 111-2119-M-001-004, by the Yushan Young Scholar Program of the Ministry of Education, Taiwan (R.O.C.) under Grants No.~NTU-109V0904, No.~NTU-110V0904, and No.~NTU-111V0904 and by the research project ``Pioneering Research in Forefront Quantum Computing, Learning and Engineering'' of National Taiwan University under Grant No. NTU-CC-111L894605.''
HC is thankful for RWTH Aachen University for accommodating the visit when doing this project.
LG is partially supported by NSF grant DMS-2154903.
MB acknowledges funding by the European Research Council (ERC Grant Agreement No.~948139). MB thanks Patrick Hayden for suggesting the task of quantum broadcast channel simulation \cite{HD07}, and Navneeth Ramakrishnan \cite{Ramakrishnan23} and Michael Walter for discussions on the topic.

%H.-C.~Cheng are supported by the Young Scholar Fellowship (Einstein Program) of the Ministry of Science and Technology in Taiwan (R.O.C.) under Grant MOST 110-2636-E-002-009, and are supported by the Yushan Young Scholar Program of the Ministry of Education in Taiwan (R.O.C.) under Grant NTU-110V0904,  Grant NTU-CC-111L894605, and Grand NTU-111L3401.

%%%%%%%%%%%%%%%%%%%%%%%%%%%%%%%%

\appendix

\section{Technical Lemmas} \label{sec:lemmas}

\begin{fact}[Uhlmann's theorem \cite{Uhl76}, \cite{FG99}, {\cite[Lemma 2.2]{DHW08}}] \label{fact:Uhlmann}
	Let $\psi_{AB} = |\psi\rangle\langle \psi|_{AB}$ and $\varphi_{AC} = |\varphi\rangle\langle \varphi|_{AC}$ be two pure quantum states.
	Then, there exists an isometry $\mathscr{V}_{B\to C}$ satisfying
    \begin{align}
        \frac12\| \psi_{A} - \varphi_{A} \|_1 \leq \varepsilon \Rightarrow
		\frac12\| \mathscr{V}_{B\to C}( \psi_{AB} ) - \varphi_{AC} \|_1 =
		\mathrm{P}(\mathscr{V}_{B\to C}( \psi_{AB} ), \varphi_{AC} ) \leq \sqrt{ 2\varepsilon - \varepsilon^2 } \leq \sqrt{ 2\varepsilon }.
    \end{align}
	% \begin{enumerate}[(i)]
	% 	\item\label{item:purified}	$\displaystyle \mathrm{P}(\psi_{A}, \varphi_{A}) =
	% 	\mathrm{P}(\mathscr{V}_{B\to C}( \psi_{AB} ), \varphi_{AC} )$.
		
	% 	\item\label{item:trace} $\displaystyle \frac12\| \psi_{A} - \varphi_{A} \|_1 \leq \varepsilon \Rightarrow
	% 	\frac12\| \mathscr{V}_{B\to C}( \psi_{AB} ) - \varphi_{AC} \|_1 =
	% 	\mathrm{P}(\mathscr{V}_{B\to C}( \psi_{AB} ), \varphi_{AC} ) \leq \sqrt{ 2\varepsilon - \varepsilon^2 } \leq \sqrt{ 2\varepsilon }$.
	% \end{enumerate}
	% Here, $\mathrm{P}(\rho, \sigma) := \sqrt{ 1 - \|\sqrt{\rho}\sqrt{\sigma}\|_1}$ is the purified distance \cite{TCR10}.
\end{fact}

\begin{lemma_dimension_bound}[Dimension bound]
	For any states $\rho_{ABC} \in \mathcal{S}(\mathsf{A} \otimes \mathsf{B} \otimes \mathsf{C})$, $\tau_{A} \in \mathcal{S}(\mathsf{A})$, and $\alpha >0$, we have
	\begin{align}
	I_\alpha\left(\rho_{ABC} \,\Vert\, \tau_A \right)
	\leq I_\alpha\left(\rho_{AB} \,\Vert\, \tau_A \right) + \frac{2\alpha}{\alpha-1} \log |\mathsf{C}|.
	\end{align}
\end{lemma_dimension_bound}

\begin{proof}
	For every $\eps>0$, let $\sigma_B \in \mathcal{S}(\mathsf{B})$ be a density operator satisfying
	\begin{align} \label{eq:dimension_bound1}
	D_\alpha(\rho_{AB} \,\Vert\, \tau_A \otimes \sigma_B) \leq I_\alpha\left(\rho_{AB} \,\Vert\, \tau_A \right)   + \eps,
	\end{align}
	and denote by $\tilde{\sigma}_{BC} := \sigma_B \otimes \frac{\mathds{1}_C}{|\mathsf{C}|}$ for short.
	Note that $\rho_{ABC} \leq |\mathsf{C}| \cdot \rho_{AB} \otimes \mathds{1}_C$, which implies
	\begin{align}
	\frac{1}{|\mathsf{C}|^2} \rho_{ABC} \leq \rho_{AB} \otimes \frac{\mathds{1}_C}{|\mathsf{C}|}.
	\end{align}
	Then,
	\begin{align}
	\frac{1}{|\mathsf{C}|^2} \left( \tau_A \otimes \tilde{\sigma}_{BC} \right)^{\frac{1-\alpha}{2\alpha}}  \rho_{ABC} \left( \tau_A \otimes \tilde{\sigma}_{BC} \right)^{\frac{1-\alpha}{2\alpha}}
	\leq
	\left( \tau_A \otimes\tilde{\sigma}_{BC} \right)^{\frac{1-\alpha}{2\alpha}}
	\rho_{AB} \otimes \frac{\mathds{1}_C}{|\mathsf{C}|}
	\left( \tau_A \otimes \tilde{\sigma}_{BC} \right)^{\frac{1-\alpha}{2\alpha}}.
	\end{align}
	This implies, for all $\alpha> 0$,
	\begin{align}
	\frac{1}{|\mathsf{C}|^{2}} \left\| \left( \tau_A \otimes \tilde{\sigma}_{BC} \right)^{\frac{1-\alpha}{2\alpha}}  \rho_{ABC} \left( \tau_A \otimes \tilde{\sigma}_{BC} \right)^{\frac{1-\alpha}{2\alpha}} \right\|_\alpha
	\leq \left\| \left( \tau_A \otimes\tilde{\sigma}_{BC} \right)^{\frac{1-\alpha}{2\alpha}}
	\rho_{AB} \otimes \frac{\mathds{1}_C}{|\mathsf{C}|}
	\left( \tau_A \otimes \tilde{\sigma}_{BC} \right)^{\frac{1-\alpha}{2\alpha}} \right\|_\alpha.
	\end{align}
	By invoking the definition of the sandwiched R\'enyi divergence $D_\alpha(\rho\,\Vert\,\sigma) := \frac{\alpha}{\alpha-1} \log \left\|\sigma^{\frac{1-\alpha}{2\alpha}} \rho \sigma^{\frac{1-\alpha}{2\alpha}} \right\|_\alpha$, we have
	\begin{align}
	D_\alpha\left(\rho_{ABC} \,\Vert\, \tau_A \otimes \tilde{\sigma}_{BC} \right)
	- \frac{2\alpha}{\alpha-1} \log |\mathsf{C}|
	&\leq D_\alpha\left(\left.\rho_{AB} \otimes \frac{\mathds{1}_C}{|\mathsf{C}|} \,\right\Vert\, \tau_A \otimes \tilde{\sigma}_{BC} \right) \\
	&= D_\alpha\left(\rho_{AB}  \,\Vert\, \tau_A \otimes {\sigma}_{B} \right) \\
	&\leq  I_\alpha\left(\rho_{AB} \,\Vert\, \tau_A \right) +\eps,
	\end{align}
	where we have used the fact that the sandwiched R\'enyi divergence remains identical by appending a state.
	
	Lastly, by definition, noting that
	\begin{align}
	D_\alpha\left(\rho_{ABC} \,\Vert\, \tau_A \otimes \tilde{\sigma}_{BC} \right) \geq I_\alpha\left(\rho_{ABC} \,\Vert\, \tau_A \right),
	\end{align}
    and letting $\eps\to 0$ conclude the proof.
\end{proof}

%\begin{lemma}[Chain rule] \label{lemma:chain_rule}
%	For any $L \in \mathds{N}$ and $L$-partite state $\rho_{A_1A_2\ldots A_L}$, we have
%	\begin{align}
%		I(A_1: A_2: \cdots : A_L)_\rho &:= D\left(\rho_{A_1A_2\ldots A_L} \Vert \otimes_{\ell \in [L]} \rho_{A_\ell} \right) \\
%		&= \inf_{\sigma_{\ell} \in \mathcal{S}(\mathsf{A}_{\ell}), \, \ell = 2,\ldots, L}
%			D\left( \rho_{A_1A_2\ldots A_L} \Vert \rho_{A_1} \ten \sigma_{A_2} \ten\cdots \ten \sigma_{A_L} \right) \\
%		&= \inf_{\sigma_{\ell} \in \mathcal{S}(\mathsf{A}_{\ell}), \, \ell\in[L]}
%		D\left( \rho_{A_1A_2\ldots A_L} \Vert \otimes_{\ell \in [L]} \sigma_{A_\ell} \right).
%	\end{align}
%\end{lemma}
%
%\begin{proof}
%	The additivity of quantum relative entropy implies that, for all $\sigma_{\ell} \in \mathcal{S}(\mathsf{A}_{\ell}), \, \ell\in[L]$,
%	\begin{align}
%		D\left(\rho_{A_1A_2\ldots A_L} \Vert \otimes_{\ell \in [L]} \sigma_{A_\ell} \right)
%		&= D\left(\rho_{A_1A_2\ldots A_L} \Vert \otimes_{\ell \in [L]} \rho_{A_\ell} \right) + D\left( \otimes_{\ell \in [L]} \rho_{A_\ell} \Vert \otimes_{\ell \in [L]} \sigma_{A_\ell} \right) \\
%		&= D\left(\rho_{A_1A_2\ldots A_L} \Vert \otimes_{\ell \in [L]} \rho_{A_\ell} \right) + \sum\nolimits_{\ell \in [L] } D\left(\rho_{A_\ell} \, \Vert \, \sigma_{A_\ell} \right).
%	\end{align}
%	Then, by minimizing all $\sigma_{A_\ell}$ for all $\ell \in [L]$ and the positive definiteness of quantum relative entropy, we prove the claim.
%\end{proof}
%

\begin{lemma_convexity}[Convexity] %\label{lemma:convexity}
	Let $L$ be any integer and $\mathcal{I}$ be any finite set.
	Let $\rho_{A_1A_2\ldots A_L E} := \sum_{i\in\mathcal{I}} p_i \rho_{A_1A_2\ldots A_L E}^{i} $ and $\tau_{A_\ell} =  \sum_{i\in\mathcal{I}} p_i \tau_{A_{\ell}}^{i} $, $\ell\in [L]$ be statistical mixtures of states for any $p_i>0$, $\sum_{i\in\mathcal{I}} p_i  = 1$.
	Then, the following holds for every $\alpha \geq \sfrac12$,
	\begin{align}
	\begin{split}
	I_\alpha( \rho_{A_1\ldots A_L E} \,\Vert \otimes_{\ell \in [L]}  \tau_{A_\ell} ) &\leq \sum\nolimits_{i\in\mathcal{I} } p_i I_\alpha\left(\rho_{A_1\ldots A_L E}^{i} \,\Vert \otimes_{\ell\in[L]} \tau_{A_\ell}^{i}\right) + L \cdot H(\{p_i\}_{i\in\mathcal{I}})
	\\ &\le \sum\nolimits_{i \in \mathcal{I} }p_i I_\alpha\left(\rho_{A_1\ldots A_L E}^{i} \,\Vert \otimes_{\ell\in[L]} \tau_{A_\ell}^{i}\right) + L \log |\mathcal{I}|.
	\end{split}
	\end{align}
	% % and
	% % \begin{align}
	% % \begin{split}
	% % I_\alpha\left( A_1 : \cdots : A_L \right)_\rho &\leq \sum\nolimits_{i\in\mathcal{I} } p_i I_\alpha\left( A_1 : \cdots : A_L \right)_{\rho^i} + L \cdot H(\{p_i\}_{i\in\mathcal{I}})
	% % \\ &\le \sum\nolimits_{i \in \mathcal{I} }p_i I_\alpha\left( A_1 : \cdots : A_L \right)_{\rho^i} + L \log |\mathcal{I}|.
	% % \end{split}
	% \end{align}	
	Here, $H(\{p_i\}_{i\in\mathcal{I}}) := - \sum_{i\in\mathcal{I}} p_i \log p_i$ denotes the Shannon entropy.
%	Moreover, 
\end{lemma_convexity}

\begin{proof}
	Below, the subscript $\ell$ runs over all $[L]$ and the subscript $i$ runs over all $\mathcal{I}$ if we do not specify it.
	For each $i\in\mathcal{I}$ and every $\eps>0$, we let $\sigma_E^{i} \in  \mathcal{S}( \mathsf{E})$ to be a state satisfying
	\begin{align}
    D_\alpha\left(\rho_{A_1\ldots A_L E}^{i} \,\Vert \otimes_{\ell} \tau_{A_{\ell}}^{i} \ten \sigma_E^{i}\right)
	\leq I_\alpha\left(\rho_{A_1\ldots A_L E}^i \,\Vert\, \otimes_{\ell} \tau_{A_\ell}^{i} \right) + \eps .
	\end{align}
	Then, by the definition of the sandwiched R\'enyi divergence, we have
	\begin{align}
	&\sum\nolimits_i p_i D_\alpha\left(\rho_{A_1\ldots A_L E}^{i} \,\Vert \otimes_{\ell} \tau_{A_{\ell}}^{i} \ten \sigma_E^{i}\right) - L \cdot p_i \log p_i \\
	&=  \sum\nolimits_i p_i D_\alpha\left(\rho_{A_1\ldots A_L E}^{i} \,\Vert\,   p_i \tau_{A_{1}}^i \ten \cdots \ten p_i \tau_{A_{L}}^i  \ten \sigma_E^{i}\right) \\
	&\overset{\textnormal{(a)}}{\geq} D_\alpha\left( \sum\nolimits_i p_i\rho_{A_1\ldots A_L E}^{i} \,\Vert \sum\nolimits_i p_i \tau_{A_{1}}^i \ten \cdots \ten p_i \tau_{A_{L}}^i  \ten p_i \sigma_E^{i}\right)  \\
	&\overset{\textnormal{(b)}}{\geq} D_\alpha\left( \sum\nolimits_i p_i\rho_{A_1\ldots A_L E}^{i} \,\Vert \sum\nolimits_{i_1} p_{i_1} \tau_{A_{1}}^{i_1} \ten \cdots \ten \sum\nolimits_{i_L} p_{i_L} \tau_{A_{L}}^{i_1}  \ten \sum\nolimits_i p_i \sigma_E^{i}\right)  \\
	&= D_\alpha\left( \rho_{A_1\ldots A_L E} \,\Vert  \tau_{A_{1}} \ten \cdots \ten  \tau_{A_{L}} \ten \sum\nolimits_i p_i \sigma_E^{i}\right)  \\
	&\geq I_\alpha( \rho_{A_1\ldots A_L E} \,\Vert \otimes_{\ell \in [L]}  \tau_{A_\ell} ),
	\end{align}
	where (a) follows from the convexity of the sandwiched R\'enyi divergence \cite{MDS+13, WWY14, Bei13, FL13}, \cite[Theorem 3.16]{MO17},
	and (b) follows from the monotonically non-increasing map $\sigma \mapsto D_\alpha(\rho\,\Vert\,\sigma)$ \cite[Proposition 4]{MDS+13}, \cite[Lemma 3.24]{MO17}, and for each $\ell \in [L]$, $ p_i \tau_{A_\ell}^i \leq \sum\nolimits_{ i_\ell  \in [L]} p_{i_\ell} \tau_{A_\ell}^{{i}_\ell}$.
	By letting $\eps\to0$, we conclude the proof.
\end{proof}

\begin{fact_exponent}[Properties of error-exponent function] %\label{fact:properties_exponent}
For any multipartite state $\rho_{A_1 A_2 \ldots A_L B} $ and any quantum broadcast channel $\mathscr{N}_{A\to B_1 B_2 \ldots B_L}$, the error-exponent function satisfies the following properties.
\begin{enumerate}[(i)]
	\item%\label{item:positivity_exponent} 
    (Positivity) For any $r>0$,
    \begin{align}
        E_r (\rho_{A_1 A_2 \ldots A_L B} \, \Vert \otimes_{\ell\in[L]} \tau_{A_\ell}) > 0 &\Longleftrightarrow r > I(\rho_{A_1 A_2 \ldots A_L B} \, \Vert \otimes_{\ell\in[L]} \tau_{A_\ell}); \\
        E_r (\mathscr{N}_{A\to B_1 B_2 \ldots B_L}) > 0 &\Longleftrightarrow r > I( \mathscr{N}_{A\to B_1 B_2 \ldots B_L} ).
    \end{align}

	\item%\label{item:additivity_exponent} 
    (Additivity) 
    For any integer $n \in \mathds{N}$ and $r>0$,
	\begin{align}
	E_{nr} \left(\rho_{A_1 A_2 \ldots A_L B}^{\otimes n} \, \Vert \otimes_{\ell\in[L]} \tau_{A_\ell}^{\otimes n} \right)
	= n E_{{r}}\left(\rho_{A_1 A_2 \ldots A_L B} \, \Vert \otimes_{\ell\in[L]} \tau_{A_\ell} \right). 
	\end{align}

 \item%\label{item:minimax_exponnet} 
 (A minimax identity and saddle-point)
    Provided that the underlying Hilbert spaces are all finite dimensional, for any $r>0$,
    there exist a saddle-point $(\psi,\alpha) \in \mathcal{S}(\mathsf{A}\otimes \mathsf{R}) \times[1,2] $\cite[\S 36]{Roc70} such that
    \begin{align}
    E_r(\mathscr{N}_{A\to B_1 B_2 \ldots B_L}) &= 
     \inf_{\psi_{AR}} E_r(B_1: B_2: \cdots B_L: R)_{(\mathscr{N}\otimes \id_R)(\psi)} \\
     &= \frac{\alpha-1}{\alpha} \left( r - I_\alpha(B_1:B_2\cdots B_L: R)_{(\mathscr{N}\otimes \id_R)(\psi)} \right).
    \end{align}
	% \begin{align}
	% \inf_{ \tau_{A_1}\in\mathcal{S}(\mathsf{A}_1), \ldots, \tau_{A_L}\in\mathcal{S}(\mathsf{A}_L) }
	% E_r(\rho_{A_1 \ldots A_L} \, \Vert\, \otimes_{\ell\in[L]} \tau_{A_\ell}  )
	% = E_{{r}}(A_1 : \cdots : A_L)_\rho .
	% \end{align}
	
	\item%\label{item:limiting_exponnet} 
 (Limiting behavior)
    Provided that the underlying Hilbert spaces are all finite dimensional, then    
	for any sequence $r_n := I\left(\mathscr{N}_{A\to B_1 B_2\ldots B_L}\right) + a_n$ satisfying  
    $a_n \downarrow 0$,
	we have
    \begin{align}
	\liminf_{n\to \infty}\frac{E_{r_n}(\mathscr{N})}{a_n^2} 
	&\geq \frac{1}{2V(\mathscr{N}_{A\to B_1 B_2\ldots B_L})}.
    \end{align}
\end{enumerate}

\end{fact_exponent}

\begin{proof}
    Items~\ref{item:positivity_exponent}, and \ref{item:additivity_exponent} directly follows from Proposition~\ref{fact:properties_sandwiched}-\ref{item:limiting_sandwiched} and \ref{item:additivity_sandwiched}.

    By Proposition~\ref{item:concavity_sandwiched} and \ref{item:concavity_p_sandwiched},  the objective function
    \begin{align}
        (\rho_A, \alpha) \mapsto \frac{\alpha-1}{\alpha} \left( r - I_{\alpha}(B:R)_{\mathscr{N}(\psi_{AR}^\rho)} \right)
    \end{align}
    is lower-semicontinuous and convex in $\rho_A \in \mathcal{S}(\mathsf{A})$ for each $\alpha\geq1$ (where $\psi_{AR}^\rho$ is a purification of $\rho_A$), and upper-semicontinuous and concave in $\alpha \in [1,2]$ for each $\rho_A \in \mathcal{S}(\mathsf{A})$
    Moreover, under the finite-dimension assumption, the sets $\mathcal{S}(\mathsf{A})$ and $[1,2]$ are both convex and compact, and the convex-concave objective function is always finite and bounded.
    Hence, the assertion of the saddle-point with its saddle-value being the error-exponent function of channel $\mathscr{N}_{A\to B_1 B_2\ldots B_L}$ follows from \cite[Theorem 36.3]{Roc70}.

    Next, we prove Item~\ref{item:limiting_exponnet}.
    Note that the error-exponent function is given by the sandwiched R\'enyi information. We will slightly weaken it via the Petz's R\'enyi information \cite{Pet86, CHT19, CGH18}.
    It will still lead to our goal since both the two R\'enyi information quantities have the same limiting behavior as $\alpha\to 1$. Such a relaxation will simplify the analysis due to the closed-form expression of Petz's R\'enyi information \cite{SW12, Tom16, CGH18}.
    Moreover, we prove the unipartite case (i.e.~$L=1$) as follows; the multipartite case for general $L\in\mathds{N}$ follows straightforwardly.

For each $n\in\mathds{N}$, by Proposition~\ref{fact:properties_exponent}-\ref{item:minimax_exponnet}, we let $(\rho_n,\alpha_n) \in  \mathcal{S}(\mathsf{A})\times [1,2]$ be
a saddle-point of the function
\begin{align}
    (\rho,\alpha)\mapsto 
    \frac{\alpha-1}{\alpha} \left( r_n - I_{\alpha_n}(B:R)_{\mathscr{N}(\psi_{AR}^\rho)} \right),
\end{align}
such that $\psi_n = \psi_{AR}^\rho$ is a purification of $\rho_{A}$, and 
\begin{align}
E_{r_n}(\mathscr{N}) &= \frac{\alpha_n-1}{\alpha_n} \left( r_n - I_{\alpha_n} (B:R)_{\mathscr{N}(\psi_n)} \right) \\
&\geq \frac{\alpha-1}{\alpha} \left( r_n - I_{\alpha}(B:R)_{\mathscr{N}(\psi_n)} \right), \forall\, \alpha \in [1,2]. \label{eq:moderate_minimax}
\end{align}
To ease the burden of notation, let us denote
\begin{align}
    I &:= I(\mathscr{N}), \\
    I_n &:= I(B:R)_{\mathscr{N}(\psi_n)}, \\
    V_n &:= V(B:R)_{\mathscr{N}(\psi_n)}.
    % , \\
    % \alpha_n &:= 1+\frac{I-I_n-a_n}{V_n}.
\end{align}
Note that since $r_n \to I$, by Proposition~\ref{fact:properties_exponent}-\ref{item:positivity_exponent}, 
we must have
\begin{align}
    \lim_{n\to \infty} \alpha_n  = 1.
\end{align}
On the other hand,
\begin{align}
    I_{\alpha_n}(B:R)_{\mathscr{N}(\psi_n)} \geq  I_{\alpha_n}(B:R)_{\mathscr{N}(\psi)}, \quad \forall\, \psi_{AR}.
\end{align}
This guarantees that $\psi_n$ converges to some (channel-capacity achieving) pure state, say $\psi_\infty$, that satisfies
\begin{align} \label{eq:I_n}
    \lim_{n\to \infty} I_{\alpha_n}(B:R)_{\mathscr{N}(\psi_n)} = I(B:R)_{\mathscr{N}(\psi_\infty)} = I.
\end{align}

We introduce the Petz-type quantities \cite{Pet86}:
\begin{align}
	\bar{D}_\alpha (\rho\,\Vert\, \sigma) &:= \frac{1}{\alpha-1}\log \Tr\left[ \rho^\alpha \sigma^{1-\alpha}\right], \\
	\bar{I}_\alpha(B:R)_{\mathscr{N}(\psi)} &:= \inf_{\sigma_{R}\in\mathcal{S}(\mathsf{R})} \bar{D}_\alpha\left( \mathscr{N}_{A\to B} (\psi_{AR}) \,\Vert\, \mathscr{N}_{A\to B} (\psi_{A}) \otimes \sigma_R \right).
	% \bar{I}_\alpha(\mathscr{N}_{A\to B}) &:= \sup_{\psi_{AR}} \bar{I}_\alpha(B:R)_{\mathscr{N}(\psi)}.
\end{align}
(One can also work with a R\'enyi information without minimizing $\sigma_{R}$ \cite{Tom16, CH17} since they all have the same Taylor's series expansion aroudn $\alpha = 1$.)

Now, we are all set for the proof.
From Eq.~\eqref{eq:moderate_minimax}, we choose a specific
\begin{align}
	\bar{\alpha}_n :=  1 + \frac{I-I_n+a_n}{V_n},
\end{align}
to replace $\alpha_n$, which yields a lower bound to the error-exponent function of channel.
We calculate
\begin{align}
	E_{r_n}(\mathscr{N})
	&= \frac{\alpha_n-1}{\alpha_n} \left( r_n - I_{\alpha_n}(B:R)_{\mathscr{N}(\psi_n)} \right)\\
	&\geq \frac{I-I_n+a_n}{V_n\left( 1 + \frac{I-I_n+a_n}{V_n} \right) }
	\left( I + a_n - I_{\bar{\alpha}_n}(B:R)_{\mathscr{N}(\psi_n)} \right) \\
	&\overset{\textnormal{(a)}}{\geq} \frac{I-I_n+a_n}{V_n\left( 1 + \frac{I-I_n+a_n}{V_n} \right) }
	\left( I + a_n - \bar{I}_{\bar{\alpha}_n}(B:R)_{\mathscr{N}(\psi_n)} \right) \\
	&\overset{\textnormal{(b)}}{=} \frac{I-I_n+a_n}{V_n\left( 1 + \frac{I-I_n+a_n}{V_n} \right) }
	\left( I + a_n - I_n - \frac{I - I_n + a_n}{2} - \frac{(\bar{\alpha}_n-1)^2}{2}\mathfrak{R}_n \right), \\
	&= \frac{(I-I_n+a_n)^2}{2V_n} \frac{1}{1+\frac{I-I_n+a_n}{V_n}} \left( 1 - \frac{\bar{\alpha}_n-1}{V_n}  \mathfrak{R}_n \right)\\
	&\overset{\textnormal{(c)}}{\geq}  \frac{a_n^2}{2V_n} \frac{1}{1+\frac{I-I_n+a_n}{V_n}} \left( 1 - \frac{\bar{\alpha}_n-1}{V_n}  \Upsilon \right),
\label{eq:moderate_1}
% \\
% &\overset{\textnormal{(a)}}{\geq}  \frac{a_n^2}{2V_n} \frac{1}{1+\frac{I-I_n+a_n}{V_n}} \left( 1 - \frac{\bar{\alpha}_n-1}{V_n}  \Upsilon \right), 
% \label{eq:moderate_1}
\end{align}
where (a) is because $D_\alpha \leq \bar{D}_\alpha$ for all $\alpha \geq 1$ \cite{LT76}, \cite[Proposition 3.20]{MO17}; and in (b) we invoke Taylor's series expansion of the map $\alpha \mapsto \bar{I}_{\alpha}(B:R)_{\mathscr{N}(\psi)}$ at $\alpha=1$ \cite{LT15, HT14, CH17}:
\begin{align}
	\bar{I}_{\bar{\alpha}_n}(B:R)_{\mathscr{N}(\psi_n)} &= I_n + \frac{\bar{\alpha}_n-1}{2} V_n 
	+ \frac{(\bar{\alpha}_n-1)^2}{2} \mathfrak{R}_n, \\
%	% \mathfrak{R}_n(\alpha_n -1),
%\end{align}
%for some remainder function
%\begin{align}
	\mathfrak{R}_n &:= \left.\frac{\partial^2}{ \partial \alpha^2 } \bar{I}_{{\alpha}}(B:R)_{\mathscr{N}(\psi_n)}\right|_{\alpha = \tilde{\alpha}_n \in [1,\bar{\alpha}_n] }.
\end{align}
%and some $\bar{\alpha}_n \in [1,\alpha_n]$.
The last line (c) is due to the fact that $I - I_n\geq 0$ for all $n\in\mathds{N}$, and 
\begin{align} \label{eq:Upsilon}
	\Upsilon := \max_{ (\psi,\alpha) \in \mathcal{S}(\mathsf{A}\otimes \mathsf{R}) \times[1,2]  } \left| \frac{\partial^2}{ \partial \alpha^2 } I_{{\alpha}}(B:R)_{\mathscr{N}(\psi)}\right|.
\end{align}
As already pointed out in \cite{CH17}, such a factor is finite by the finite-dimensional assumption of the underlying Hilbert space $\mathsf{A}$, $\mathsf{B}$, and $\mathsf{R}$, the uniform continuity of the quantity in \eqref{eq:Upsilon}, which holds because of the closed-form expression of the Petz's R\'enyi information \cite{SW12, Tom16, CGH18}.

Note that $\bar{\alpha}_n \to 1$ by Eq.~\eqref{eq:I_n}. 
We have
% In the end, we want
\begin{align}
    \liminf_{n\to \infty}\frac{E_{r_n}(\mathscr{N})}{a_n^2} 
    &\geq \liminf_{n\to\infty} \frac{1}{2V(B:R)_{\mathscr{N}(\psi_n)}} \\
    &\geq \frac{1}{2V(\mathscr{N})},
\end{align}
where where the last line follows from the continuity of $\psi \mapsto V(B:R)_{\mathscr{N}(\psi)}$ \cite{TV15, CHT17, CH17}, the fact that $\psi_n \to \psi_\infty$ and Eq.~\eqref{eq:I_n}, and the definition of $V(\mathscr{N})$ given in Eq.~\eqref{eq:V}. 
We conclude the proof
% This yields the desired Eq.~\eqref{eq:moderate_goal_uni0}.
\end{proof}

\section{The Post-Selection Technique} \label{sec:Post-Selection}

% {\color{red} Copied from Mario's paper. Need to modify and shorten.}

% We use a norm on the set of CPTP maps which essentially measures the probability by which two such mappings can be distinguished. The norm is known as diamond norm in quantum information theory~\cite{Kitaev97}. Here, we present it in a formulation which highlights that it is dual to the well-known completely bounded (cb) norm~\cite{Paulsen}. 

% \begin{definition}[Diamond norm]
% 	Let $\cE_{A}:\cL(\cH_{A})\mapsto\cL(\cH_{B})$ be a linear map. The \emph{diamond norm} of $\cE_{A}$ is defined as
% 	\begin{align}
% 		\|\cE_{A}\|_{\diamond}=\sup_{k\in\mathbb{N}}\|\cE_{A}\otimes\cI_{k}\|_{1}\ ,
% 	\end{align}
% 	where $\|\cF\|_{1}=\sup_{\sigma\in\cS_{\leq}(\cH)}\|\cF(\sigma)\|_{1}$ and $\cI_{k}$ denotes the identity map on states of a $k$-dimensional quantum system.
% 	\label{kitaev}
% \end{definition}

% \begin{proposition}[\hspace{1sp}{\cite{Kitaev97, Paulsen}}]
% 	The supremum in Definition~\ref{kitaev} is reached for $k=|A|$. Furthermore the diamond norm defines a norm on the set of CPTP maps.
% 	\label{feeling}
% \end{proposition}

% \noindent
% Two CPTP maps $\cE$ and $\cF$ are called $\eps$-close if they are $\eps$-close in the metric induced by the diamond norm.

% \noindent
% \begin{definition}[De Finetti states]
% 	Let $\sigma\in\cS_{=}(\cH)$ and $\mu(.)$ be a probability measure on $\cS_{=}(\cH)$. Then
% 	\begin{align}
% 		\zeta^{n}=\int\sigma^{\otimes n}\mu(\sigma)\in\cS_{=}(\cH^{\otimes n})
% 	\end{align}
% 	is called \emph{de Finetti state}.
% 	\label{definetti}
% \end{definition}

% \noindent
The following facts are the well-known \textit{Post-Selection Technique} \cite{ChristKoenRennerPostSelect} also known as \textit{de Finetti reductions} or {\it universal state} \cite{Hay09}, and the formulation below is taken from \cite[Appendix D]{BCR11}.

\begin{fact}[\hspace{1sp}{\cite{ChristKoenRennerPostSelect}}]
	Let $\eps>0$ and $\cE^{n}_{A}$ and $\cF^{n}_{A}$ be completely positive and trace-preserving maps from $\mathcal{B}(\mathsf{A}^{\otimes n})$ to $\mathcal{B}(\mathsf{B})$. If there exists a CPTP map $K_{\pi}$ for any permutation $\pi$ such that $(\cE^{n}_{A}-\cF^{n}_{A})\circ\pi=K_{\pi}\circ(\cE^{n}_{A}-\cF^{n}_{A})$, then $\cE^{n}_{A}$ and $\cF^{n}_{A}$ are $\eps$-close in diamond norm whenever
	\begin{align}
		\left\|((\cE^{n}_{A}-\cF^{n}_{A})\otimes\cI_{RR'})(\zeta^{n}_{ARR'})\right\|_{1}\leq\eps(n+1)^{-(|\mathsf{A}|^{2}-1)}\ ,
	\end{align}
	where $\zeta^{n}_{ARR'}$ is a purification of the de Finetti state $\zeta_{AR}^{n}=\int\psi_{AR}^{\otimes n}\, \mathrm{d}(\psi_{AR})$ with $\psi_{AR}$ is a pure state on $\mathsf{A}\otimes \mathsf{R}$ with $\mathsf{A}\cong\mathsf{R}$, and $\mathrm{d}(\cdot)$ is the measure on the normalized pure states on $\mathsf{A}\otimes \mathsf{R}$ induced by the Haar measure on the unitary group acting on $\mathsf{A}\otimes \mathsf{R}$, normalized to $\int \mathrm{d}(\cdot)=1$. Furthermore, we can assume without loss of generality that $|\mathsf{R}'|\leq(n+1)^{|\mathsf{A}|^{2}-1}$.
	\label{posti}
\end{fact}

% \begin{fact}[\hspace{1sp}{\cite[Carath\'{e}odory]{Gruber93}}]
% 	Let $d\in\mathbb{N}$ and $x$ be a point that lies in the convex hull of a set $P$ of points in $\mathbb{R}^{d}$. Then there exists a subset $P'$ of $P$ consisting of
% 	$d+1$ or fewer point such that $x$ lies in the convex hull of $P'$.
% 	\label{cara}
% \end{fact}
By Carath\'eodory's theorem for convex hulls, one also has
\begin{fact}[\hspace{1sp}{\cite{Gruber93}}]
	Let $\zeta_{AR}^{n}=\int\psi_{AR}^{\otimes n}\, \mathrm{d}(\psi_{AR})$ as in Fact~\ref{posti}. Then, we can write $\zeta_{AR}^{n}=\sum_{i}p_{i}\left(\psi^{i}_{AR}\right)^{\otimes n}$ with $\omega^{i}_{AR}$ being a pure state on $\mathsf{A}\otimes \mathsf{R}$, $i\in\{1,2,\ldots,(n+1)^{2|\mathsf{A}||\mathsf{R}|-2}\}$, and $\{p_{i}\}_i$ a probability distribution.
	\label{mario}
\end{fact}

% \begin{proof}
% 	We can think of $\zeta_{AR}^{n}$ as a normalized state on the symmetric subspace $\mathrm{Sym}^{n}(\cH_{AR})\subset\cH_{AR}^{\otimes n}$. The dimension of $\mathrm{Sym}^{n}(\cH_{AR})$ is bounded by $k=(n+1)^{|A||R|-1}$. Furthermore the normalized states on $\mathrm{Sym}^{n}(\cH_{AR})$ can be seen as living in an $m$-dimensional real vector space where $m=k-1+2\cdot\frac{k(k-1)}{2}=k^{2}-1$. Now define $S$ as the set of all $\xi_{AR}^{n}=\omega_{AR}^{\otimes n}$, where $\omega_{AR}=\proj{\omega}_{AR}\in\cS_{=}(\cH_{AR})$. Then $\zeta^{n}_{AR}$ lies in the convex hull of the set $S\subset\mathbb{R}^{k^{2}-1}$. Using Carath\'eodory's theorem (Theorem~\ref{cara}), we have that $\zeta_{AR}^{n}$ lies in the convex hull of a set $S'\subset S$ where $S'$ consists of at most $p=k^{2}-1+1=k^{2}$ points. Hence we can write $\zeta_{AR}^{n}$ as a convex combination of $p=(n+1)^{2|A||R|-2}$ extremal points in $S'$, i.e.~$\zeta_{AR}^{n}=\sum_{i}p_{i}(\omega^{i}_{AR})^{\otimes n}$, where $\omega^{i}_{AR}=\proj{\omega^{i}}_{AR}\in\cS_{=}(\cH_{AR})$, $i\in\{1,2,\ldots,(n+1)^{2|A||R|-2}\}$ and $p_{i}$ a probability distribution.
% \end{proof}

%%%%%%%%%%%%%%%%%%%%%%%%%%%%%%%%

{\larger
\bibliographystyle{myIEEEtran}
\bibliography{reference.bib, reference2}

 \newcommand{\noop}[1]{}
\begin{thebibliography}{100}
\providecommand{\url}[1]{#1}
\csname url@samestyle\endcsname
\providecommand{\newblock}{\relax}
\providecommand{\bibinfo}[2]{#2}
\providecommand{\BIBentrySTDinterwordspacing}{\spaceskip=0pt\relax}
\providecommand{\BIBentryALTinterwordstretchfactor}{4}
\providecommand{\BIBentryALTinterwordspacing}{\spaceskip=\fontdimen2\font plus
\BIBentryALTinterwordstretchfactor\fontdimen3\font minus
  \fontdimen4\font\relax}
\providecommand{\BIBforeignlanguage}[2]{{%
\expandafter\ifx\csname l@#1\endcsname\relax
\typeout{** WARNING: IEEEtran.bst: No hyphenation pattern has been}%
\typeout{** loaded for the language `#1'. Using the pattern for}%
\typeout{** the default language instead.}%
\else
\language=\csname l@#1\endcsname
\fi
#2}}
\providecommand{\BIBdecl}{\relax}
\BIBdecl

\bibitem{BSS+02}
C.~Bennett, P.~Shor, J.~Smolin, and A.~Thapliyal, ``Entanglement-assisted
  capacity of a quantum channel and the reverse {Shannon} theorem,''
  \href{http://dx.doi.org/10.1109/tit.2002.802612}{\emph{{IEEE} Transactions on
  Information Theory}},
  \href{http://dx.doi.org/10.1109/tit.2002.802612}{vol.~48},
  \href{http://dx.doi.org/10.1109/tit.2002.802612}{no.~10},
  \href{http://dx.doi.org/10.1109/tit.2002.802612}{pp. 2637--2655},
  \href{http://dx.doi.org/10.1109/tit.2002.802612}{oct 2002}.

\bibitem{BDH+14}
C.~H. {Bennett}, I.~{Devetak}, A.~W. {Harrow}, P.~W. {Shor}, and A.~{Winter},
  ``The quantum reverse {Shannon} theorem and resource tradeoffs for simulating
  quantum channels,''
  \href{http://dx.doi.org/10.1109/TIT.2014.2309968}{\emph{IEEE Transactions on
  Information Theory}},
  \href{http://dx.doi.org/10.1109/TIT.2014.2309968}{vol.~60},
  \href{http://dx.doi.org/10.1109/TIT.2014.2309968}{no.~5},
  \href{http://dx.doi.org/10.1109/TIT.2014.2309968}{pp. 2926--2959},
  \href{http://dx.doi.org/10.1109/TIT.2014.2309968}{May 2014}.

\bibitem{BCR11}
M.~Berta, M.~Christandl, and R.~Renner, ``The quantum reverse {Shannon} theorem
  based on one-shot information theory,''
  \href{http://dx.doi.org/10.1007/s00220-011-1309-7}{\emph{Communications in
  Mathematical Physics}},
  \href{http://dx.doi.org/10.1007/s00220-011-1309-7}{vol. 306},
  \href{http://dx.doi.org/10.1007/s00220-011-1309-7}{no.~3},
  \href{http://dx.doi.org/10.1007/s00220-011-1309-7}{pp. 579--615},
  \href{http://dx.doi.org/10.1007/s00220-011-1309-7}{aug 2011}.

\bibitem{BSS+99}
C.~H. Bennett, P.~W. Shor, J.~A. Smolin, and A.~V. Thapliyal,
  ``Entanglement-assisted classical capacity of noisy quantum channels,''
  \href{http://dx.doi.org/10.1103/physrevlett.83.3081}{\emph{Physical Review
  Letters}}, \href{http://dx.doi.org/10.1103/physrevlett.83.3081}{vol.~83},
  \href{http://dx.doi.org/10.1103/physrevlett.83.3081}{no.~15},
  \href{http://dx.doi.org/10.1103/physrevlett.83.3081}{pp. 3081--3084},
  \href{http://dx.doi.org/10.1103/physrevlett.83.3081}{oct 1999}.

\bibitem{Hol02}
A.~S. Holevo, ``On entanglement-assisted classical capacity,''
  \href{http://dx.doi.org/10.1063/1.1495877}{\emph{Journal of Mathematical
  Physics}}, \href{http://dx.doi.org/10.1063/1.1495877}{vol.~43},
  \href{http://dx.doi.org/10.1063/1.1495877}{no.~9},
  \href{http://dx.doi.org/10.1063/1.1495877}{pp. 4326--4333},
  \href{http://dx.doi.org/10.1063/1.1495877}{sep 2002}.

\bibitem{Sha48}
C.~E. Shannon, ``A mathematical theory of communication,'' \emph{The Bell
  System Technical Journal}, vol.~27, pp. 379--423, 1948.

\bibitem{haddadpour2016simulation}
F.~Haddadpour, M.~H. Yassaee, S.~Beigi, A.~Gohari, and M.~R. Aref, ``Simulation
  of a channel with another channel,''
  \href{http://dx.doi.org/10.1109/TIT.2016.2635660}{\emph{IEEE Transactions on
  Information Theory}},
  \href{http://dx.doi.org/10.1109/TIT.2016.2635660}{vol.~63},
  \href{http://dx.doi.org/10.1109/TIT.2016.2635660}{no.~5},
  \href{http://dx.doi.org/10.1109/TIT.2016.2635660}{pp. 2659--2677},
  \href{http://dx.doi.org/10.1109/TIT.2016.2635660}{2016}.

\bibitem{Sudan20}
M.~Sudan, H.~Tyagi, and S.~Watanabe, ``Communication for generating
  correlation: A unifying survey,''
  \href{http://dx.doi.org/10.1109/TIT.2019.2946364}{\emph{IEEE Transactions on
  Information Theory}},
  \href{http://dx.doi.org/10.1109/TIT.2019.2946364}{vol.~66},
  \href{http://dx.doi.org/10.1109/TIT.2019.2946364}{no.~1},
  \href{http://dx.doi.org/10.1109/TIT.2019.2946364}{pp. 5--37},
  \href{http://dx.doi.org/10.1109/TIT.2019.2946364}{2020}.

\bibitem{Wyn73}
A.~Wyner, ``A theorem on the entropy of certain binary sequences and
  applications--{II},''
  \href{http://dx.doi.org/10.1109/tit.1973.1055108}{\emph{{IEEE} Transactions
  on Information Theory}},
  \href{http://dx.doi.org/10.1109/tit.1973.1055108}{vol.~19},
  \href{http://dx.doi.org/10.1109/tit.1973.1055108}{no.~6},
  \href{http://dx.doi.org/10.1109/tit.1973.1055108}{pp. 772--777},
  \href{http://dx.doi.org/10.1109/tit.1973.1055108}{nov 1973}.

\bibitem{Wyn75}
A.~D. Wyner, ``On source coding with side information at the decoder,''
  \href{http://dx.doi.org/10.1109/tit.1975.1055374}{\emph{{IEEE} Transactions
  on Information Theory}},
  \href{http://dx.doi.org/10.1109/tit.1975.1055374}{vol.~21},
  \href{http://dx.doi.org/10.1109/tit.1975.1055374}{no.~3},
  \href{http://dx.doi.org/10.1109/tit.1975.1055374}{pp. 294--300},
  \href{http://dx.doi.org/10.1109/tit.1975.1055374}{May 1975}.

\bibitem{Wyner75}
A.~Wyner, ``The common information of two dependent random variables,''
  \href{http://dx.doi.org/10.1109/TIT.1975.1055346}{\emph{IEEE Transactions on
  Information Theory}},
  \href{http://dx.doi.org/10.1109/TIT.1975.1055346}{vol.~21},
  \href{http://dx.doi.org/10.1109/TIT.1975.1055346}{no.~2},
  \href{http://dx.doi.org/10.1109/TIT.1975.1055346}{pp. 163--179},
  \href{http://dx.doi.org/10.1109/TIT.1975.1055346}{1975}.

\bibitem{GK73}
P.~G{\'a}cs and J.~K{\"o}rner, ``Common information is far less than mutual
  information,'' \emph{Probl. Contr lnform. Theory}, vol.~2, no.~2, pp.
  149--162, 1973.

\bibitem{Cuf08}
P.~Cuff, ``Communication requirements for generating correlated random
  variables,'' in \emph{2008 {IEEE} International Symposium on Information
  Theory}.\hskip 1em plus 0.5em minus 0.4em\relax {IEEE}, jul 2008.

\bibitem{Yu22}
L.~Yu and V.~Y.~F. Tan, ``Common information, noise stability, and their
  extensions,'' \href{http://dx.doi.org/10.1561/0100000122}{\emph{Foundations
  and Trends{\textregistered} in Communications and Information Theory}},
  \href{http://dx.doi.org/10.1561/0100000122}{vol.~19},
  \href{http://dx.doi.org/10.1561/0100000122}{no.~2},
  \href{http://dx.doi.org/10.1561/0100000122}{pp. 107--389},
  \href{http://dx.doi.org/10.1561/0100000122}{2022}.

\bibitem{HV93}
T.~Han and S.~Verdu, ``Approximation theory of output statistics,''
  \href{http://dx.doi.org/10.1109/18.256486}{\emph{{IEEE} Transactions on
  Information Theory}}, \href{http://dx.doi.org/10.1109/18.256486}{vol.~39},
  \href{http://dx.doi.org/10.1109/18.256486}{no.~3},
  \href{http://dx.doi.org/10.1109/18.256486}{pp. 752--772},
  \href{http://dx.doi.org/10.1109/18.256486}{may 1993}.

\bibitem{HV93b}
T.~S. Han and S.~Verd{\'u}, ``\BIBforeignlanguage{English (US)}{Spectrum
  invariancy under output approximation full-rank discrete memoryless
  channels},'' \emph{\BIBforeignlanguage{English (US)}{Problemy Peredachi
  Informatsii}}, no.~2, pp. 9--27, Apr. 1993.

\bibitem{SV94}
Y.~Steinberg and S.~Verd{\'u}, ``Channel simulation and coding with side
  information,'' \href{http://dx.doi.org/10.1109/18.335877}{\emph{{IEEE}
  Transactions on Information Theory}},
  \href{http://dx.doi.org/10.1109/18.335877}{vol.~40},
  \href{http://dx.doi.org/10.1109/18.335877}{no.~3},
  \href{http://dx.doi.org/10.1109/18.335877}{pp. 634--646},
  \href{http://dx.doi.org/10.1109/18.335877}{may 1994}.

\bibitem{Hay06_resolvability}
M.~Hayashi, ``General nonasymptotic and asymptotic formulas in channel
  resolvability and identification capacity and their application to the
  wiretap channel,''
  \href{http://dx.doi.org/10.1109/tit.2006.871040}{\emph{{IEEE} Transactions on
  Information Theory}},
  \href{http://dx.doi.org/10.1109/tit.2006.871040}{vol.~52},
  \href{http://dx.doi.org/10.1109/tit.2006.871040}{no.~4},
  \href{http://dx.doi.org/10.1109/tit.2006.871040}{pp. 1562--1575},
  \href{http://dx.doi.org/10.1109/tit.2006.871040}{apr 2006}.

\bibitem{Cuf13}
P.~Cuff, ``Distributed channel synthesis,''
  \href{http://dx.doi.org/10.1109/tit.2013.2279330}{\emph{{IEEE} Transactions
  on Information Theory}},
  \href{http://dx.doi.org/10.1109/tit.2013.2279330}{vol.~59},
  \href{http://dx.doi.org/10.1109/tit.2013.2279330}{no.~11},
  \href{http://dx.doi.org/10.1109/tit.2013.2279330}{pp. 7071--7096},
  \href{http://dx.doi.org/10.1109/tit.2013.2279330}{nov 2013}.

\bibitem{cuff2010coordination}
P.~W. Cuff, H.~H. Permuter, and T.~M. Cover, ``Coordination capacity,''
  \emph{IEEE Transactions on Information Theory}, vol.~56, no.~9, pp.
  4181--4206, 2010.

\bibitem{Yassaee15}
M.~H. Yassaee, A.~Gohari, and M.~R. Aref, ``Channel simulation via interactive
  communications,''
  \href{http://dx.doi.org/10.1109/TIT.2015.2428236}{\emph{IEEE Transactions on
  Information Theory}},
  \href{http://dx.doi.org/10.1109/TIT.2015.2428236}{vol.~61},
  \href{http://dx.doi.org/10.1109/TIT.2015.2428236}{no.~6},
  \href{http://dx.doi.org/10.1109/TIT.2015.2428236}{pp. 2964--2982},
  \href{http://dx.doi.org/10.1109/TIT.2015.2428236}{2015}.

\bibitem{YC19}
S.~Yagli and P.~Cuff, ``Exact exponent for soft covering,''
  \href{http://dx.doi.org/10.1109/tit.2019.2917182}{\emph{{IEEE} Transactions
  on Information Theory}},
  \href{http://dx.doi.org/10.1109/tit.2019.2917182}{vol.~65},
  \href{http://dx.doi.org/10.1109/tit.2019.2917182}{no.~10},
  \href{http://dx.doi.org/10.1109/tit.2019.2917182}{pp. 6234--6262},
  \href{http://dx.doi.org/10.1109/tit.2019.2917182}{oct 2019}.

\bibitem{CG22}
H.-C. Cheng and L.~Gao, ``Error exponent and strong converse for quantum soft
  covering,'' \emph{arXiv:2202.10995 [quant-ph]}, 2022.

\bibitem{Wyn75c}
A.~D. Wyner, ``The wire-tap channel,''
  \href{http://dx.doi.org/10.1002/j.1538-7305.1975.tb02040.x}{\emph{Bell System
  Technical Journal}},
  \href{http://dx.doi.org/10.1002/j.1538-7305.1975.tb02040.x}{vol.~54},
  \href{http://dx.doi.org/10.1002/j.1538-7305.1975.tb02040.x}{no.~8},
  \href{http://dx.doi.org/10.1002/j.1538-7305.1975.tb02040.x}{pp. 1355--1387},
  \href{http://dx.doi.org/10.1002/j.1538-7305.1975.tb02040.x}{oct 1975}.

\bibitem{BL13}
M.~R. Bloch and J.~N. Laneman, ``Strong secrecy from channel resolvability,''
  \href{http://dx.doi.org/10.1109/tit.2013.2283722}{\emph{{IEEE} Transactions
  on Information Theory}},
  \href{http://dx.doi.org/10.1109/tit.2013.2283722}{vol.~59},
  \href{http://dx.doi.org/10.1109/tit.2013.2283722}{no.~12},
  \href{http://dx.doi.org/10.1109/tit.2013.2283722}{pp. 8077--8098},
  \href{http://dx.doi.org/10.1109/tit.2013.2283722}{dec 2013}.

\bibitem{PTM17}
M.~B. Parizi, E.~Telatar, and N.~Merhav, ``Exact random coding secrecy
  exponents for the wiretap channel,''
  \href{http://dx.doi.org/10.1109/tit.2016.2628307}{\emph{{IEEE} Transactions
  on Information Theory}},
  \href{http://dx.doi.org/10.1109/tit.2016.2628307}{vol.~63},
  \href{http://dx.doi.org/10.1109/tit.2016.2628307}{no.~1},
  \href{http://dx.doi.org/10.1109/tit.2016.2628307}{pp. 509--531},
  \href{http://dx.doi.org/10.1109/tit.2016.2628307}{jan 2017}.

\bibitem{Hay132}
M.~Hayashi, ``Quantum wiretap channel with non-uniform random number and its
  exponent and equivocation rate of leaked information,'' \emph{IEEE
  Transactions on Information Theory, Volume 61, Issue 10, 5595-5622}, 2015.

\bibitem{Hay15}
------, ``Quantum wiretap channel with non-uniform random number and its
  exponent and equivocation rate of leaked information,''
  \href{http://dx.doi.org/10.1109/tit.2015.2464215}{\emph{{IEEE} Transactions
  on Information Theory}},
  \href{http://dx.doi.org/10.1109/tit.2015.2464215}{vol.~61},
  \href{http://dx.doi.org/10.1109/tit.2015.2464215}{no.~10},
  \href{http://dx.doi.org/10.1109/tit.2015.2464215}{pp. 5595--5622},
  \href{http://dx.doi.org/10.1109/tit.2015.2464215}{oct 2015}.

\bibitem{Hay17}
------, \emph{Quantum Information Theory}.\hskip 1em plus 0.5em minus
  0.4em\relax Springer Berlin Heidelberg, 2017.

\bibitem{harsha2010communication}
P.~Harsha, R.~Jain, D.~McAllester, and J.~Radhakrishnan, ``The communication
  complexity of correlation,'' \emph{IEEE Transactions on Information Theory},
  vol.~56, no.~1, pp. 438--449, 2010.

\bibitem{Win02}
\BIBentryALTinterwordspacing
A.~Winter, ``Compression of sources of probability distributions and density
  operators,'' 2002. [Online]. Available:
  \url{https://www.arxiv.org/abs/quant-ph/0208131}
\BIBentrySTDinterwordspacing

\bibitem{Win04}
------, ``Extrinsic and intrinsic data in quantum measurements: Asymptotic
  convex decomposition of positive operator valued measures,''
  \href{http://dx.doi.org/10.1007/s00220-003-0989-z}{\emph{Communications in
  Mathematical Physics}},
  \href{http://dx.doi.org/10.1007/s00220-003-0989-z}{vol. 244},
  \href{http://dx.doi.org/10.1007/s00220-003-0989-z}{no.~1},
  \href{http://dx.doi.org/10.1007/s00220-003-0989-z}{pp. 157--185},
  \href{http://dx.doi.org/10.1007/s00220-003-0989-z}{jan 2004}.

\bibitem{WHB+12}
M.~M. Wilde, P.~Hayden, F.~Buscemi, and M.-H. Hsieh, ``The
  information-theoretic costs of simulating quantum measurements,''
  \href{http://dx.doi.org/10.1088/1751-8113/45/45/453001}{\emph{Journal of
  Physics A: Mathematical and Theoretical}},
  \href{http://dx.doi.org/10.1088/1751-8113/45/45/453001}{vol.~45},
  \href{http://dx.doi.org/10.1088/1751-8113/45/45/453001}{no.~45},
  \href{http://dx.doi.org/10.1088/1751-8113/45/45/453001}{p. 453001},
  \href{http://dx.doi.org/10.1088/1751-8113/45/45/453001}{oct 2012}.

\bibitem{Berta14}
M.~Berta, J.~M. Renes, and M.~M. Wilde, ``Identifying the information gain of a
  quantum measurement,''
  \href{http://dx.doi.org/10.1109/TIT.2014.2365207}{\emph{IEEE Transactions on
  Information Theory}},
  \href{http://dx.doi.org/10.1109/TIT.2014.2365207}{vol.~60},
  \href{http://dx.doi.org/10.1109/TIT.2014.2365207}{no.~12},
  \href{http://dx.doi.org/10.1109/TIT.2014.2365207}{pp. 7987--8006},
  \href{http://dx.doi.org/10.1109/TIT.2014.2365207}{2014}.

\bibitem{SV96}
Y.~Steinberg and S.~Verd{\'u}, ``Simulation of random processes and
  rate-distortion theory,''
  \href{http://dx.doi.org/10.1109/18.481779}{\emph{{IEEE} Transactions on
  Information Theory}}, \href{http://dx.doi.org/10.1109/18.481779}{vol.~42},
  \href{http://dx.doi.org/10.1109/18.481779}{no.~1},
  \href{http://dx.doi.org/10.1109/18.481779}{pp. 63--86},
  \href{http://dx.doi.org/10.1109/18.481779}{1996}.

\bibitem{LD09}
Z.~Luo and I.~Devetak, ``Channel simulation with quantum side information,''
  \emph{{IEEE} Transactions on Information Theory}, vol.~55, no.~3, pp.
  1331--1342, mar 2009.

\bibitem{DHW13}
N.~Datta, M.-H. Hsieh, and M.~M. Wilde, ``Quantum rate distortion, reverse
  shannon theorems, and source-channel separation,''
  \href{http://dx.doi.org/10.1109/tit.2012.2215575}{\emph{{IEEE} Transactions
  on Information Theory}},
  \href{http://dx.doi.org/10.1109/tit.2012.2215575}{vol.~59},
  \href{http://dx.doi.org/10.1109/tit.2012.2215575}{no.~1},
  \href{http://dx.doi.org/10.1109/tit.2012.2215575}{pp. 615--630},
  \href{http://dx.doi.org/10.1109/tit.2012.2215575}{jan 2013}.

\bibitem{HHH+03}
M.~Horodecki, K.~Horodecki, P.~Horodecki, R.~Horodecki, J.~Oppenheim,
  A.~Sen(De), and U.~Sen, ``Local information as a resource in distributed
  quantum systems,''
  \href{http://dx.doi.org/10.1103/physrevlett.90.100402}{\emph{Physical Review
  Letters}}, \href{http://dx.doi.org/10.1103/physrevlett.90.100402}{vol.~90},
  \href{http://dx.doi.org/10.1103/physrevlett.90.100402}{no.~10},
  \href{http://dx.doi.org/10.1103/physrevlett.90.100402}{mar 2003}.

\bibitem{HHH+05}
M.~Horodecki, P.~Horodecki, R.~Horodecki, J.~Oppenheim, A.~Sen(De), U.~Sen, and
  B.~Synak-Radtke, ``Local versus nonlocal information in quantum-information
  theory: Formalism and phenomena,''
  \href{http://dx.doi.org/10.1103/physreva.71.062307}{\emph{Physical Review
  A}}, \href{http://dx.doi.org/10.1103/physreva.71.062307}{vol.~71},
  \href{http://dx.doi.org/10.1103/physreva.71.062307}{no.~6},
  \href{http://dx.doi.org/10.1103/physreva.71.062307}{jun 2005}.

\bibitem{Dev05}
I.~Devetak, ``Distillation of local purity from quantum states,''
  \href{http://dx.doi.org/10.1103/physreva.71.062303}{\emph{Physical Review
  A}}, \href{http://dx.doi.org/10.1103/physreva.71.062303}{vol.~71},
  \href{http://dx.doi.org/10.1103/physreva.71.062303}{no.~6},
  \href{http://dx.doi.org/10.1103/physreva.71.062303}{jun 2005}.

\bibitem{KD07}
H.~Krovi and I.~Devetak, ``Local purity distillation with bounded classical
  communication,''
  \href{http://dx.doi.org/10.1103/physreva.76.012321}{\emph{Physical Review
  A}}, \href{http://dx.doi.org/10.1103/physreva.76.012321}{vol.~76},
  \href{http://dx.doi.org/10.1103/physreva.76.012321}{no.~1},
  \href{http://dx.doi.org/10.1103/physreva.76.012321}{jul 2007}.

\bibitem{DHW08}
I.~Devetak, A.~W. Harrow, and A.~J. Winter, ``A resource framework for quantum
  {Shannon} theory,''
  \href{http://dx.doi.org/10.1109/TIT.2008.928980}{\emph{IEEE Transactions on
  Information Theory}},
  \href{http://dx.doi.org/10.1109/TIT.2008.928980}{vol.~54},
  \href{http://dx.doi.org/10.1109/TIT.2008.928980}{no.~10},
  \href{http://dx.doi.org/10.1109/TIT.2008.928980}{pp. 4587--4618},
  \href{http://dx.doi.org/10.1109/TIT.2008.928980}{2008}.

\bibitem{CRB+22}
\BIBentryALTinterwordspacing
M.~X. Cao, N.~Ramakrishnan, M.~Berta, and M.~Tomamichel, ``Channel simulation:
  Finite blocklengths and broadcast channels,'' 2022. [Online]. Available:
  \url{https://www.arxiv.org/abs/2212.11666}
\BIBentrySTDinterwordspacing

\bibitem{fang2019quantum}
K.~Fang, X.~Wang, M.~Tomamichel, and M.~Berta, ``Quantum channel simulation and
  the channel's smooth max-information,'' \emph{IEEE Transactions on
  Information Theory}, vol.~66, no.~4, pp. 2129--2140, 2019.

\bibitem{RTB23}
N.~Ramakrishnan, M.~Tomamichel, and M.~Berta, ``Moderate deviation expansion
  for fully quantum tasks,''
  \href{http://dx.doi.org/10.1109/TIT.2023.3262145}{\emph{IEEE Transactions on
  Information Theory}}, \href{http://dx.doi.org/10.1109/TIT.2023.3262145}{pp.
  1--1}, \href{http://dx.doi.org/10.1109/TIT.2023.3262145}{2023}.

\bibitem{LY21b}
K.~Li and Y.~Yao, ``Reliable simulation of quantum channels,''
  \emph{arXiv:2112.04475 [quant-ph]}, 2021.

\bibitem{Cov72}
T.~Cover, ``Broadcast channels,''
  \href{http://dx.doi.org/10.1109/tit.1972.1054727}{\emph{{IEEE} Transactions
  on Information Theory}},
  \href{http://dx.doi.org/10.1109/tit.1972.1054727}{vol.~18},
  \href{http://dx.doi.org/10.1109/tit.1972.1054727}{no.~1},
  \href{http://dx.doi.org/10.1109/tit.1972.1054727}{pp. 2--14},
  \href{http://dx.doi.org/10.1109/tit.1972.1054727}{jan 1972}.

\bibitem{el2011network}
A.~El~Gamal and Y.-H. Kim, \emph{Network information theory}.\hskip 1em plus
  0.5em minus 0.4em\relax Cambridge university press, 2011.

\bibitem{YHD11}
J.~Yard, P.~Hayden, and I.~Devetak, ``Quantum broadcast channels,''
  \href{http://dx.doi.org/10.1109/tit.2011.2165811}{\emph{{IEEE} Transactions
  on Information Theory}},
  \href{http://dx.doi.org/10.1109/tit.2011.2165811}{vol.~57},
  \href{http://dx.doi.org/10.1109/tit.2011.2165811}{no.~10},
  \href{http://dx.doi.org/10.1109/tit.2011.2165811}{pp. 7147--7162},
  \href{http://dx.doi.org/10.1109/tit.2011.2165811}{Oct 2011}.

\bibitem{DHL10}
F.~Dupuis, P.~Hayden, and K.~Li, ``A father protocol for quantum broadcast
  channels,'' \href{http://dx.doi.org/10.1109/tit.2010.2046217}{\emph{{IEEE}
  Transactions on Information Theory}},
  \href{http://dx.doi.org/10.1109/tit.2010.2046217}{vol.~56},
  \href{http://dx.doi.org/10.1109/tit.2010.2046217}{no.~6},
  \href{http://dx.doi.org/10.1109/tit.2010.2046217}{pp. 2946--2956},
  \href{http://dx.doi.org/10.1109/tit.2010.2046217}{jun 2010}.

\bibitem{SW15}
I.~Savov and M.~M. Wilde, ``Classical codes for quantum broadcast channels,''
  \href{http://dx.doi.org/10.1109/tit.2015.2485998}{\emph{{IEEE} Transactions
  on Information Theory}},
  \href{http://dx.doi.org/10.1109/tit.2015.2485998}{vol.~61},
  \href{http://dx.doi.org/10.1109/tit.2015.2485998}{no.~12},
  \href{http://dx.doi.org/10.1109/tit.2015.2485998}{pp. 7017--7028},
  \href{http://dx.doi.org/10.1109/tit.2015.2485998}{Dec 2015}.

\bibitem{Cheng2021b}
H.-C. Cheng, N.~Dattaand, and C.~Rou\'ze, ``Strong converse bounds in quantum
  network information theory,''
  \href{http://dx.doi.org/10.1109/TIT.2021.3058166}{\emph{IEEE Transactions on
  Information Theory}},
  \href{http://dx.doi.org/10.1109/TIT.2021.3058166}{vol.~67},
  \href{http://dx.doi.org/10.1109/TIT.2021.3058166}{no.~4},
  \href{http://dx.doi.org/10.1109/TIT.2021.3058166}{April 2021}.

\bibitem{Li22}
C.~T. Li, ``First-order theory of probabilistic independence and single-letter
  characterizations of capacity regions,'' in \emph{2022 {IEEE} International
  Symposium on Information Theory ({ISIT})}.\hskip 1em plus 0.5em minus
  0.4em\relax {IEEE}, jun 2022.

\bibitem{Kor87}
J.~K{\"o}rner, ``The concept of single-letterization in information theory,''
  in \emph{Open Problems in Communication and Computation}.\hskip 1em plus
  0.5em minus 0.4em\relax Springer New York, 1987,
  \href{http://dx.doi.org/10.1007/978-1-4612-4808-8_5}{pp. 35--36}.

\bibitem{Mar79}
K.~Marton, ``A coding theorem for the discrete memoryless broadcast channel,''
  \href{http://dx.doi.org/10.1109/tit.1979.1056046}{\emph{{IEEE} Transactions
  on Information Theory}},
  \href{http://dx.doi.org/10.1109/tit.1979.1056046}{vol.~25},
  \href{http://dx.doi.org/10.1109/tit.1979.1056046}{no.~3},
  \href{http://dx.doi.org/10.1109/tit.1979.1056046}{pp. 306--311},
  \href{http://dx.doi.org/10.1109/tit.1979.1056046}{may 1979}.

\bibitem{Cov98}
T.~Cover, ``Comments on broadcast channels,''
  \href{http://dx.doi.org/10.1109/18.720547}{\emph{{IEEE} Transactions on
  Information Theory}}, \href{http://dx.doi.org/10.1109/18.720547}{vol.~44},
  \href{http://dx.doi.org/10.1109/18.720547}{no.~6},
  \href{http://dx.doi.org/10.1109/18.720547}{pp. 2524--2530},
  \href{http://dx.doi.org/10.1109/18.720547}{oct 1998}.

\bibitem{GvdM81}
A.~E. Gamal and E.~van~der Meulen, ``A proof of {Marton's} coding theorem for
  the discrete memoryless broadcast channel (corresp.),''
  \href{http://dx.doi.org/10.1109/tit.1979.1056046}{\emph{{IEEE} Transactions
  on Information Theory}},
  \href{http://dx.doi.org/10.1109/tit.1979.1056046}{vol.~27},
  \href{http://dx.doi.org/10.1109/tit.1979.1056046}{no.~1},
  \href{http://dx.doi.org/10.1109/tit.1979.1056046}{pp. 120--122},
  \href{http://dx.doi.org/10.1109/tit.1979.1056046}{Jan 1981}.

\bibitem{GP80}
\BIBentryALTinterwordspacing
M.~S.~P. S.~I.~Gel'fand, ``Capacity of a broadcast channel with one
  deterministic component,'' \emph{Problems Inform. Transmission}, vol.~16, pp.
  17--25, 1980. [Online]. Available: \url{https://zbmath.org/0458.94035}
\BIBentrySTDinterwordspacing

\bibitem{LK07}
Y.~Liang and G.~Kramer, ``Rate regions for relay broadcast channels,''
  \href{http://dx.doi.org/10.1109/tit.2007.904962}{\emph{{IEEE} Transactions on
  Information Theory}},
  \href{http://dx.doi.org/10.1109/tit.2007.904962}{vol.~53},
  \href{http://dx.doi.org/10.1109/tit.2007.904962}{no.~10},
  \href{http://dx.doi.org/10.1109/tit.2007.904962}{pp. 3517--3535},
  \href{http://dx.doi.org/10.1109/tit.2007.904962}{oct 2007}.

\bibitem{LKP08}
Y.~Liang, G.~Kramer, and H.~V. Poor, ``Equivalence of two inner bounds on the
  capacity region of the broadcast channel,'' in \emph{2008 46th Annual
  Allerton Conference on Communication, Control, and Computing}.\hskip 1em plus
  0.5em minus 0.4em\relax {IEEE}, sep 2008.

\bibitem{CK11}
I.~Csisz{\'a}r and J.~K{\"o}rner, \emph{Information Theory: Coding Theorems for
  Discrete Memoryless Systems}.\hskip 1em plus 0.5em minus 0.4em\relax
  Cambridge University Press ({CUP}), 2011.

\bibitem{Dut11}
\BIBentryALTinterwordspacing
N.~Dutil, ``Multiparty quantum protocols for assisted entanglement
  distillation,'' 2011, {P}h.{D}.~Thesis, McGill University, Montr{\'e}al.
  [Online]. Available: \url{https://arxiv.org/abs/1105.4657}
\BIBentrySTDinterwordspacing

\bibitem{DF13}
L.~Drescher and O.~Fawzi, ``On simultaneous min-entropy smoothing,'' in
  \emph{2013 {IEEE} International Symposium on Information Theory}.\hskip 1em
  plus 0.5em minus 0.4em\relax {IEEE}, jul 2013.

\bibitem{Kly06}
A.~A. Klyachko, ``Quantum marginal problem and {$N$}-representability,''
  \href{http://dx.doi.org/10.1088/1742-6596/36/1/014}{\emph{Journal of Physics:
  Conference Series}},
  \href{http://dx.doi.org/10.1088/1742-6596/36/1/014}{vol.~36},
  \href{http://dx.doi.org/10.1088/1742-6596/36/1/014}{pp. 72--86},
  \href{http://dx.doi.org/10.1088/1742-6596/36/1/014}{apr 2006}.

\bibitem{HD07}
\BIBentryALTinterwordspacing
P.~Hayden and F.~Dupuis. (2007) A reverse {Shannon} theorem for quantum
  broadcast channels. [Online]. Available:
  \url{https://www2.cms.math.ca/Events/summer07/abs/pdf/qit-ph.pdf}
\BIBentrySTDinterwordspacing

\bibitem{Ramakrishnan23}
N.~Ramakrishnan, ``Communication tasks in quantum information,'' 2023,
  {P}h.{D}.~Thesis, Department of Computing, Imperial College London.

\bibitem{AJW18}
A.~Anshu, R.~Jain, and N.~A. Warsi, ``A generalized quantum {Slepian--Wolf},''
  \href{http://dx.doi.org/10.1109/tit.2017.2786348}{\emph{{IEEE} Transactions
  on Information Theory}},
  \href{http://dx.doi.org/10.1109/tit.2017.2786348}{vol.~64},
  \href{http://dx.doi.org/10.1109/tit.2017.2786348}{no.~3},
  \href{http://dx.doi.org/10.1109/tit.2017.2786348}{pp. 1436--1453},
  \href{http://dx.doi.org/10.1109/tit.2017.2786348}{mar 2018}.

\bibitem{Gill54}
W.~McGill, ``Multivariate information transmission,''
  \href{http://dx.doi.org/10.1109/TIT.1954.1057469}{\emph{Transactions of the
  IRE Professional Group on Information Theory}},
  \href{http://dx.doi.org/10.1109/TIT.1954.1057469}{vol.~4},
  \href{http://dx.doi.org/10.1109/TIT.1954.1057469}{no.~4},
  \href{http://dx.doi.org/10.1109/TIT.1954.1057469}{pp. 93--111},
  \href{http://dx.doi.org/10.1109/TIT.1954.1057469}{1954}.

\bibitem{Watanabe60}
S.~Watanabe, ``Information theoretical analysis of multivariate correlation,''
  \href{http://dx.doi.org/10.1147/rd.41.0066}{\emph{IBM Journal of Research and
  Development}}, \href{http://dx.doi.org/10.1147/rd.41.0066}{vol.~4},
  \href{http://dx.doi.org/10.1147/rd.41.0066}{no.~1},
  \href{http://dx.doi.org/10.1147/rd.41.0066}{pp. 66--82},
  \href{http://dx.doi.org/10.1147/rd.41.0066}{1960}.

\bibitem{ChristKoenRennerPostSelect}
M.~Christandl, R.~K{\"o}nig, and R.~Renner, ``Postselection technique for
  quantum channels with applications to quantum cryptography,''
  \href{http://dx.doi.org/10.1103/physrevlett.102.020504}{\emph{Physical Review
  Letters}}, \href{http://dx.doi.org/10.1103/physrevlett.102.020504}{vol. 102},
  \href{http://dx.doi.org/10.1103/physrevlett.102.020504}{no.~2},
  \href{http://dx.doi.org/10.1103/physrevlett.102.020504}{jan 2009}.

\bibitem{anshu2017quantum}
A.~Anshu, V.~K. Devabathini, and R.~Jain, ``Quantum communication using
  coherent rejection sampling,''
  \href{http://dx.doi.org/10.1103/PhysRevLett.119.120506}{\emph{Physical Review
  Letters}}, \href{http://dx.doi.org/10.1103/PhysRevLett.119.120506}{vol. 119},
  \href{http://dx.doi.org/10.1103/PhysRevLett.119.120506}{no.~12},
  \href{http://dx.doi.org/10.1103/PhysRevLett.119.120506}{p. 120506},
  \href{http://dx.doi.org/10.1103/PhysRevLett.119.120506}{2017}.

\bibitem{anshu2017unified}
A.~Anshu, R.~Jain, and N.~A. Warsi, ``A unified approach to source and message
  compression,'' \emph{arXiv preprint arXiv:1707.03619}, 2017.

\bibitem{CG23}
\BIBentryALTinterwordspacing
H.-C. Cheng and L.~Gao, ``Tight one-shot analysis for convex splitting with
  applications in quantum information theory,'' 2023. [Online]. Available:
  \url{https://arxiv.org/abs/2304.12055}
\BIBentrySTDinterwordspacing

\bibitem{vonNeumann1951Various}
J.~Von~Neumann, ``Various techniques used in connection with random digits,''
  in \emph{Monte Carlo Method}, ser. National Bureau of Standards: Applied
  Mathematics Series, A.~S. Householder, G.~E. Forsythe, and H.~H. Germond,
  Eds.\hskip 1em plus 0.5em minus 0.4em\relax United States Government Printing
  Office, 1951, vol.~12, ch.~13, pp. 36--38.

\bibitem{robert1999monte}
C.~P. Robert, G.~Casella, and G.~Casella, \emph{Monte Carlo statistical
  methods}.\hskip 1em plus 0.5em minus 0.4em\relax Springer, 1999, vol.~2.

\bibitem{jain2003direct}
R.~Jain, J.~Radhakrishnan, and P.~Sen, ``A direct sum theorem in communication
  complexity via message compression,'' in \emph{International Colloquium on
  Automata, Languages, and Programming}.\hskip 1em plus 0.5em minus 0.4em\relax
  Springer, 2003, pp. 300--315.

\bibitem{MDS+13}
M.~M{\"u}ller-Lennert, F.~Dupuis, O.~Szehr, S.~Fehr, and M.~Tomamichel, ``On
  quantum {R{\'e}nyi} entropies: A new generalization and some properties,''
  \href{http://dx.doi.org/10.1063/1.4838856}{\emph{Journal of Mathematical
  Physics}}, \href{http://dx.doi.org/10.1063/1.4838856}{vol.~54},
  \href{http://dx.doi.org/10.1063/1.4838856}{no.~12},
  \href{http://dx.doi.org/10.1063/1.4838856}{p. 122203},
  \href{http://dx.doi.org/10.1063/1.4838856}{2013}.

\bibitem{WWY14}
M.~M. Wilde, A.~Winter, and D.~Yang, ``Strong converse for the classical
  capacity of entanglement-breaking and {Hadamard} channels via a sandwiched
  {R{\'{e}}nyi} relative entropy,''
  \href{http://dx.doi.org/10.1007/s00220-014-2122-x}{\emph{Communications in
  Mathematical Physics}},
  \href{http://dx.doi.org/10.1007/s00220-014-2122-x}{vol. 331},
  \href{http://dx.doi.org/10.1007/s00220-014-2122-x}{no.~2},
  \href{http://dx.doi.org/10.1007/s00220-014-2122-x}{pp. 593--622},
  \href{http://dx.doi.org/10.1007/s00220-014-2122-x}{Jul 2014}.

\bibitem{HT14}
M.~Hayashi and M.~Tomamichel, ``Correlation detection and an operational
  interpretation of the {R{\'{e}}nyi} mutual information,''
  \href{http://dx.doi.org/10.1063/1.4964755}{\emph{Journal of Mathematical
  Physics}}, \href{http://dx.doi.org/10.1063/1.4964755}{vol.~57},
  \href{http://dx.doi.org/10.1063/1.4964755}{no.~10},
  \href{http://dx.doi.org/10.1063/1.4964755}{p. 102201},
  \href{http://dx.doi.org/10.1063/1.4964755}{Oct 2016}.

\bibitem{Dup21}
\BIBentryALTinterwordspacing
F.~Dupuis, ``Privacy amplification and decoupling without smoothing,'' 2021.
  [Online]. Available: \url{https://arxiv.org/abs/2105.05342}
\BIBentrySTDinterwordspacing

\bibitem{SW23}
\BIBentryALTinterwordspacing
P.~{Colomer Saus} and A.~Winter, ``Decoupling by local random unitaries without
  simultaneous smoothing, and applications to multi-user quantum information
  tasks,'' 2023. [Online]. Available: \url{https://arxiv.org/abs/2304.12114}
\BIBentrySTDinterwordspacing

\bibitem{Proc465}
A.~Abeyesinghe, I.~Devetak, P.~Hayden, and A.~Winter, ``The mother of all
  protocols: restructuring quantum informations family tree,'' \emph{Proc. R.
  Soc. A, \textbf{465}(2108), 2537-2563}, 2009.

\bibitem{Kit97}
A.~Y. Kitaev, ``Quantum computations: algorithms and error correction,''
  \href{http://dx.doi.org/10.1070/rm1997v052n06abeh002155}{\emph{Russian
  Mathematical Surveys}},
  \href{http://dx.doi.org/10.1070/rm1997v052n06abeh002155}{vol.~52},
  \href{http://dx.doi.org/10.1070/rm1997v052n06abeh002155}{no.~6},
  \href{http://dx.doi.org/10.1070/rm1997v052n06abeh002155}{pp. 1191--1249},
  \href{http://dx.doi.org/10.1070/rm1997v052n06abeh002155}{dec 1997}.

\bibitem{Pau03}
V.~Paulsen, \emph{Completely Bounded Maps and Operator Algebras}.\hskip 1em
  plus 0.5em minus 0.4em\relax Cambridge University Press, feb 2003.

\bibitem{CTT2017}
C.~T. Chubb, V.~Y.~F. Tan, and M.~Tomamichel, ``Moderate deviation analysis for
  classical communication over quantum channels,''
  \href{http://dx.doi.org/10.1007/s00220-017-2971-1}{\emph{Communications in
  Mathematical Physics}},
  \href{http://dx.doi.org/10.1007/s00220-017-2971-1}{vol. 355},
  \href{http://dx.doi.org/10.1007/s00220-017-2971-1}{no.~3},
  \href{http://dx.doi.org/10.1007/s00220-017-2971-1}{pp. 1283--1315},
  \href{http://dx.doi.org/10.1007/s00220-017-2971-1}{Nov 2017}.

\bibitem{CH17}
H.-C. Cheng and M.-H. Hsieh, ``Moderate deviation analysis for
  classical-quantum channels and quantum hypothesis testing,''
  \href{http://dx.doi.org/10.1109/TIT.2017.2781254}{\emph{{IEEE} Transactions
  on Information Theory}},
  \href{http://dx.doi.org/10.1109/TIT.2017.2781254}{vol.~64},
  \href{http://dx.doi.org/10.1109/TIT.2017.2781254}{no.~2},
  \href{http://dx.doi.org/10.1109/TIT.2017.2781254}{pp. 1385--1403},
  \href{http://dx.doi.org/10.1109/TIT.2017.2781254}{feb 2018}.

\bibitem{YAG13}
M.~H. Yassaee, M.~R. Aref, and A.~Gohari, ``A technique for deriving one-shot
  achievability results in network information theory,'' in \emph{2013 {IEEE}
  International Symposium on Information Theory}.\hskip 1em plus 0.5em minus
  0.4em\relax {IEEE}, jul 2013.

\bibitem{LCV15}
J.~Liu, P.~Cuff, and S.~Verd{\'u}, ``One-shot mutual covering lemma and
  {Marton's} inner bound with a common message,'' in \emph{2015 {IEEE}
  International Symposium on Information Theory ({ISIT})}.\hskip 1em plus 0.5em
  minus 0.4em\relax {IEEE}, jun 2015.

\bibitem{Horodecki2005}
\BIBentryALTinterwordspacing
M.~Horodecki, J.~Oppenheim, and A.~Winter, ``Partial quantum information,''
  \href{http://dx.doi.org/10.1038/nature03909}{\emph{Nature}},
  \href{http://dx.doi.org/10.1038/nature03909}{vol. 436},
  \href{http://dx.doi.org/10.1038/nature03909}{no. 7051},
  \href{http://dx.doi.org/10.1038/nature03909}{pp. 673--676},
  \href{http://dx.doi.org/10.1038/nature03909}{Aug. 2005}. [Online]. Available:
  \url{https://doi.org/10.1038/nature03909}
\BIBentrySTDinterwordspacing

\bibitem{horodecki2007quantum}
------, ``Quantum state merging and negative information,''
  \emph{Communications in Mathematical Physics}, vol. 269, no.~1, pp. 107--136,
  2007.

\bibitem{DH10}
\BIBentryALTinterwordspacing
N.~Dutil and P.~Hayden, ``One-shot multiparty state merging,'' 2010. [Online].
  Available: \url{https://arxiv.org/abs/1011.1974}
\BIBentrySTDinterwordspacing

\bibitem{anshu2019minimax}
A.~Anshu, M.~Berta, R.~Jain, and M.~Tomamichel, ``A minimax approach to
  one-shot entropy inequalities,'' \emph{Journal of Mathematical Physics},
  vol.~60, no.~12, p. 122201, 2019.

\bibitem{AD89}
R.~Ahlswede and G.~Dueck, ``Identification via channels,''
  \href{http://dx.doi.org/10.1109/18.42172}{\emph{{IEEE} Transactions on
  Information Theory}}, \href{http://dx.doi.org/10.1109/18.42172}{vol.~35},
  \href{http://dx.doi.org/10.1109/18.42172}{no.~1},
  \href{http://dx.doi.org/10.1109/18.42172}{pp. 15--29},
  \href{http://dx.doi.org/10.1109/18.42172}{1989}.

\bibitem{AW02}
R.~Ahlswede and A.~Winter, ``Strong converse for identification via quantum
  channels,'' \href{http://dx.doi.org/10.1109/18.985947}{\emph{{IEEE}
  Transactions on Information Theory}},
  \href{http://dx.doi.org/10.1109/18.985947}{vol.~48},
  \href{http://dx.doi.org/10.1109/18.985947}{no.~3},
  \href{http://dx.doi.org/10.1109/18.985947}{pp. 569--579},
  \href{http://dx.doi.org/10.1109/18.985947}{mar 2002}.

\bibitem{Hay06}
M.~Hayashi, ``General nonasymptotic and asymptotic formulas in channel
  resolvability and identification capacity and their application to the
  wiretap channel,''
  \href{http://dx.doi.org/10.1109/tit.2006.871040}{\emph{{IEEE} Transactions on
  Information Theory}},
  \href{http://dx.doi.org/10.1109/tit.2006.871040}{vol.~52},
  \href{http://dx.doi.org/10.1109/tit.2006.871040}{no.~4},
  \href{http://dx.doi.org/10.1109/tit.2006.871040}{pp. 1562--1575},
  \href{http://dx.doi.org/10.1109/tit.2006.871040}{apr 2006}.

\bibitem{HM16}
M.~Hayashi and R.~Matsumoto, ``Secure multiplex coding with dependent and
  non-uniform multiple messages,''
  \href{http://dx.doi.org/10.1109/tit.2016.2530088}{\emph{{IEEE} Transactions
  on Information Theory}},
  \href{http://dx.doi.org/10.1109/tit.2016.2530088}{vol.~62},
  \href{http://dx.doi.org/10.1109/tit.2016.2530088}{no.~5},
  \href{http://dx.doi.org/10.1109/tit.2016.2530088}{pp. 2355--2409},
  \href{http://dx.doi.org/10.1109/tit.2016.2530088}{may 2016}.

\bibitem{YT19}
L.~Yu and V.~Y.~F. Tan, ``R{\'{e}}nyi resolvability and its applications to the
  wiretap channel,''
  \href{http://dx.doi.org/10.1109/tit.2018.2885318}{\emph{{IEEE} Transactions
  on Information Theory}},
  \href{http://dx.doi.org/10.1109/tit.2018.2885318}{vol.~65},
  \href{http://dx.doi.org/10.1109/tit.2018.2885318}{no.~3},
  \href{http://dx.doi.org/10.1109/tit.2018.2885318}{pp. 1862--1897},
  \href{http://dx.doi.org/10.1109/tit.2018.2885318}{mar 2019}.

\bibitem{SGC22b}
H.-C. Cheng and L.~Gao, ``Optimal second-order rates for quantum soft covering
  and privacy amplification,'' \emph{arXiv:2202.11590 [quant-ph]}, 2022.

\bibitem{AJW19a}
A.~Anshu, R.~Jain, and N.~A. Warsi, ``Building blocks for communication over
  noisy quantum networks,''
  \href{http://dx.doi.org/10.1109/TIT.2018.2851297}{\emph{IEEE Transactions on
  Information Theory}},
  \href{http://dx.doi.org/10.1109/TIT.2018.2851297}{vol.~65},
  \href{http://dx.doi.org/10.1109/TIT.2018.2851297}{no.~2},
  \href{http://dx.doi.org/10.1109/TIT.2018.2851297}{pp. 1287--1306},
  \href{http://dx.doi.org/10.1109/TIT.2018.2851297}{Feb. 2019}.

\bibitem{DGH+20}
D.~Ding, H.~Gharibyan, P.~Hayden, and M.~Walter, ``A quantum multiparty packing
  lemma and the relay channel,''
  \href{http://dx.doi.org/10.1109/tit.2019.2960500}{\emph{{IEEE} Transactions
  on Information Theory}},
  \href{http://dx.doi.org/10.1109/tit.2019.2960500}{vol.~66},
  \href{http://dx.doi.org/10.1109/tit.2019.2960500}{no.~6},
  \href{http://dx.doi.org/10.1109/tit.2019.2960500}{pp. 3500--3519},
  \href{http://dx.doi.org/10.1109/tit.2019.2960500}{jun 2020}.

\bibitem{Sen21}
P.~Sen, ``Unions, intersections and a one-shot quantum joint typicality
  lemma,''
  \href{http://dx.doi.org/10.1007/s12046-020-01555-3}{\emph{S{\={a}}dhan{\={a}}}},
  \href{http://dx.doi.org/10.1007/s12046-020-01555-3}{vol.~46},
  \href{http://dx.doi.org/10.1007/s12046-020-01555-3}{no.~1},
  \href{http://dx.doi.org/10.1007/s12046-020-01555-3}{mar 2021}.

\bibitem{Win99}
A.~Winter, ``Coding theorems of quantum information theory,''
  \emph{Ph.D.~Thesis, Universit{\"{a}}t Bielefeld), quant-ph/9907077}, 1999.

\bibitem{KW20}
S.~Khatri and M.~M. Wilde, ``Principles of quantum communication theory: A
  modern approach,'' \emph{arXiv:2011.04672 [quant-ph]}, 2020.

\bibitem{Tom16}
M.~Tomamichel, \emph{Quantum Information Processing with Finite
  Resources}.\hskip 1em plus 0.5em minus 0.4em\relax Springer International
  Publishing, 2016.

\bibitem{YCL21}
J.-K. You, H.-C. Cheng, and Y.-H. Li, ``Minimizing quantum {R{\'e}nyi}
  divergences via mirror descent with {Polyak} step size,''
  \emph{arXiv:2109.06054 [cs.IT]}, 2021.

\bibitem{Ume62}
H.~Umegaki, ``Conditional expectation in an operator algebra. {IV}. entropy and
  information,'' \href{http://dx.doi.org/10.2996/kmj/1138844604}{\emph{Kodai
  Mathematical Seminar Reports}},
  \href{http://dx.doi.org/10.2996/kmj/1138844604}{vol.~14},
  \href{http://dx.doi.org/10.2996/kmj/1138844604}{no.~2},
  \href{http://dx.doi.org/10.2996/kmj/1138844604}{pp. 59--85},
  \href{http://dx.doi.org/10.2996/kmj/1138844604}{1962}.

\bibitem{MO14}
M.~Mosonyi and T.~Ogawa, ``Quantum hypothesis testing and the operational
  interpretation of the quantum {R{\'{e}}nyi} relative entropies,''
  \href{http://dx.doi.org/10.1007/s00220-014-2248-x}{\emph{Communications in
  Mathematical Physics}},
  \href{http://dx.doi.org/10.1007/s00220-014-2248-x}{vol. 334},
  \href{http://dx.doi.org/10.1007/s00220-014-2248-x}{no.~3},
  \href{http://dx.doi.org/10.1007/s00220-014-2248-x}{pp. 1617--1648},
  \href{http://dx.doi.org/10.1007/s00220-014-2248-x}{Dec 2014}.

\bibitem{TH13}
M.~Tomamichel and M.~Hayashi, ``A {Hierarchy} of {Information} {Quantities} for
  {Finite} {Block} {Length} {Analysis} of {Quantum} {Tasks},''
  \href{http://dx.doi.org/10.1109/TIT.2013.2276628}{\emph{IEEE Transactions on
  Information Theory}},
  \href{http://dx.doi.org/10.1109/TIT.2013.2276628}{vol.~59},
  \href{http://dx.doi.org/10.1109/TIT.2013.2276628}{no.~11},
  \href{http://dx.doi.org/10.1109/TIT.2013.2276628}{pp. 7693--7710},
  \href{http://dx.doi.org/10.1109/TIT.2013.2276628}{Nov. 2013}, 00112 arXiv:
  1208.1478.

\bibitem{Li14}
K.~Li, ``Second-order asymptotics for quantum hypothesis testing,''
  \href{http://dx.doi.org/10.1214/13-aos1185}{\emph{The Annals of Statistics}},
  \href{http://dx.doi.org/10.1214/13-aos1185}{vol.~42},
  \href{http://dx.doi.org/10.1214/13-aos1185}{no.~1},
  \href{http://dx.doi.org/10.1214/13-aos1185}{pp. 171--189},
  \href{http://dx.doi.org/10.1214/13-aos1185}{Feb 2014}.

\bibitem{CGH18}
H.-C. Cheng, L.~Gao, and M.-H. Hsieh, ``Properties of noncommutative
  r{\'{e}}nyi and {Augustin} information,''
  \href{http://dx.doi.org/10.1007/s00220-022-04319-8}{\emph{Communications in
  Mathematical Physics}},
  \href{http://dx.doi.org/10.1007/s00220-022-04319-8}{feb 2022}.

\bibitem{Hao-Chung}
H.-C. Cheng, ``Error exponent analysis in quantum information theory,''
  \emph{PhD Thesis (University of Technology Sydney)}, 2018.

\bibitem{CHT19}
H.-C. Cheng, M.-H. Hsieh, and M.~Tomamichel, ``Quantum sphere-packing bounds
  with polynomial prefactors,''
  \href{http://dx.doi.org/10.1109/tit.2019.2891347}{\emph{{IEEE} Transactions
  on Information Theory}},
  \href{http://dx.doi.org/10.1109/tit.2019.2891347}{vol.~65},
  \href{http://dx.doi.org/10.1109/tit.2019.2891347}{no.~5},
  \href{http://dx.doi.org/10.1109/tit.2019.2891347}{pp. 2872--2898},
  \href{http://dx.doi.org/10.1109/tit.2019.2891347}{May 2019}.

\bibitem{CHDH2-2018}
H.-C. Cheng, E.~P. Hanson, N.~Datta, and M.-H. Hsieh, ``Duality between source
  coding with quantum side information and c-q channel coding,''
  \emph{arXiv:1809.11143 [quant-ph]}, 2018.

\bibitem{Roc70}
R.~T. Rockafellar, \emph{Convex Analysis}.\hskip 1em plus 0.5em minus
  0.4em\relax Walter de Gruyter {GmbH}, Jan 1970.

\bibitem{Kos84}
H.~Kosaki, ``Applications of the complex interpolation method to a {von
  Neumann} algebra: Non-commutative {$L_p$}-spaces,''
  \href{http://dx.doi.org/10.1016/0022-1236(84)90025-9}{\emph{Journal of
  Functional Analysis}},
  \href{http://dx.doi.org/10.1016/0022-1236(84)90025-9}{vol.~56},
  \href{http://dx.doi.org/10.1016/0022-1236(84)90025-9}{no.~1},
  \href{http://dx.doi.org/10.1016/0022-1236(84)90025-9}{pp. 29--78},
  \href{http://dx.doi.org/10.1016/0022-1236(84)90025-9}{mar 1984}.

\bibitem{JX03}
\BIBentryALTinterwordspacing
M.~Junge and Q.~Xu, ``Noncommutative {Burkholder/Rosenthal} inequalities,''
  \emph{The Annals of Probability}, vol.~31, no.~2, pp. 948--995, 2003.
  [Online]. Available: \url{http://www.jstor.org/stable/3481667}
\BIBentrySTDinterwordspacing

\bibitem{BL76}
J.~Bergh and J.~L{\"{o}}fstr{\"{o}}m, \emph{Interpolation Spaces}.\hskip 1em
  plus 0.5em minus 0.4em\relax Springer Berlin Heidelberg, 1976.

\bibitem{ADK+17}
A.~Anshu, V.~K. Devabathini, and R.~Jain, ``Quantum communication using
  coherent rejection sampling,''
  \href{http://dx.doi.org/10.1103/PhysRevLett.119.120506}{\emph{Phys. Rev.
  Lett.}}, \href{http://dx.doi.org/10.1103/PhysRevLett.119.120506}{vol. 119},
  \href{http://dx.doi.org/10.1103/PhysRevLett.119.120506}{p. 120506},
  \href{http://dx.doi.org/10.1103/PhysRevLett.119.120506}{Sep 2017}.

\bibitem{Win06}
A.~Winter, ``Identification via quantum channels in the presence of prior
  correlation and feedback,'' in \emph{Lecture Notes in Computer
  Science}.\hskip 1em plus 0.5em minus 0.4em\relax Springer Berlin Heidelberg,
  2006, \href{http://dx.doi.org/10.1007/11889342_27}{pp. 486--504}.

\bibitem{wilde2011classical}
M.~M. Wilde, ``From classical to quantum {S}hannon theory,'' \emph{arXiv
  preprint arXiv:1106.1445}, 2011.

\bibitem{Winter16}
\BIBentryALTinterwordspacing
A.~Winter, ``Tight uniform continuity bounds for quantum entropies: Conditional
  entropy, relative entropy distance and energy constraints,''
  \href{http://dx.doi.org/10.1007/s00220-016-2609-8}{\emph{Communications in
  Mathematical Physics}},
  \href{http://dx.doi.org/10.1007/s00220-016-2609-8}{vol. 347},
  \href{http://dx.doi.org/10.1007/s00220-016-2609-8}{no.~1},
  \href{http://dx.doi.org/10.1007/s00220-016-2609-8}{pp. 291--313},
  \href{http://dx.doi.org/10.1007/s00220-016-2609-8}{2016}. [Online].
  Available: \url{https://doi.org/10.1007/s00220-016-2609-8}
\BIBentrySTDinterwordspacing

\bibitem{HW16}
M.~Hayashi and S.~Watanabe, ``Uniform random number generation from {Markov}
  chains: Non-asymptotic and asymptotic analyses,''
  \href{http://dx.doi.org/10.1109/tit.2016.2530084}{\emph{{IEEE} Transactions
  on Information Theory}},
  \href{http://dx.doi.org/10.1109/tit.2016.2530084}{vol.~62},
  \href{http://dx.doi.org/10.1109/tit.2016.2530084}{no.~4},
  \href{http://dx.doi.org/10.1109/tit.2016.2530084}{pp. 1795--1822},
  \href{http://dx.doi.org/10.1109/tit.2016.2530084}{apr 2016}.

\bibitem{WH17}
S.~Watanabe and M.~Hayashi, ``Finite-length analysis on tail probability for
  markov chain and application to simple hypothesis testing,''
  \href{http://dx.doi.org/10.1214/16-aap1216}{\emph{The Annals of Applied
  Probability}}, \href{http://dx.doi.org/10.1214/16-aap1216}{vol.~27},
  \href{http://dx.doi.org/10.1214/16-aap1216}{no.~2},
  \href{http://dx.doi.org/10.1214/16-aap1216}{apr 2017}.

\bibitem{NC09}
M.~A. Nielsen and I.~L. Chuang, \emph{Quantum Computation and Quantum
  Information}.\hskip 1em plus 0.5em minus 0.4em\relax Cambridge University
  Press, 2009.

\bibitem{Uhl76}
A.~Uhlmann, ``The ``transition probability'' in the state space of a
  {$*$}-algebra,''
  \href{http://dx.doi.org/https://doi.org/10.1016/0034-4877(76)90060-4}{\emph{Reports
  on Mathematical Physics}},
  \href{http://dx.doi.org/https://doi.org/10.1016/0034-4877(76)90060-4}{vol.~9},
  \href{http://dx.doi.org/https://doi.org/10.1016/0034-4877(76)90060-4}{no.~2},
  \href{http://dx.doi.org/https://doi.org/10.1016/0034-4877(76)90060-4}{pp.
  273--279},
  \href{http://dx.doi.org/https://doi.org/10.1016/0034-4877(76)90060-4}{1976}.

\bibitem{FG99}
C.~Fuchs and J.~van~de Graaf, ``Cryptographic distinguishability measures for
  quantum-mechanical states,''
  \href{http://dx.doi.org/10.1109/18.761271}{\emph{{IEEE} Transactions on
  Information Theory}}, \href{http://dx.doi.org/10.1109/18.761271}{vol.~45},
  \href{http://dx.doi.org/10.1109/18.761271}{no.~4},
  \href{http://dx.doi.org/10.1109/18.761271}{pp. 1216--1227},
  \href{http://dx.doi.org/10.1109/18.761271}{may 1999}.

\bibitem{Bei13}
S.~Beigi, ``Sandwiched {R{\'e}nyi} divergence satisfies data processing
  inequality,'' \href{http://dx.doi.org/10.1063/1.4838855}{\emph{Journal of
  Mathematical Physics}}, \href{http://dx.doi.org/10.1063/1.4838855}{vol.~54},
  \href{http://dx.doi.org/10.1063/1.4838855}{no.~12},
  \href{http://dx.doi.org/10.1063/1.4838855}{p. 122202},
  \href{http://dx.doi.org/10.1063/1.4838855}{2013}.

\bibitem{FL13}
R.~L. Frank and E.~H. Lieb, ``Monotonicity of a relative {R{\'e}nyi} entropy,''
  \href{http://dx.doi.org/10.1063/1.4838835}{\emph{Journal of Mathematical
  Physics}}, \href{http://dx.doi.org/10.1063/1.4838835}{vol.~54},
  \href{http://dx.doi.org/10.1063/1.4838835}{no.~12},
  \href{http://dx.doi.org/10.1063/1.4838835}{p. 122201},
  \href{http://dx.doi.org/10.1063/1.4838835}{2013}.

\bibitem{MO17}
M.~Mosonyi and T.~Ogawa, ``Strong converse exponent for classical-quantum
  channel coding,''
  \href{http://dx.doi.org/10.1007/s00220-017-2928-4}{\emph{Communications in
  Mathematical Physics}},
  \href{http://dx.doi.org/10.1007/s00220-017-2928-4}{vol. 355},
  \href{http://dx.doi.org/10.1007/s00220-017-2928-4}{no.~1},
  \href{http://dx.doi.org/10.1007/s00220-017-2928-4}{pp. 373--426},
  \href{http://dx.doi.org/10.1007/s00220-017-2928-4}{Oct 2017}.

\bibitem{Pet86}
D.~Petz, ``Quasi-entropies for finite quantum systems,''
  \href{http://dx.doi.org/10.1016/0034-4877(86)90067-4}{\emph{Reports on
  Mathematical Physics}},
  \href{http://dx.doi.org/10.1016/0034-4877(86)90067-4}{vol.~23},
  \href{http://dx.doi.org/10.1016/0034-4877(86)90067-4}{no.~1},
  \href{http://dx.doi.org/10.1016/0034-4877(86)90067-4}{pp. 57--65},
  \href{http://dx.doi.org/10.1016/0034-4877(86)90067-4}{Feb 1986}.

\bibitem{SW12}
N.~Sharma and N.~A. Warsi, ``Fundamental bound on the reliability of quantum
  information transmission,''
  \href{http://dx.doi.org/10.1103/physrevlett.110.080501}{\emph{Physical Review
  Letters}}, \href{http://dx.doi.org/10.1103/physrevlett.110.080501}{vol. 110},
  \href{http://dx.doi.org/10.1103/physrevlett.110.080501}{no.~8},
  \href{http://dx.doi.org/10.1103/physrevlett.110.080501}{Feb 2013}.

\bibitem{LT76}
E.~H. Lieb and W.~E. Thirring, ``Inequalities for the moments of the
  eigenvalues of the {Schrodinger} {Hamiltonian} and their relation to
  {Sobolev} inequalities.''\hskip 1em plus 0.5em minus 0.4em\relax Walter de
  Gruyter {GmbH}.

\bibitem{LT15}
S.~M. Lin and M.~Tomamichel, ``Investigating properties of a family of quantum
  r{\'{e}}nyi divergences,''
  \href{http://dx.doi.org/10.1007/s11128-015-0935-y}{\emph{Quantum Information
  Processing}}, \href{http://dx.doi.org/10.1007/s11128-015-0935-y}{vol.~14},
  \href{http://dx.doi.org/10.1007/s11128-015-0935-y}{no.~4},
  \href{http://dx.doi.org/10.1007/s11128-015-0935-y}{pp. 1501--1512},
  \href{http://dx.doi.org/10.1007/s11128-015-0935-y}{Feb 2015}.

\bibitem{TV15}
M.~Tomamichel and V.~Y.~F. Tan, ``Second-order asymptotics for the classical
  capacity of image-additive quantum channels,''
  \href{http://dx.doi.org/10.1007/s00220-015-2382-0}{\emph{Communications in
  Mathematical Physics}},
  \href{http://dx.doi.org/10.1007/s00220-015-2382-0}{vol. 338},
  \href{http://dx.doi.org/10.1007/s00220-015-2382-0}{no.~1},
  \href{http://dx.doi.org/10.1007/s00220-015-2382-0}{pp. 103--137},
  \href{http://dx.doi.org/10.1007/s00220-015-2382-0}{May 2015}.

\bibitem{CHT17}
H.-C. Cheng, M.-H. Hsieh, and M.~Tomamichel, ``Quantum sphere-packing bounds
  with polynomial prefactors,''
  \href{http://dx.doi.org/10.1109/TIT.2019.2891347}{\emph{arXiv:1704.05703}},
  \href{http://dx.doi.org/10.1109/TIT.2019.2891347}{2017}.

\bibitem{Hay09}
M.~Hayashi, ``Universal coding for classical-quantum channel,''
  \href{http://dx.doi.org/10.1007/s00220-009-0825-1}{\emph{Communications in
  Mathematical Physics}},
  \href{http://dx.doi.org/10.1007/s00220-009-0825-1}{vol. 289},
  \href{http://dx.doi.org/10.1007/s00220-009-0825-1}{no.~3},
  \href{http://dx.doi.org/10.1007/s00220-009-0825-1}{pp. 1087--1098},
  \href{http://dx.doi.org/10.1007/s00220-009-0825-1}{May 2009}.

\bibitem{Gruber93}
P.~M. Gruber and J.~M. Wills, \emph{Handbook of Convex Geometry, Vol. A}.\hskip
  1em plus 0.5em minus 0.4em\relax Elsevier Science Publishers, 1993.

\end{thebibliography}
}

\end{document}